\documentclass[12pt, oneside]{amsart}   	
\usepackage{geometry}                		                		 
\usepackage{graphicx, array, makecell}
\usepackage{amssymb, color}
 \usepackage{hyperref}       
\allowdisplaybreaks[1]
\theoremstyle{plain}
\newtheorem{theorem}{Theorem}
\newtheorem{prop}{Proposition}
\newtheorem{lemma}{Lemma}
\newtheorem{conjecture}{Conjecture}
\newtheorem{cor}{Corollary}
\theoremstyle{definition}
\newtheorem{defi}{Definition}

\newtheorem{remark}{Remark}
\newcommand{\M}{\mathcal{M}}
\newcommand{\F}{\mathcal{F}}
\newcommand{\R}{\mathcal{R}}
\newcommand{\G}{\mathcal{G}}

\newcommand{\EE}{\mathbb{E}}
\newcommand{\CC}{\mathbb{C}}
\newcommand{\QQ}{\mathbb{Q}}
\newcommand{\ZZ}{\mathbb{Z}}
\newcommand{\e}{\epsilon}
\newcommand{\nn}{\nonumber}
\newcommand{\p}{\partial}

\newcommand{\W}{\mathcal{W}}
\newcommand{\beq}{\begin{equation}}
\newcommand{\eeq}{\end{equation}}

\def\bt{\mathbf t}
\def\bT{\mathbf T}

\def\={\,=\,}
\def\+{\,+\,}

\def\G{\mathcal{G}}

\def\P{\mathcal{P}}
\def\W{\mathcal{W}}
\def\thin{\hskip 1 pt}

\newcolumntype{?}{!{\vrule width 1pt}}

\title{Mapping partition functions}
\author{Di Yang}
\author{Don Zagier}
\date{}
\begin{document}
\begin{abstract}
We introduce an infinite group action on
partition functions of WK type, meaning of the type of the partition function 
$Z^{\scriptscriptstyle \rm WK}$ in the famous result of Witten and Kontsevich 
 expressing the partition function of $\psi$-class integrals on the compactified moduli space $\overline{\M}_{g,n}$ 
as a $\tau$-function for the Korteweg--de~Vries hierarchy. 
Specifically, the group which acts 
 is the group $\G$ of formal power series of one variable $\varphi(V)=V+O(V^2)$,
with group law given by composition, acting in a suitable way on the infinite tuple of variables 
of the partition functions. 
In particular, any \hbox{$\varphi \in \G$} sends the Witten--Kontsevich (WK) partition function $Z^{\scriptscriptstyle \rm WK}$ to a new partition function
$Z^\varphi$, which we call the {\it WK mapping partition function associated to~$\varphi$}.   
We show that the genus zero part of $\log Z^\varphi$ is independent of~$\varphi$
and give an explicit recursive description for its higher genus parts (loop equation), 
and as applications of this obtain relationships of the $\psi$-class integrals 
to Gaussian Unitary Ensemble and generalized Br\'ezin--Gross--Witten correlators.
In a different direction, we use $Z^\varphi$ to construct a new integrable hierarchy, which we call 
 the {\it WK mapping hierarchy associated to~$\varphi$}. We show that this hierarchy is a 
 bihamiltonian perturbation of the Riemann--Hopf hierarchy possessing a $\tau$-structure, 
 and prove that it is a universal object for all such perturbations.
Similarly, for any \hbox{$\varphi\in\G$}, 
we define the {\it Hodge mapping partition function associated to~$\varphi$}, prove that it is integrable, and 
study its role in hamiltonian perturbations of the Riemann--Hopf hierarchy possessing a $\tau$-structure. 
Finally, we establish a {\it generalized Hodge--WK correspondence} relating different Hodge mapping partition functions.
\end{abstract}

\maketitle

\setcounter{tocdepth}{1}
\tableofcontents

\section{Introduction}
{\it The Korteweg--de Vries (KdV) equation} 
\beq\label{kdvequation1002}
\frac{\p u}{\p t} \= u \, \frac{\p u}{\p x} \+ \frac{\e^2}{12} \, \frac{\p^3 u}{\p x^3}
\eeq
was discovered in the study of shallow water waves in the 19th century~\cite{Bouss, KdV}. 
It was shown~\cite{Lax, Miura68, MGK} in the 1960s that this equation can be  
extended to a family of pairwise commuting evolutionary PDEs, called the {\it KdV hierarchy}:
\beq\label{kdvflows26}
\frac{\p u}{\p t_i} \= \frac{u^i}{i!} \, \frac{\p u}{\p x} 
\+ \e^2 \, K_i\biggl(u, \frac{\p u}{\p x}, \frac{\p^2 u}{\p x^2}, \dots, \frac{\p^{2i+1} u}{\p x^{2i+1}}, \e \biggr)\,, \quad i\ge0\,.
\eeq
Here $t_0=x$, $t_1=t$, and $K_i$, $i\ge0$, are certain polynomials.
For more about the KdV hierarchy 
see e.g.~\cite{DYZ21, DZ-norm, Sato, Wi91}. 
The {\it Riemann--Hopf (RH) hierarchy}, aka the dispersionless KdV hierarchy, is  
defined again as~\eqref{kdvflows26} but with $\e$ taken to be~0.

In 1990, Witten~\cite{Wi91} made a famous conjecture: the partition function $Z^{\scriptscriptstyle \rm WK}(\bt;\e)$ of $\psi$-class 
integrals on the Deligne--Mumford moduli space of algebraic curves~\cite{DM69}
\beq\label{defZ}
Z^{\scriptscriptstyle\rm WK}(\bt;\e) \= \exp\biggl(\,\sum_{g, \, n\ge0} \e^{2g-2}\sum_{i_1,\dots,i_n\ge0} \frac{t_{i_1} \cdots t_{i_n}}{n!}\int_{\overline{\M}_{g,n}} \, \psi_1^{i_1}\cdots\psi_n^{i_n}\biggr) \,, 
\eeq
 is a $\tau$-function for the KdV hierarchy, and in particular, 
\beq\label{defuwk114}
u^{\scriptscriptstyle \rm WK}(\bt;\e):=\e^2 \p_{t_0}^2 (\log Z^{\scriptscriptstyle \rm WK}(\bt;\e))
\eeq 
 satisfies the KdV hierarchy~\eqref{kdvflows26}.
Here $\bt=(t_0,t_1,t_2,\dots)$ is an 
infinite tuple of indeterminates, $\overline{\M}_{g,n}$ denotes the moduli space 
of stable algebraic curves of genus~$g$ with~$n$ distinct marked points, and  
$\psi_a$ ($a=1,\dots,n$) denotes the first Chern class of the 
 $a$th tautological line bundle on~$\overline{\M}_{g,n}$. 
Note that
the integral appearing in the right-hand side of~\eqref{defZ} vanishes unless the degree-dimension matching condition  
\beq\label{dd1003}
i_1 \+ \cdots \+ i_n \= 3g \,-\, 3 \+ n 
\eeq
is satisfied.
Witten's conjecture, that opens the studies of the deep relations between topology of~$\overline{\M}_{g,n}$ 
and integrable systems, was first proved by Kontsevich~\cite{Ko92} and is now known as 
the {\it Witten--Kontsevich theorem}. 
See \cite{AIS, CLL08, KL07, KL05, Mir07, OP09} for several other proofs of this theorem. 
The function $Z^{\scriptscriptstyle \rm WK}(\bt;\e)$ is 
 referred to indifferently as the {\it WK (Witten--Kontsevich) partition function} 
or as the {\it WK tau-function}.

The general notion behind the story,  which has appeared in combinatorics, statistical physics, 
matrix models, and other places, is that many interesting partition functions are $\tau$-functions of 
integrable systems. 
On one hand, there are axiomatic or constructive ways approaching topologically 
  interesting numbers~\cite{DZ-norm, Gi01, KM}, the 
 partition functions of which would correspond to some 
 integrable systems. 
 In particular, Dubrovin and Zhang~\cite{DZ-norm} gave a constructive way of defining 
a hierarchy of evolutionary PDEs in $(1+1)$ 
dimensions\footnote{For readers not familiar with some of the terminology, we refer to 
Section~\ref{newnewsection7} for a brief review.}  
associated to essentially any partition function.
For instance, their construction applied to the WK partition  function gives the KdV hierarchy. 
The Dubrovin--Zhang hierarchy corresponding to the partition function of Hodge integrals
on~$\overline{\M}_{g,n}$ depending on an infinite family of  parameters~\cite{DLYZ16} 
will also play an important role in this paper and will be called simply the {\it Hodge hierarchy}. 
On the other hand, one is interested in finding certain  integrable systems that admit 
$\tau$-functions, sometimes called possessing a $\tau$-structure\footnote{In literature (see e.g.~\cite{DLYZ16, DZ-norm}), 
 ``(bi-)hamiltonian $\tau$-structure" is often specialized to $\tau$-symmetry, but ``$\tau$-structure" in this paper has the 
broader meaning (for details see Section~\ref{newnewsection8}; cf.~\cite{DYZ21}, \cite{VY}).}, which axiomatically 
 leads to certain classification invariants~\cite{DLYZ16, DZ-norm, LYZZ22}. 
 The deep relations between these two notions is revealed most beautifully 
 when there is a one-to-one correspondence between them, an example
being  the {\it Hodge universality conjecture} in the study of  
Hodge integrals (rank~1 cohomological field theories) and 
$\tau$-symmetric integrable hierarchies of Hamiltonian 
evolutionary PDEs~\cite{DLYZ16}, which says that the Hodge hierarchy is a universal object for
one-component $\tau$-symmetric integrable Hamiltonian perturbations of the 
RH hierarchy,\footnote{Here ``perturbation of the RH hierarchy"
means a hierarchy of evolutionary PDEs whose right-hand sides differ from those of 
the RH hierarchy by terms with more than one spatial derivative.} i.e., conjecturally any such integrable hierarchy 
is equivalent to the Hodge hierarchy. 

In this paper we will study  
the deep relations from a novel perspective that sheds new light on both sides:

\vspace{1.5mm} 

\noindent {\bf (a)} We introduce an infinite group action, different from those of Givental or 
Sato--Segal--Wilson, on the arguments (infinite tuples) of partition functions. 
The group which acts is the group~$\G$ of power series of one variable $\varphi(V)=V+O(V^2)$, acting on the right
on the infinite tuple (denoted $\bt \mapsto \bt.\varphi$ and defined in equation~\eqref{defsigmaction} below).
In particular, if we start with the WK partition function then each element $\varphi$ of the group 
defines a new partition function $Z^\varphi(\bt;\e) := Z^{\scriptscriptstyle\rm WK}(\bt.\varphi^{-1};\e)$, 
which we will call the {\it WK mapping partition function associated to~$\varphi$}.
The coefficients of its logarithm provide new and potentially interesting numbers, although we 
do not know their topological meaning.  We show that the genus-zero part of this logarithm is 
independent of~$\varphi$,  
give the dilaton equation and Virasoro constraints, and derive loop equations determining also the higher genus parts.
Now applying the Dubrovin--Zhang construction to the WK mapping partition function we obtain a new
hierarchy which we will call  the {\it WK mapping hierarchy associated to~$\varphi$}. We show that this hierarchy
can be obtained  by a space-time exchange combined with a Miura-type transformation on the KdV hierarchy,
and then by using a recent result given by S.-Q.~Liu, Z.~Wang and Y.~Zhang~\cite{LWangZ}, 
prove the following theorem in Section~\ref{section6}: 
\begin{theorem}\label{thmmainshortversion}
The WK mapping hierarchy
is a bihamiltonian perturbation of the RH hierarchy possessing a $\tau$-structure. 
\end{theorem} 

\noindent {\bf (b)} 
We study the classification of bihamiltonian perturbations of the RH hierarchy possessing 
a $\tau$-structure under the Miura-type group action. In \cite{DLYZ16} a related but different 
classification work was studied and it was conjectured that the universal object 
for the $\tau$-symmetric integrable hierarchies of bihamiltonian evolutionary PDEs
is the Volterra lattice hierarchy. 
Here, however, we consider a larger class by allowing a weaker form of the 
$\tau$-symmetry condition used in~\cite{DLYZ16, DZ-norm}. 
It turns out that there is a rich family of such bihamiltonian perturbations, part of which can be seen from 
Theorem~\ref{thmmainshortversion}, and we propose and prove the {\it WK mapping universality 
theorem}: the WK mapping hierarchy is a universal object in one-component bihamiltonian perturbations 
of the RH hierarchy possessing a $\tau$-structure. This theorem has a precise numerical meaning and 
we also give verifications of it to high orders in Section~\ref{section6}. 

\vspace{1.5mm}

Similarly, we consider the $\G$-action on the Hodge partition function.
The resulting power series will be called the {\it Hodge mapping partition function}, and 
the Dubrovin--Zhang hierarchy for the Hodge mapping partition function will be called the {\it Hodge mapping hierarchy}.
\begin{theorem}\label{thmmain2}
The Hodge mapping hierarchy
is an integrable perturbation of the RH hierarchy possessing a $\tau$-structure. 
\end{theorem}
\noindent 
The proof of a refined version of this theorem is given in Section~\ref{section8}.
We expect that this integrable hierarchy is hamiltonian.
Note that our proof for the integrability also works for the WK mapping hierarchy, and also that  
Theorem~\ref{thmmain2} generalizes part of the result in Theorem~\ref{thmmainshortversion}.
 We will also propose in Section~\ref{section8} (see Conjecture~\ref{Hmuconj226} there)
 the {\it Hodge mapping universality conjecture}: the Hodge mapping hierarchy
is a universal object in one-component hamiltonian perturbations of the RH hierarchy possessing a $\tau$-structure 
(weakening again the $\tau$-symmetry condition from~\cite{DLYZ16, DZ-norm}). This conjecture 
generalizes the Hodge universality conjecture~\cite{DLYZ16}.

\smallskip

For the special case when the group element $\varphi$ is taken to be 
\beq\label{varphispecialdefi}
\varphi_{\rm special}(V) \, :=  \, \frac{e^{2qV}-1}{2q}\,,
\eeq 
by using the loop equation we will
prove in Section~\ref{sectionexample223} 
the {\it Hodge--WK correspondence} described in the following theorem, which is a relationship between 
a certain special-Hodge 
partition function $Z_{\Omega^{\rm special}(q)}(\bt;q)$ (see \eqref{hodgepar1111} and~\eqref{Omegaspecial81} in Section~\ref{sectionexample223} for the definition) and 
the WK partition function $Z^{\scriptscriptstyle \rm WK}(\bt;\e)$.
\begin{theorem}\label{thmhwk} 
The following identity holds in $\CC((\e^2))[[q]][[\bt]]:$ 
\beq\label{mainidentity}
Z_{\Omega^{\rm special}(q)}(\bt . \varphi_{\rm special};\e) \= Z^{\scriptscriptstyle \rm WK}(\bt; \e)\,.
\eeq
\end{theorem}
We note that, although not completely obvious, this theorem is equivalent to a result of Alexandrov~\cite{Al212}; 
see Section~\ref{sectionexample223} for more details.

As an application of the Hodge--WK correspondence, 
we will establish in the following two theorems 
explicit relationships of the WK partition function 
$Z^{\scriptscriptstyle \rm WK}(\bt; \epsilon)$ to the modified GUE partition function
 $Z^{\scriptscriptstyle \rm meGUE}(x, {\bf s}; \e)$ and to 
the generalized BGW partition function $Z^{\scriptscriptstyle\rm cBGW}(x,{\bf r};\e)$
(see~\cite{DLYZ20, DY17-2} or Section~\ref{sectionexample223} for 
the definition of $Z^{\scriptscriptstyle \rm meGUE}(x, {\bf s}; \e)$ and 
see~\cite{YZ21} (cf.~\cite{Al18, BG, GW, MMS96, YZ21}) for the definition of $Z^{\scriptscriptstyle\rm cBGW}(x,{\bf r};\e)$). 
Here and below, ``GUE" refers to 
 Gaussian Unitary Ensemble, and ``BGW" refers to Br\'ezin--Gross--Witten.
\begin{theorem}\label{WKGUEthm81}
The following identity holds true in $\CC((\e^2))[[x-1]][[{\bf s}]]:$
\beq\label{wkgue1106}
Z^{\scriptscriptstyle \rm WK}\bigl({\bf t}^{\scriptscriptstyle \rm WK-GUE}(x,{\bf s}); \epsilon\bigr) \, e^{\frac{A(x,{\bf s})}{\e^2}}
\= Z^{\scriptscriptstyle \rm meGUE}\Bigl(x, {\bf s}; \frac{\e}{\sqrt{2}}\Bigr) \,,
\eeq
where $A(x,{\bf s})$ is a quadratic series defined by
\begin{align}
A(x,{\bf s}) \=  & \frac12 \sum_{j_1,j_2\ge1} \frac{j_1j_2}{j_1+j_2} \binom{2j_1}{j_1} \binom{2j_2}{j_2} 
\Bigl(s_{j_1}-\frac{\delta_{j_1,1}}2\Bigr) \Bigl(s_{j_2} - \frac{\delta_{j_2,1}}2\Bigr) \label{aseries1107}\\
& +  \, x \sum_{j\geq1} \binom{2j}j \Bigl(s_{j}-\frac{\delta_{j,1}}2\Bigr) \,, \nn
\end{align}
 and 
\begin{align}
&\frac{2^{m}}{(2m+1)!!} \, t_m^{\scriptscriptstyle \rm WK-GUE}(x,{\bf s}) \\
& \; \= \frac23 \, \delta_{m,1} 
\+ \frac{1}{2m+1} \, x 
\+ \sum_{j\geq 1} \, \binom{m+j-1/2}{j-1} \, 2^{2j-1} \, \Bigl(s_{j}-\frac{\delta_{j,1}}2\Bigr) \,, \quad m\ge0\,. \nn
\end{align}
\end{theorem}
We call~\eqref{wkgue1106} the {\it WK--GUE correspondence}.
\begin{theorem}\label{thmwkbgw}
The following identity holds true in $\QQ((\e^2))[[x+2]][[{\bf r}]]:$
\beq\label{wkbgw1106}
Z^{\scriptscriptstyle \rm WK}\bigl({\bf t}^{\scriptscriptstyle \rm WK-BGW}(x,{\bf r}); \sqrt{-4}\epsilon\bigr) \, 
e^{\frac{A_{\scriptscriptstyle \rm cBGW} (x,{\bf r})}{\e^2}}
\=  Z^{\scriptscriptstyle\rm cBGW}(x,{\bf r};\e) \,,
\eeq
where $A_{\scriptscriptstyle \rm cBGW} (x,{\bf r})$ is a quadratic function given by
\beq\label{Abgw1107}
A_{\scriptscriptstyle \rm cBGW} (x,{\bf r}) \= \frac12 \sum_{a,b\ge0} \frac{(r_a-\delta_{a,0})(r_b-\delta_{b,0})}{a! \, b! \, (a+b+1)} 
- x\, \sum_{b\ge0} \frac{r_b-\delta_{b,0}}{b! \, (2b+1)} \,, 
\eeq
 and 
\begin{align}
&t_m^{\scriptscriptstyle \rm WK-BGW}(x,{\bf r}) \= \delta_{m,1} \+ 2 \, \delta_{m,0} \+ \frac{(2m-1)!!}{2^m} \, x -2 
\sum_{j \geq m} \, \frac{(-1)^m}{(j-m)!}\, r_j \,. 
\end{align}
\end{theorem}
We call~\eqref{wkbgw1106} the {\it WK--BGW correspondence}. 

Using the Hodge--WK correspondence and the $\G$-action 
we establish in Theorem~\ref{thm1114} an explicit relationship between 
the Hodge mapping partition function with a special choice of 
its parameters associated to an arbitrarily given group element $\psi\in\G$  
(which will be called the {\it special-Hodge mapping partition function associated to~$\psi$}) 
 and the WK mapping partition function associated to~$\varphi$, where $\varphi$ and $\psi$ are related by
 $\varphi = \varphi_{\rm special} \circ \psi$ with $\varphi_{\rm special}$ as in~\eqref{varphispecialdefi}, i.e.,
\beq\label{phipsi1122}
\varphi(V) \=  \frac{e^{2 \, q \, \psi(V)}-1}{2 \, q}\,, \quad \; \psi(V) \= \frac{\log (1+ 2 \, q \, \varphi(V))}{2 \, q} \,.
\eeq
Such a relationship will be called the {\it generalized Hodge--WK correspondence}.

\medskip

\noindent {\bf Organization of the paper.} 
In Section~\ref{section2} we introduce the infinite group action on infinite tuples,  
 define the WK mapping partition function, and prove Theorem~\ref{thmgenus0}: the genus zero part is a fixed point of the group action.  
In Section~\ref{section3} we give the dilaton equation and Virasoro constraints for the WK mapping partition function. 
In Section~\ref{section4} we give a geometric proof of Theorem~\ref{thmgenus0}.
In Section~\ref{section5} we prove the existence of the jet-variable representation for the higher genus WK mapping free energies,
and in Section~\ref{newnewsection6} we derive the loop equation. 
In Sections~\ref{newnewsection7} and~\ref{newnewsection8} we study the classification of 
hamiltonian and bihamiltonian perturbations of the RH hierarchy possessing a $\tau$-structure. 
In Section~\ref{section6} we 
prove Theorem~\ref{thmmainshortversion} and prove the WK mapping universality theorem. 
A particular example is discussed in Section~\ref{sectionexample223}, where we prove 
Theorems \ref{thmhwk}, \ref{WKGUEthm81}, \ref{thmwkbgw}.
In Section~\ref{section8} we prove Theorem~\ref{thmmain2} and 
propose the Hodge mapping universality conjecture. 
Section~\ref{section9} is devoted to generalizations.

\smallskip 

\noindent {\bf Acknowledgements}. One of the authors D.Y. is grateful to Youjin Zhang and Boris Dubrovin 
for their advice and teaching over many 
years and for specific suggestions that were important for this paper.
Part of the work of D.Y. was done during his visit 
in Max Planck Institute for Mathematics; he thanks MPIM for excellent working conditions and financial support. 
The work was partially supported by NSFC No.~12371254,  the CAS No.~YSBR-032, 
and by National Key R and D Program of China 2020YFA0713100.

\section{$\G$-action and the definition of the WK mapping partition function}\label{section2}
In this section we define an infinite group action on infinite tuples and the WK mapping 
partition function, and prove Theorem~\ref{thmgenus0} below.

Fix a ground ring $R$ (which for us will always be a $\QQ$-algebra, usually $\QQ$ or $\CC$ or~$\QQ[q]\thin$)
and let $\G=V+V^2R[[V]]$ be the group of invertible power series of one variable with leading coefficient~1, 
with the group law given by composition and denoted by~$\circ$. 
We define an affine-linear right action of the group~$\G$ on tuples~$\bt$ by 
\beq
\bt \= (t_0,t_1,t_2,\dots) \quad  \mapsto \quad \bt.\varphi \= \bT \= (T_0, T_1,T_2,\dots) \,,
\eeq
where $\bt$ and $\bT$ are related by 
\beq \label{defsigmaction}
B_{\bT}(V) \= \sqrt{\varphi'(V)} \,  B_\bt(\varphi(V)) 
\eeq
with $B_\bt$ defined for any infinite tuple~$\bt$ by 
\beq\label{Bvt1001}
B_{\bt}(v) \,:=\, v \,-\, \sum_{i\ge0} \, \frac{t_i}{i!} \, v^i \,.
\eeq
Explicitly, if we write $\varphi(V)=\sum_{k=0}^\infty a_k V^k$ with $a_0=0$, $a_1=1$, then 
\begin{align}
& T_0= t_0\,, \quad T_1=t_1 \+ a_2 \, t_0\,, \quad T_2 = t_2 \+ 4 \, a_2 \, (t_1-1) \+ \bigl(3a_3-a_2^2\bigr) \, t_0\,,  \quad \dots\,,\\ 
& t_0 = T_0\,, \quad t_1= T_1 - a_2 \, T_0\,, \quad t_2 = T_2 - 4  \, a_2 \, (T_1-1) - \bigl(3a_3-5a_2^2\bigr) \, T_0\,, \quad \dots\,.
\end{align}
Note that if we introduce 
for any tuple~$\bf t$  
the 1-form
\beq
\omega_\bt(v) \,:=\, 
B_\bt(v)^2 \, dv \,,
\eeq
then the defining equation~\eqref{defsigmaction} for the $\G$-action can be stated equivalently as
\beq
\omega_\bT(V) \= \omega_\bt(\varphi(V)) \,.  
\eeq 

Let $E(\bt)$ denote the following power series 
\beq\label{defVandvf1001}
E(\bt) \= \sum_{n\ge1} \, \frac1{n} \, \sum_{i_1,\dots,i_n\ge0 \atop i_1+\cdots+i_n=n-1} \, \frac{t_{i_1}}{i_1!} \cdots \frac{t_{i_n}}{i_n!} 
\= t_0 \+ t_0 \, t_1 \+ \frac{2 \, t_0 \, t_1^2 + t_0^2 \, t_2} 2 \+ \cdots \,,
\eeq
which is the unique power-series solution (see~\cite{DW90, Wi91}) to the 
RH hierarchy
\beq\label{RHhierarchy}
\frac{\p E(\bt)}{\p t_i} \= \frac{E(\bt)^i}{i!} \, \frac{\p E(\bt)}{\p x}\,, 
\quad i\geq0
\eeq
specified by the initial condition $E(x,0,\dots)=x$, where $x := t_0$.
Alternatively, it can be uniquely determined by the following equation:   
\beq\label{ELEv}
B_\bt(E(\bt)) \= 0 
\eeq
 (see~e.g.~\cite{DW90}), sometimes called the {\it genus zero Euler--Lagrange equation} \cite{Du96, DZ-norm}.
It is easily seen (and well known) that the power series~$E(\bt)$ has the property:
\beq
\frac{\p^k E(\bt)}{\p x^k} \= \delta_{k,1} \+ t_k \+ \text{higher degree terms}\,.
\eeq
 
The following two lemmas are important.
\begin{lemma}\label{lemmavfv}
For any $\varphi\in\G$, we have the identity:
\beq\label{gvVx0}
\varphi(E(\bT)) \= E(\bt) \,, 
\eeq
where $\bt$ and~$\bT$ are related by~\eqref{defsigmaction}. 
\end{lemma}
\begin{proof}
By definition we have
$B_\bT(E(\bT)) = 0$.
Then by~\eqref{defsigmaction} we obtain 
\beq
B_\bt(\varphi(E(\bT))) \= 0\,.
\eeq
Note that $\varphi(E(\bT))$ can be viewed as a power series of~$\bt$. The identity~\eqref{gvVx0} then holds due to
 the uniqueness of power-series solution to equation~\eqref{ELEv}.
\end{proof}

\begin{remark} We can 
extend the group~$\G$ to a semi-direct product consisting
of all pairs $(m, \varphi)$ with power series  
$m(V) \in  1 + V R[[V]]$ and $\varphi(V) \in V+V^2R[[V]]$, and with the group law $*$ given by 
\beq
(m_1,\varphi_1) * (m_2,\varphi_2)
\= (m_1 \cdot (m_2 \circ \varphi_1), \varphi_1 \circ \varphi_2) \,,
\eeq 
where ``$\cdot$" denotes multiplication of power series.
It acts on tuples~$\bt$ by sending~$\bt$ to $\bT=\bt . (m,\varphi)$,  where $B_{\bT}(V) = m(V) B_{\bt}(\varphi(V))$.
One can verify that the identity~\eqref{gvVx0} still holds for $\varphi$ in this larger group. 
One could therefore also consider partition functions under the extended group action,
but we do not know whether this would have any interesting applications.
\end{remark}

\begin{lemma}
We have
\begin{align}\label{gvVx}
\frac{\p E(\bT)}{\p t_0} \= \sqrt{\varphi'(E(\bT))} \, \frac{\p E(\bT)}{\p T_0}\,, 
\end{align}
where $\bt$ and $\bT$ are related by~\eqref{defsigmaction}.
\end{lemma}
\begin{proof}
By~\eqref{defsigmaction} and~\eqref{RHhierarchy}.
\end{proof}

For convenience, we denote $X\equiv T_0$ and $x\equiv t_0$ as in~\eqref{RHhierarchy},
and write~\eqref{gvVx} as 
\beq\label{gvVx1115}
\frac{\p E(\bT)}{\p x} \= \sqrt{\varphi'(E(\bT))} \, \frac{\p E(\bT)}{\p X}\,.
\eeq
By using the identities~\eqref{gvVx0}, \eqref{gvVx1115} iteratively, one can obtain the map 
between the higher $x$-derivatives of~$E(\bt)$ and $X$-derivatives of~$E(\bT)$. 
For instance,  
\begin{align}
\frac{\p E(\bt)}{\p x} \= \, & \varphi'(E(\bT))^{3/2} \, \frac{\p E(\bT)}{\p X} \,, \label{jetsrelation1}\\
\frac{\p^2 E(\bt)}{\p x^2} 
\= \, & 2 \, \varphi'(E(\bT)) \, \varphi''(E(\bT)) \, \biggl(\frac{\p E(\bT)}{\p X}\biggr)^2 \+ \varphi'(E(\bT))^2 \, \frac{\p^2 E(\bT)}{\p X^2}\,.
\end{align}
By induction we arrive at the following lemma describing this map.
\begin{lemma}\label{vVmaplemma}
For each $k\ge0$, there exists a function $M_k(V_0,\dots,V_k)$, which is a polynomial of $V_1,\dots,V_k$,
 such that 
\beq
\frac{\p E(\bt)}{\p x^k} \= M_k\biggl(E(\bT), \frac{\p E(\bT)}{\p X}, \dots, \frac{\p^k E(\bT)}{\p X^k}\biggr)\,.
\eeq
Moreover, for $k\ge1$, the function $M_k(V_0,\dots,V_k)$ satisfies the homogeneity condition:
\beq
\sum_{j=1}^k \, j \, V_j \, \frac{\p M_k(V_0,\dots,V_k)}{\p V_j} \= k \, M_k(V_0,\dots,V_k)\,.
\eeq
\end{lemma}
The first few $M_k$ are $M_0(V)=V$, $M_1(V,V_1)=\varphi'(V)^{3/2} V_1$, 
$M_2(V,V_1,V_2)=2 \varphi'(V) \varphi''(V) \, V_1^2 + \varphi'(V)^2 V_2$.

Recall that the free energy $\F^{\scriptscriptstyle \rm WK}(\bt;\e)$ of $\psi$-class intersection numbers is defined by 
\beq
\F^{\scriptscriptstyle \rm WK}(\bt;\e) \,:=\, \log Z^{\scriptscriptstyle \rm WK}(\bt;\e) \,.
\eeq
By definition the free energy $\F^{\scriptscriptstyle \rm WK}(\bt;\e)$ admits the following genus expansion:
\beq\label{defFZ813}
\F^{\scriptscriptstyle\rm WK}(\bt; \epsilon) \,=:\, \sum_{g\ge0} \, \e^{2g-2} \, \F^{\scriptscriptstyle\rm WK}_g(\bt) \,. 
\eeq
We call $\F^{\scriptscriptstyle\rm WK}_g(\bt)$ $(g\ge0)$ the {\it genus~$g$ free energy of $\psi$-class intersection numbers}. 
Explicitly, 
\beq
\F^{\scriptscriptstyle\rm WK}_g(\bt) \= \sum_{n\ge0} \, \frac{1}{n!} \, 
\int_{\overline{\M}_{g,n}} \, t(\psi_1) \cdots t(\psi_n)\,, \qquad t(z) \,:=\, \sum_{i\ge0} \, t_i \, z^i\,.
\eeq

\begin{defi} \label{definitionfreeenergyphi}
Let $\varphi\in\G$. The {\it WK mapping free energy associated to~$\varphi$} is defined by
\beq\label{fphi1001} 
\F^\varphi(\bT;\e) \,:=\, \F^{\scriptscriptstyle \rm WK} \bigl(\bT . \varphi^{-1}; \e \bigr) \,,
\eeq
and define the {\it genus $g$ WK mapping free energy associated to~$\varphi$}, denoted as $\F^\varphi_g(\bT)$, by 
\beq\label{fgphi1001} 
\F^\varphi_g(\bT) \,:=\, \F^{\scriptscriptstyle \rm WK}_g \bigl(\bT . \varphi^{-1}\bigr) \,, \quad g\ge0\,.
\eeq
\end{defi}
\begin{remark}
By the degree-dimension matching~\eqref{dd1003}, 
one can deduce that $\F^\varphi(\bT;\e)$ and 
$\F^\varphi_g(\bT)$ ($g\ge0$) are well-defined elements in $\e^{-2}R[[\bT]][[\e^2]]$ and $R[[\bT]]$, respectively. 
Another consequence of~\eqref{dd1003} is that we can upgrade our $\G$-action to an action of the full group of units $\CC[[V]]^\times$ by setting 
\beq\label{upFphi1003}
\F^{\varphi}(\bT;\e) \= \F^{\scriptscriptstyle\rm WK}\bigl(\bT . \varphi^{-1}, \e/\varphi'(0)^{3/2}\bigr)
\eeq
and 
\beq 
\F^\varphi_g(\bT) \= \varphi'(0)^{3g-3} \, \F^{\scriptscriptstyle\rm WK}_g \bigl(\bT . \varphi^{-1}\bigr) \,,
\eeq 
and similarly for~$Z^\varphi(\bT;\e)$ below. 
Note that formula~\eqref{upFphi1003} makes sense even without choosing a square-root of $\varphi'(0)$, 
because $\F^{\scriptscriptstyle \rm WK}$ is an even power series of~$\e$.
\end{remark}
\begin{defi}\label{definitionZphi}
For $\varphi \in \G$, 
the {\it WK mapping partition function associated to~$\varphi$} is defined by
\beq
Z^\varphi(\bT;\e) \,:=\, Z^{\scriptscriptstyle\rm WK}\bigl(\bT . \varphi^{-1};\e\bigr) \,.
\eeq 
\end{defi}

It is clear from the definitions that $\F^\varphi(\bT; \epsilon)$ has a genus expansion
\beq\label{FZphig1001}
\F^\varphi(\bT; \epsilon) \= \sum_{g\ge0} \, \e^{2g-2} \, \F^\varphi_g(\bT)
\eeq
and that 
\beq
Z^\varphi(\bT;\e) \= e^{\F^\varphi(\bT;\e)}  \,.
\eeq

Our first main result says that the power series
$\F_0^{\scriptscriptstyle\rm WK}(\bt)$ is $\G$-invariant, i.e., 
\beq
\F_0^{\scriptscriptstyle \rm WK}(\bt)\=\F_0^{\scriptscriptstyle \rm WK}(\bt . \varphi)\,, \quad \forall\,\varphi\in\G \,.
\eeq
In view of the definition of the group action~\eqref{fgphi1001}, 
we can state this even more compactly in the following way. 
\begin{theorem}\label{thmgenus0}
For any $\varphi\in\G$, we have $\F_0^\varphi=\F_0^{\scriptscriptstyle\rm WK}$.
\end{theorem}
\begin{proof}
The $\psi$-class intersection numbers in genus zero have the well-known formula:
\beq\label{genuszero712}
\int_{\overline{\M}_{0,n}} \, \psi_1^{i_1}\cdots\psi_n^{i_n} \= \binom{n-3}{i_1,\dots,i_n} \,, \quad i_1,\dots,i_n\ge0\,.
\eeq
Writing  $(n-3)!$ as $\int_0^\infty s^{n-3} e^{-s} ds$, we find
\begin{align}
\F^{\scriptscriptstyle\rm WK}_0(\bt) 
& \= \sum_{n\ge0} \, \frac1{n!} \, \sum_{i_1,\dots,i_n\ge0 \atop i_1+\cdots+i_n=n-3} \prod_{j=1}^n \frac{t_{i_j}}{i_j!} \, \int_0^\infty 
 s^{i_1+\cdots+i_n} \, e^{-s} \, ds \label{F0twoexps} \\
& \=  {\rm res}_{z=0} \biggl(\int_0^\infty e^{-s} \sum_{n=0}^\infty \sum_{i_1,\dots,i_n=0}^\infty \frac{s^{i_1+\dots+i_n}}{n!} \, z^{2-n+i_1+\dots+i_n} \prod_{j=1}^n \frac{t_{i_j}}{i_j!} \, ds \;  dz\biggr)\,, \nn \\
& \=  {\rm res}_{z=0} \biggl(\int_0^\infty e^{-B_\bt(v)/z} \, dv \,  z \; dz\biggr)\,,\nn
\end{align}
where $B_\bt(v)$ is defined by~\eqref{Bvt1001}. Here, 
in the last equality we employed the change of variables $v=zs$.

Therefore, 
\begin{align}
\F_0^\varphi(\bT) \,:=\, \F^{\scriptscriptstyle \rm WK}_0(\bt) & \=  
{\rm res}_{z=0} \biggl( \int_0^\infty e^{-B_\bt(v)/z} dv \, z \; dz\biggr) \nn\\
& \=  
{\rm res}_{\tilde z=0} \biggl(\int_0^\infty e^{-\sqrt{\varphi'(V)} \, B_\bt(\varphi(V))/\tilde z} dV \, \tilde z \; d\tilde z\biggr) \nn\\
& \= 
{\rm res}_{z=0} \biggl( \int_0^\infty e^{-B_\bT(V)/z} dV \, z \; d z\biggr) \,. \nn
\end{align}
Here, in the first line we used \eqref{fgphi1001} and \eqref{F0twoexps}, in the second line we 
employed the change of variables $v=\varphi(V)$ and $\tilde z=\sqrt{\varphi'(V)} \, z$, and 
in the last equality we used the definition~\eqref{defsigmaction}.
The theorem is proved.
\end{proof}

We also note that the power series $E(\bt)$ defined in~\eqref{defVandvf1001} is equal 
to the second $t_0$-derivative of $\F^{\scriptscriptstyle \rm WK}_0(\bt)$, i.e., 
\beq\label{idvf0xx1002}
E(\bt) \= \frac{\p^2 \F^{\scriptscriptstyle \rm WK}_0(\bt)}{\p t_0^2} \,. 
\eeq
Then from~\eqref{F0twoexps} we immediately get an integral representation for $E(\bt)$ as follows:
\begin{align}
E(\bt) 
\=  {\rm res}_{z=0} \biggl(\int_0^\infty e^{-B_\bt(v)/z} \, dv \, \frac1z \; dz\biggr)
\end{align}

Before ending this section, we make the following remark.
\begin{remark}\label{remark2}
There is another way to state Theorem~\ref{thmgenus0}. 
For any $\varphi\in\G$, define a modified right action, denoted  
 $\bt\mapsto\hat{{\bf T}}=\bt|\varphi$  
by the following formula which is similar to~\eqref{defsigmaction}, 
but with the map now being {\it linear} rather than {\it affine linear}:
\beq \label{tTdef627}
\sum_{i\ge0} \, \hat T_i \, \frac{V^i}{i!} \= \sqrt{\varphi'(V)} \, \sum_{i\ge0} \, t_i \, \frac{\varphi(V)^i}{i!} \,.
\eeq
It is clear from that~\eqref{defsigmaction} and~\eqref{tTdef627} that 
\begin{align}\label{ThatT}
 T_i \= \hat T_i + \delta_{i,1}- C_i, \qquad  \sum_{i\ge0} \, C_i \, \frac{V^i}{i!} \,:=\, \sqrt{\varphi'(V)} \, \varphi(V)\,.
\end{align}
It is also easy to deduce from~\eqref{tTdef627} that 
\begin{align}\label{thatTrelation}
B_{\bT}(V) \= \sqrt{\varphi'(V)} \,  B_{\bt+{\bf d}-{\bf c}}(\varphi(V))\,, \quad  \sum_{i\ge0} \, c_i \, \frac{v^i}{i!} \,:= \sqrt{f'(v)} \, f(v) \,, 
\quad f \,:=\, \varphi^{-1}\,,
\end{align}
where ${\bf d}=(0,1,0,0,0,\dots)$ and ${\bf c}=(c_0,c_1,c_2,\dots)$.
Theorem~\ref{thmgenus0} can then be alternatively stated as follows:  
\beq\label{f0f0Tt}
\F_0\bigl(\hat T_0, \hat T_1, \hat T_2, \dots\bigr) \= \F_0(t_0, t_1, t_2-c_2, t_3-c_3,\dots) \,.
\eeq
One can use the modified group action to define a modified 
WK mapping partition function for any $\varphi\in\G$. It leads to the same WK mapping 
hierarchy as before. The shifts are nevertheless interesting due to their connection to 
higher Weil--Petersson volumes 
\cite{KMZ, MZ00, LX09, MS08, BDY16}) 
and will be useful for several of the applications later, e.g. in
connection with the Alexandrov formula where the $c_j$ (up to a scaling
factor~$q^{j-1}$) are specific numbers $(-4, 23, -176, \cdots)$ (Theorem~B in Section~\ref{sectionexample223}). \end{remark}

\section{Virasoro constraints for the WK mapping partition function}\label{section3}
In this section, we give the Virasoro constraints for the WK mapping partition function. 

Recall from~\cite{Wi91} that the free energy $\F^{\scriptscriptstyle \rm WK}(\bt;\e)$ 
satisfies the following {\it dilaton} and {\it string} equations, respectively:
\begin{align}
& \sum_{i\ge0} \, t_i \, \frac{\p \F^{\scriptscriptstyle \rm WK}(\bt;\e)}{\p t_i} \+ \epsilon \, \frac{\p \F^{\scriptscriptstyle\rm WK}(\bt;\e)}{\p \epsilon} \+ \frac1{24} \= 
\frac{\p \F^{\scriptscriptstyle\rm WK}(\bt;\e)}{\p t_1} \,, \label{gwdilaton} \\
& \sum_{i\ge0} \, t_{i+1} \, \frac{\p \F^{\scriptscriptstyle \rm WK}(\bt;\e)}{\p t_i} \+ \frac{t_0^2}{2 \, \epsilon^2}
\= \frac{\p \F^{\scriptscriptstyle \rm WK}(\bt;\e)}{\p t_0} \,. \label{gwstring}
\end{align}

Recall also that the Witten--Kontsevich theorem can be equivalently formulated as 
an infinite family of linear constraints for~$Z^{\scriptscriptstyle \rm WK}(\bt;\e)$, which come from a  
realization of half of the Virasoro algebra of central charge~$1$, called the {\it Virasoro constraints}~\cite{DVV91, DZ-norm}. 
More precisely, define the linear operators $L_k$, $k\geq-1$, by
\begin{align}
L^{\scriptscriptstyle \rm WK}_k \= & \sum_{i\ge0} \, \frac{(2i+2k+1)!!}{2^{k+1} \, (2i-1)!!} \, t_i \, \frac{\p}{\p t_{i+k}} 
- \frac{(2k+3)!!}{2^{k+1}} \, \frac{\p}{\p t_{1+k}} \+ \frac{\delta_{k,0}}{16} \label{virakdvdefinition}\\
& + \, \frac{\epsilon^2}{2} \, \sum_{i,j\ge0 \atop i+j=k-1} \, \frac{(2i+1)!! \, (2j+1)!!}{2^{k+1}} \, \frac{\p^2}{\p t_{i} \p t_{j}} 
\+ \frac{t_0^2}{2 \, \e^2} \, \delta_{k,-1}\,. \nn
\end{align}
These operators satisfy 
the Virasoro commutation relations:
\beq
\bigl[L^{\scriptscriptstyle \rm WK}_{k_1},L^{\scriptscriptstyle \rm WK}_{k_2}\bigr] \= (k_1-k_2) \, L^{\scriptscriptstyle \rm WK}_{k_1+k_2} \,, \quad \forall\, k_1,k_2\geq -1\,.
\eeq
The Virasoro constraints for~$Z^{\scriptscriptstyle \rm WK}(\bt;\e)$ then read 
\beq\label{virasorowk}
L^{\scriptscriptstyle \rm WK}_k \bigl(Z^{\scriptscriptstyle \rm WK}(\bt;\e)\bigr) \= 0\,,  \qquad k\geq -1\,.
\eeq

Obviously, 
the $k=-1$ constraint in~\eqref{virasorowk} is the same as~\eqref{gwstring}.

\begin{prop} \label{propvira}
We have
\begin{align}
& \sum_{i\ge0} \, \tilde T_i \, \frac{\p Z^\varphi(\bT;\e)}{\p T_i} \+ \e \, \frac{\p Z^\varphi(\bT;\e)}{\p \e} \+ \frac1{24} \, Z^\varphi(\bT;\e) 
\= 0\,, \label{dilatonT}\\
& \overline{L}^\varphi_k \bigl(Z^\varphi(\bT;\e)\bigr) \= 0\,, \quad k\geq-1\,, \label{viraD}
\end{align}
where $\tilde T_i:= T_i - \delta_{i,1}$, and $\overline{L}^\varphi_k$, $k\geq-1$, are linear operators of the form 
\begin{align}
\overline{L}^\varphi_k \=  \e^2 \, \sum_{i,j\ge0} \, a^\varphi_{ij}(k) \, \frac{\p^2}{\p T_i \p T_j} \+ 
\sum_{i,j\ge0} \, b^\varphi_{ij}(k) \, \tilde T_i \, \frac{\p}{\p T_j} 
\+ \frac{T_0^2}{2 \e^2} \, \delta_{k,-1} \,+\, \frac{\delta_{k,0}}{16} \,. 
\end{align}
Here the coefficients $a_{ij}^\varphi(k), b_{ij}^\varphi(k)$ depend on $\varphi$ and~$k$. 
\end{prop}

\begin{proof} 
For $\varphi\in\G$, write 
\beq
T_i \= \delta_{i,1} \+ \sum_{m=0}^i \, M^m_i \, (t_m-\delta_{m,1}) \,, 
\quad t_m \= \delta_{m,1} \+ \sum_{i=0}^m \, N^i_m \, (T_i-\delta_{i,1}) \,. 
\eeq
Here $M^m_i$ and $N^i_m$ have dependence on~$\varphi$.
Then we have
\beq
\sum_{m\ge0} \, (t_m-\delta_{m,1}) \, \frac{\p}{\p t_m} \= 
\sum_{m\ge0} \sum_{i \ge m} \, (t_m-\delta_{m,1}) \, M^m_i \, \frac{\p}{\p T_i} 
\= \sum_{i \ge 0} \, (T_i-\delta_{i,1}) \, \frac{\p}{\p T_i} \,.
\eeq
Similarly, 
\begin{align}
 L_k^{\scriptscriptstyle \rm WK}
&  \= \frac{T_0^2}{2 \e^2} \, \delta_{k,-1} \+ \frac{\delta_{k,0}}{16} \+ \sum_{i \ge 0} \, \frac{(2i+2k+1)!!}{2^{k+1} \, (2i-1)!!} \, \sum_{j=0}^i \, N^j_i \, \tilde T_j \, \sum_{r=0}^{i+k} M^{i+k}_r \frac{\p}{\p T_r}
\nn\\
& \quad +\, \frac{\epsilon^2}{2} \, \sum_{i,j\ge0 \atop i+j=k-1} \, \frac{(2i+1)!! \, (2j+1)!!}{2^{k+1}} \, \sum_{r_1=0}^{i} \sum_{r_2=0}^j 
M^{i}_{r_1} M^{j}_{r_2} \frac{\p}{\p T_{r_1}} \frac{\p}{\p T_{r_2}} . \nn
\end{align}
The proposition is proved.
\end{proof}

We call~\eqref{dilatonT} the {\it dilaton equation} for~$Z^\varphi(\bT;\e)$, and~\eqref{viraD} the {\it Virasoro constraints} for~$Z^\varphi(\bT;\e)$ because 
the operators $\overline{L}^\varphi_k$ 
satisfy the Virasoro commutation relations:
\beq
\Bigl[\overline{L}^\varphi_k \,, \, \overline{L}^\varphi_\ell\Bigr]
\= (k-\ell) \, \overline{L}^\varphi_{k+\ell}\,, \quad k,\ell\geq-1\,.
\eeq

\section{The WK mapping free energy in genus zero}\label{section4}
In this section, we give a different and more geometric proof
of Theorem~\ref{thmgenus0}.

\begin{lemma} \label{lemma2derivs}
The following identity holds:
\beq\label{genuszerof0ff0}
\frac{\p^2 \F^{\scriptscriptstyle\rm WK}_0(\bt)}{\p t_i \p t_\ell} \= 
\frac{\p^2 (\F^{\scriptscriptstyle\rm WK}_0 (\bt . \varphi))}{\p t_i \p t_\ell}  \,, \quad i,\ell\ge0 \,.
\eeq
\end{lemma}

\begin{proof}
Recall the following well-known identity
\begin{align}\label{twopointgenus0formula}
\frac{\p^2 \F^{\scriptscriptstyle\rm WK}_0(\bt)}{\p t_i \p t_\ell} \= \frac{E(\bt)^{i+\ell+1}}{i! \, \ell! \, (i+\ell+1)} \,, \qquad \forall\, i, \ell\geq0\,,
\end{align}
(see~e.g.~\cite{Du96, DZ-norm})
with the $i=\ell=0$ case being the same as~\eqref{idvf0xx1002}.
From~\eqref{twopointgenus0formula} one can easily deduce:
\begin{align}
& \frac{\p^3 \F^{\scriptscriptstyle \rm WK}_0(\bt)}{\p t_i \p t_\ell \p t_m} \= \frac{E(\bt)^{i+\ell+m}}{i! \, \ell! \, m!} \frac{\p v(\bt)}{\p t_0} \,, \label{3derivatives1}\\
& \frac{\p^3 (\F^{\scriptscriptstyle \rm WK}_0 (\bt . \varphi))}{\p t_i \p t_\ell \p t_m} 
\= \sum_{m_1,m_2,m_3\ge0} \frac{\p T_{m_1}}{\p t_i} \frac{\p T_{m_2}}{\p t_\ell} 
\frac{\p T_{m_3}}{\p t_m} 
\frac{E(\bT)^{m_1+m_2+m_3}}{m_1! \, m_2! \, m_3! } \frac{\p E(\bT)}{\p T_0} \,, \label{3drho1002}
\end{align}
where $\bT$ and $\bt$ are related by $\bT=\bt . \varphi$, and $i, \ell, m\geq0$.
By using~\eqref{defsigmaction} we have
\beq \label{123456}
\sum_{m=0}^\infty \, \frac{\p T_m}{\p t_i} \, \frac{E(\bT)^m}{m!} \= \sqrt{\varphi'(E(\bT))} \, \frac{\varphi(E(\bT))^i}{i!}\,.
\eeq

Using \eqref{3drho1002}, \eqref{123456} and Lemma~\ref{lemmavfv}, we obtain that 
\begin{align}
\frac{\p^3 (\F^{\scriptscriptstyle\rm WK}_0 (\bt . \varphi))}{\p t_i \p t_\ell \p t_m} \= \varphi'(E(\bT))^{3/2} \, \frac{E(\bt)^{i+\ell+m}}{i! \, \ell! \, m!} \, 
\frac{\p E(\bT)}{\p T_0} \,. \label{3derivatives2}
\end{align}
We conclude from~\eqref{3derivatives1}, \eqref{3derivatives2}, \eqref{jetsrelation1} that 
\beq
\frac{\p^3 \F^{\scriptscriptstyle\rm WK}_0(\bt)}{\p t_i \p t_\ell \p t_m} 
\= \frac{\p^3 (\F^{\scriptscriptstyle\rm WK}_0 (\bt . \varphi))}{\p t_i \p t_\ell \p t_m} \,, \quad i,\ell,m\geq0\,.
\eeq
So the two sides of~\eqref{genuszerof0ff0} can only differ by a constant. 
The lemma is then proved by observing that they both vanish when $\bt={\bf 0}$. 
\end{proof}

\begin{proof}[A second proof of Theorem~\ref{thmgenus0}]
It follows from the dilaton equation that 
\beq\label{47eq}
2 \, \F^{\scriptscriptstyle\rm WK}_0(\bt) \= 
\sum_{i\ge0} \, t_i \, \frac{\p \F^{\scriptscriptstyle\rm WK}_0(\bt)}{\p t_i} - \frac{\p \F^{\scriptscriptstyle \rm WK}_0(\bt)}{\p t_1} \,.
\eeq
Differentiating this identity with respect to $t_\ell$ we find 
\beq\label{48eq}
\frac{\p \F^{\scriptscriptstyle \rm WK}_0(\bt)}{\p t_\ell} \= 
\sum_{i\ge0} \, t_i \, \frac{\p^2 \F^{\scriptscriptstyle \rm WK}_0(\bt)}{\p t_i \p t_\ell} 
- \frac{\p^2 \F^{\scriptscriptstyle \rm WK}_0(\bt)}{\p t_1 \p t_\ell} \,, \quad \ell\ge0\,.
\eeq
From \eqref{47eq}, \eqref{48eq} we see that the power series~$\F^{\scriptscriptstyle \rm WK}_0(\bt)$ is uniquely determined 
by $\frac{\p^2 \F^{\scriptscriptstyle \rm WK}_0(\bt)}{\p t_i \p t_\ell}$, $i,\ell\geq0$. 
Combined with Lemma~\ref{lemma2derivs}, this proves the theorem. 
\end{proof}

\section{The higher genus WK mapping free energies} \label{section5}
In this section, we show that the higher genus WK mapping free energies admit jet representations.

It is known that \cite{DW90, DY20-2, DZ-norm, Getzler} 
the power series $\F^{\scriptscriptstyle\rm WK}_g(\bt)$, $g\ge1$, has the $(3g-2)$-jet representation, i.e., 
there exists $F^{\scriptscriptstyle\rm WK}_g(v_1,\dots,v_{3g-2})$, such that 
\beq\label{WK3gminus2}
\F^{\scriptscriptstyle\rm WK}_g(\bt) \= 
F^{\scriptscriptstyle\rm WK}_g\biggl(\frac{\p E(\bt)}{\p t_0}, \dots, \frac{\p^{3g-2} E(\bt)}{\p t_0^{3g-2}}\biggr)\,, \quad g\geq1\,,
\eeq
with 
\beq
F^{\scriptscriptstyle\rm WK}_1(v_1) \= \frac1{24} \log v_1 \,.
\eeq
Moreover, for $g\ge2$, $F^{\scriptscriptstyle\rm WK}_g(v_1,\dots,v_{3g-2})$ is a polynomial 
of $v_2,\dots,v_{3g-2}$ and $v_1^{-1}$ (see~e.g.~\cite{DY20-2, DZ-norm}), that 
satisfies the following two homogeneity conditions:
\begin{align}
& \sum_{k \ge 1} \, k \, v_k \, \frac{\p F^{\scriptscriptstyle\rm WK}_g(v_1,\dots,v_{3g-2})}{\p v_k} 
\= (2g-2) \, F^{\scriptscriptstyle\rm WK}_g(v_1,\dots,v_{3g-2}) \,, \quad g\ge2\,, \label{homowk-1}\\
& \sum_{k \ge 2} \, (k-1) \, v_k \, \frac{\p F^{\scriptscriptstyle\rm WK}_g(v_1,\dots,v_{3g-2})}{\p v_k} 
\= (3g-3) \, F^{\scriptscriptstyle\rm WK}_g(v_1,\dots,v_{3g-2}) \,, \quad g\ge2\,. \label{homowk-2}
\end{align}

By using \eqref{fgphi1001}, Lemma~\ref{vVmaplemma} and~\eqref{WK3gminus2}--\eqref{homowk-2} we arrive at the following proposition.
\begin{prop}\label{jetreprep}
For $g=1$ we have the identity:
\beq\label{genus1K1}
\F^\varphi_1(\bT) \= F^\varphi_1\biggl(E(\bT), \frac{\p E(\bT)}{\p X}\biggr), ~ {\rm with~}
F^\varphi_1(V,V_1) \,:=\, \frac{1}{24} \log V_1 \+ \frac1{16} \log \varphi'(V) \,.
\eeq
For each $g\ge2$,  $\F^\varphi_g(\bT)$ is given by  
\begin{align}
& \F^\varphi_g(\bT) \= F^\varphi_g\biggl(E(\bT), \dots, \frac{\p^{3g-2} E(\bT)}{\p X^{3g-2}}\biggr)  \label{jetfvg89}
\end{align}
for some function $F^\varphi_g(V_0, \dots, V_{3g-2})$
which is a polynomial in $V_1^{-1}$,
$V_2$, \dots, $V_{3g-2}$. Moreover, this polynomial is weighted homogeneous of degree $2g-2$ (where $V_i$ has weight $i$), 
i.e.,
\beq
\sum_{k=1}^{3g-2} \, k \, V_k \, \frac{\p F^\varphi_g(V_0, \dots, V_{3g-2})}{V_k} \= (2g-2) \, F^\varphi_g(V_0, \dots, V_{3g-2})\,.
\eeq
What is more, for each $g\ge1$ the differences 
$\F^\varphi_g(\bT) - \F^{\scriptscriptstyle\rm WK}_g(\bT)$
and $F^\varphi_g(V_0, \dots, V_{3g-2}) - F^{\scriptscriptstyle\rm WK}_g(V_1, \dots, V_{3g-2})$, 
as power series of $a_2,a_3,a_4,\dots$, have vanishing constant terms.
\end{prop}
For instance, for $g=2$ we have the following explicit expression for $F^\varphi_2$:
\begin{align}
& F^\varphi_2(V,V_1,V_2,V_3,V_4) \= \frac{V_4}{1152 \, V_1^2} \,-\, \frac{7 \, V_3 V_2}{1920 \, V_1^3} \+ \frac{V_2^3}{360 \, V_1^4} \label{expF2phi}\\
& \quad\quad\quad +\, 
\frac{\varphi''(V)}{320 \, \varphi'(V)} \, \frac{V_3}{V_1} \, - \, \frac{11 \, \varphi''(V)}{3840 \, \varphi'(V)} \, \frac{V_2^2}{V_1^2}  \+ \biggl( \frac{5 \, \varphi^{(3)}(V)}{768 \, \varphi'(V)} 
- \frac{29 \, \varphi''(V)^2}{7680 \, \varphi'(V)^2} \biggr) \, V_2 \nn\\
& \quad\quad\quad +\, \biggl( \frac{\varphi^{(4)}(V)}{384 \, \varphi'(V)} + \frac{\varphi''(V)^3}{11520 \, \varphi'(V)^3} - \frac{\varphi^{(3)}(V) \, \varphi''(V)}{384 \, \varphi'(V)^2} \biggr) \, V_1^2 \,. \nn
\end{align}
In the next section we will show that this function, and also the higher $F_g^\varphi$, are equal to a power of $V_1$ 
times a weighted homogeneous polynomial in the variables $V_{i+1}/V_1^{i+1}$ 
and $d^i(\log \varphi'(V))/dV^i$ (eqs~\eqref{fp59} and~\eqref{fkpk59}).

From the last statement in Proposition~\ref{jetreprep} we know that for each $g\ge1$, $\F^\varphi_g(\bT)$ is a deformation 
of~$\F^{\scriptscriptstyle\rm WK}_g(\bT)$,  as well as that $F^\varphi_g(V_0, \dots, V_{3g-2})$ is a 
deformation of $F^{\scriptscriptstyle \rm WK}_g(V_1, \dots, V_{3g-2})$. 
For $g=1,2$, this is obvious from~\eqref{genus1K1}, \eqref{expF2phi}. 
An alternative way to see this e.g. in genus $g=1$ is from the identity
\beq
\F^\varphi_1(\bT) - \F^{\scriptscriptstyle\rm WK}_1(\bT) \= \frac1{16} \, \log \varphi'\biggl(\frac{\p^2\F^{\scriptscriptstyle\rm WK}_0(\bT)}{\p T_0^2}\biggr) \,.
\eeq

\section{The loop equation for the WK mapping free energy}\label{newnewsection6}
This section devotes to the derivation of the loop equations for the WK mapping free energy.

Following~\cite{DZ-norm}, introduce the following creation and annihilation operators: 
\beq
a_{p} \= \left\{\begin{array}{cc} 
\e \, \frac{\p}{\p t_{p-1/2}}\,, & \quad p>0\,, \\ 
\e^{-1} \, (-1)^{p+1/2} \, (t_{-p-1/2}-\delta_{p,-3/2})\,, & \quad p<0 \,, \\
\end{array}\right.
\eeq
where $p$ is a half integer. Let 
\beq
f \= \sum_{p\in \ZZ+\frac12} \, a_{p} \, \int_0^\infty \, e^{-\lambda z} \, z^{p-1} dz \= \sum_{p\in \ZZ+\frac12} \, a_{p} \, \Gamma(p) \, \lambda^{-p}\,.
\eeq
Then
\beq
\p_\lambda (f) \= 
- \sum_{p\in \ZZ+\frac12} \, a_{p} \, \Gamma(p+1) \, \lambda^{-p-1} \, =: \, - 
\sqrt{\pi} \, \e \, A \,-\, \frac{\sqrt{\pi}}{\e} \, B\,,
\eeq
with 
\begin{align}
& A 
\= \sum_{m\ge0} \, \frac{(2m+1)!!}{2^{m+1}} \, \lambda^{-m-3/2} \, \frac{\p}{\p t_{m}} \,, \\
& B 
\= \sum_{m\ge0} \, \frac{2^m}{(2m-1)!!} \, \lambda^{m-1/2} \, (t_{m}-\delta_{m,1})  \,. 
\end{align}

It can then be verified that 
\beq
\sum_{k\geq-1} \, \frac{L_k}{\lambda^{k+2}} \;:=\; (T(\lambda))_{\leq-1}\,,
\eeq 
where $L_k$, $k\ge-1$, are the operators defined in~\eqref{virakdvdefinition}, and 
\begin{align}
&T(\lambda) \;:=\; \frac1{2\pi} \, : (\p_\lambda f)^2 : \+ \frac1{16 \, \lambda^2}
\= \frac{\e^2}2 \, A^2 \+ B \circ A \+ \frac1{2\e^2} \, B^2 \+ \frac1{16 \, \lambda^2} \, .
\end{align}
The Virasoro constraints~\eqref{virasorowk} can now be written as
\beq
(T(\lambda))_{-} (Z(\bt;\e)) \= 0 \,.
\eeq 
Here ``$-$" means taking the negative power of~$\lambda$.
By definition we have
\beq\label{virazft}
(T(\lambda))_{-} \bigl(Z^\varphi(\bT;\e)\bigr) \= 0 \,.
\eeq 

Dividing both sides of~\eqref{virazft} by~$Z^\varphi$ and taking the coefficients of~$\e^{-2}$, we have
\beq\label{virazero}
 (B \circ A)_- \bigl(\F^\varphi_0(\bT)\bigr) 
\+ \frac{1}2 \, \bigl(A\bigl(\F^\varphi_0(\bT)\bigr)\bigr)^2 \+ \frac1{2} \, (B^2)_- \= 0 \,.
\eeq
For $k\ge0$, applying $\p_X^{k+2}$ on both sides of the equality~\eqref{virazero}, one obtains 
\begin{align}
& 
\bigl(A\bigl(\F^\varphi_0\bigr) \circ A + (B\circ A)_- \bigr) (V_{k}) \label{6089}\\
 \= & - \sum_{m=1}^k \binom{k}{m} \, \p_X^{m-1} \bigl(A\bigl(\F^\varphi_{0X}\bigr)\bigr) \, \p_X^{k-m} (A(V)) \nn\\
&  - \p_X^k\Bigl((k+2) \bigl(B_{X}  \circ A \bigl(\F^\varphi_{0X}\bigr)\bigr)_-
\+ \bigl(A\bigl(\F^\varphi_{0X}\bigr)\bigr)^2 \+ \bigl((B_{X})^2\bigr)_-  \Bigr) \nn\\
 \= & - \sum_{m=1}^k \binom{k}{m} \, 
 \Bigl(\p_X^{m-1} \bigl(B_X+A\bigl(\F^\varphi_{0X}\bigr)\bigr) \, \p_X^{k-m+1} \bigl(B_X+A\bigl(\F^\varphi_{0X}\bigr)\bigr)\Bigr)_- \nn\\
&  - \p_X^k\Bigl(\bigl(B_{X} +  A \bigl(\F^\varphi_{0X}\bigr)\bigr)^2 \Bigr)_-\,. \nn
\end{align}
Here $\F^\varphi_0:=\F^\varphi_0(\bT)=\F_0(\bT)$.

Introduce  
\beq
\F^\varphi_{\rm h.g.} \= \F^\varphi_{\rm h.g.} (\bT;\e) \;:=\; \F^\varphi(\bT; \epsilon) - \e^{-2} \F_0 (\bT) \= \sum_{g\geq1} \e^{2g-2} \F^\varphi_g (\bT)\,.
\eeq
Here and below ``h.g." stands for higher genera.
Dividing both sides of~\eqref{virazft} by~$Z^\varphi$ and taking the coefficients of nonnegative power of~$\e$, we find 
\begin{align}
& (B \circ A)_- \bigl(\F^\varphi_{\rm h.g.}\bigr) \+ \frac{\e^2}2 \, \Bigl(\bigl(A\bigl(\F^\varphi_{\rm h.g.}\bigr)\bigr)^2 \+ A^2 \bigl(\F^\varphi_{\rm h.g.}\bigr) \Bigr) \label{higherviraft} \\
& \+ A\bigl(\F^\varphi_0\bigr) \, A\bigl(\F^\varphi_{\rm h.g.}\bigr) \+ \frac12 \, A^2 \bigl(\F^\varphi_0\bigr) \+ \frac1{16 \, \lambda^2} \= 0 \,. \nn
\end{align}
Substituting the jet representation \eqref{genus1K1}, \eqref{jetfvg89} for~$\Delta\F^\varphi$ into~\eqref{higherviraft}, we obtain 
\begin{align}
& - \sum_{k\ge0} \sum_{m=1}^k \binom{k}{m} \Bigl(\p_X^{m-1}\bigl(B_X+A\bigl(\F^\varphi_{0X}\bigr)\bigr) \, 
\p_X^{k-m+1} \bigl(B_X+A(\F^\varphi_{0X})\bigr)\Bigr)_- \, \frac{\p \F^\varphi_{\rm h.g.}}{\p V_{k}} \label{equation80} \\
& - \sum_{k\ge0} \p_X^k\Bigl(\bigl(B_{X} + A \bigl(\F^\varphi_{0X})\bigr)^2\Bigr)_- \, 
 \frac{\p \F^\varphi_{\rm h.g.}}{\p V_{k}} \nn\\
& \+ \frac{\e^2}2 \, \sum_{q_1,q_2\ge0} \, \p_X^{q_1+1} \bigl(A\bigl(\F^\varphi_{0 X}\bigr)\bigr) \, \p_X^{q_2+1} \bigl(A\bigl(\F^\varphi_{0X}\bigr)\bigr) \, 
\Bigl(\frac{\p \F^\varphi_{\rm h.g.}}{\p V_{q_1}}  \frac{\p \F^\varphi_{\rm h.g.}}{\p V_{q_2}} +
\frac{\p^2 \F^\varphi_{\rm h.g.}}{\p V_{q_1} \p V_{q_2}}\Bigr)\nn\\
& \+ \frac{\e^2}2 \, \sum_{m\ge0} \, A^2 (V_{m}) \, \frac{\p \F^\varphi_{\rm h.g.}}{\p V_{m}} \+ \frac12 \, A^2 \bigl(\F^\varphi_0\bigr) 
\+ \frac1{16 \, \lambda^2} \= 0 \,. \nn
\end{align}
Here we also used~\eqref{6089}. 
We note that 
\begin{align}
& A\bigl(\F^\varphi_{0 X}\bigr) 
\= \sum_{a, \,i\ge0} \, \frac{(2a+1)!!}{2^{a+1}} \, \frac{E(\bt)^{a+i+1}}{a! \, i! \, (a+i+1)} \, \frac{\p t_i}{\p X} \, \lambda^{-a-3/2} \,, \label{Aappformula} \\
&A^2\bigl(\F^\varphi_0\bigr) \= \frac{1}{8 \, (\lambda-\varphi(E(\bT)))^2} \,-\, \frac{1}{8 \, \lambda^2} \,,\\
& B_X \= \sum_{i\geq0} \frac{\p B}{\p t_i} \frac{\p t_i}{\p X} \= 
\sum_{b\ge0} \, \frac{2^b \, b!}{(2b-1)!!} \, \lambda^{b-1/2} \, {\rm Coef}\biggl(x^b,\frac1{\sqrt{\varphi'(\varphi^{-1}(x))}}\biggr) \,. \label{BXformula}
\end{align}

It follows from the equality~\eqref{Aappformula} and Lemma~\ref{lemmavfv} that 
\begin{align}
&A\bigl(\F^\varphi_{0 X}\bigr) \= \frac12 \, \int_0^{E(\bt)} \, \frac1{(\lambda-x)^{3/2}}\, \frac{dx}{\sqrt{\varphi'(\varphi^{-1}(x))}} 
\= \frac12 \, \int_0^{E(\bT)} \, \frac{\sqrt{\varphi'(x)}}{(\lambda-\varphi(x))^{3/2}}\, dx \,.
\end{align}
Together with~\eqref{BXformula}
we find via integration by parts 
that $B_X + A(\F^\varphi_{0 X})$ admits the following explicit Puiseux expansion 
as $\lambda\to \varphi(V)$:
\begin{align}
& B_X \+ A\bigl(\F^\varphi_{0 X}\bigr) \\
&  \= \Biggl(
\sum_{k\geq-1} \, \frac{2^{k+1}}{(2k+1)!!}  \, (\lambda-\varphi(V))^{k+1/2}  \, \biggl(\frac1{\varphi'(V)} \, \p_V\biggr)^{k+1} \biggl(\frac1{\sqrt{\varphi'(V)}}\biggr) \Biggr)\Bigg|_{V=E(\bT)} \,.  \nn
\end{align}
Then by noticing that 
the genus~$g$ part of equation~\eqref{equation80} admits a Laurent expansion as $\lambda\to\varphi(V)$
and that the vanishing of the 
coefficients of negative powers in~$\lambda$ is equivalent to the vanishing of the coefficients 
of negative powers in~$\lambda-\varphi(V)$ in the Laurent expansion, 
we arrive at
\begin{align}
	& - \sum_{k\geq 0} \, \Biggl(\p^k\big(\W(\lambda)^2\big) \+ \sum_{j=1}^{k}\,\binom{k}{j}\, \Bigl(
	\p^{j-1}\big(\W(\lambda)\big) \, \p^{k+1-j} \big(\W(\lambda)\big)\Bigr) \Biggr)^- \, \frac{\p \F^\varphi_{\rm h.g.}}{\p V_k} \label{eqn:loopdz} \\
	& \+ \frac{\epsilon^2}{2}\,\sum_{k, \ell\geq 0}\,\p^{k+1}\big(\W(\lambda)\big) \, \p^{\ell+1}\big(\W(\lambda)\big)\,
	\bigg(\frac{\p^2 \F^\varphi_{\rm h.g.}}{\p V_k\p V_\ell} \+ \frac{\p \F^\varphi_{\rm h.g.}}{\p V_k}\frac{\p \F^\varphi_{\rm h.g.}}{\p V_\ell}\bigg)\nn\\
	& \+ \frac{\epsilon^2}{16} \, \sum_{k\ge0} \, \p^{k+2}\biggl(\frac{1}{(\lambda-\varphi(V))^2}\biggr) \,\frac{\p \F^\varphi_{\rm h.g.}}{\p V_k} 
	\+ \frac{1}{16} \, \frac{1}{(\lambda-\varphi(V))^2}\=0\,, \nn
\end{align}
where $\p=\sum_k V_{k+1} \p/\p V_k$, 
\beq
\W(\lambda) \; := \;
\sum_{s\geq0} \, \frac{2^{s}}{(2s-1)!!}  \, (\lambda-\varphi(V))^{s-1/2}  \, \biggl(\frac1{\varphi'(V)} \p_V\biggr)^{s} \biggl(\frac1{\sqrt{\varphi'(V)}}\biggr)\,,
\eeq
and $(\bullet)^-$ means taking terms having negative powers of~$\lambda-\varphi(V)$.

We have the following theorem.
\begin{theorem}\label{thmloop}
The generating function 
\beq 
F^\varphi_{\rm h.g.} \= F^\varphi_{\rm h.g.}(\epsilon) \:= \sum_{g\ge1} \, \e^{2g-2} \, F^\varphi_g(V_0, \dots, V_{3g-2})
\eeq 
satisfies
\begin{align}
	& - \sum_{k\ge0} \, \biggl(\p^k\big(W^2\big) \+ \sum_{j=1}^{k}\,\binom{k}{j}\,
	\p^{j-1}\big(W\big) \, \p^{k+1-j} \big(W\big) \biggr)^- \; \frac{\p F^\varphi_{\rm h.g.}}{\p V_k} \label{loopequationfinal} \\
	& \+ \frac{\epsilon^2}{2}\,\sum_{k,\ell\ge0} \, \Bigl(\p^{k+1}\big(W\big) \, \p^{\ell+1}\big(W\big) \Bigr)^- \,
	\biggl(\frac{\p^2 F^\varphi_{\rm h.g.}}{\p V_k\p V_\ell} \+ \frac{\p F^\varphi_{\rm h.g.}}{\p V_k}\frac{\p F^\varphi_{\rm h.g.}}{\p V_\ell}\biggr)\nn\\
	& \+ \frac{\epsilon^2}{16} \, \sum_{k\ge0} \, \p^{k+2}\biggl(\frac{1}{\Delta^2}\biggr) \, \frac{\p F^\varphi_{\rm h.g.}}{\p V_k} 
	\+ \frac{1}{16} \, \frac{1}{\Delta^2} \= 0\,, \nn
\end{align}
where $W$ is the element in 
\beq
\varphi'(V)^{-1/2} \Delta^{-1/2} \QQ\bigl[\varphi'(V)^{\pm1}, \varphi''(V), \varphi'''(V), \dots\bigr][[\Delta]]
\eeq 
defined by 
\beq 
W \;:=\; 
\sum_{s\geq0}\frac{2^{s}}{(2s-1)!!} \, 
\biggl(\biggl(\frac1{\varphi'(V)} \frac{\p}{\p V} \biggr)^{s} \biggl(\frac1{\sqrt{\varphi'(V)}}\biggr)\biggr) \, \Delta^{s-\frac12} \,,
\eeq
the operator $\p$ is defined on functions of $\Delta, V_0, V_1, V_2, \dots$ by 
\beq
\p \:= - \varphi'(V) \, V_1 \, \frac{\p }{\p \Delta} \+ \sum_{k\ge0} \, V_{k+1} \, \frac{\p}{\p V_k}  \, ,
\eeq 
and $(\bullet)^-$ means taking terms having negative powers of~$\Delta$.
Moreover, the solution to~\eqref{loopequationfinal} is unique up to a sequence of 
additive constants which can be uniquely fixed by~\eqref{genus1K1} and the following equation:
\beq\label{homogfvpg}
\sum_{k\geq1} \, k \, V_k \, \frac{\p F^\varphi_g}{\p V_k} \= (2g-2) \, F^\varphi_g \+ \frac{\delta_{g,1}}{24} \,, \quad g\geq1\,.
\eeq
\end{theorem} 

\begin{proof} 
Observing that $\W(\lambda)=W|_{\Delta\mapsto\lambda-\varphi(V)}$, equation~\eqref{loopequationfinal} 
is just a rewriting of~\eqref{eqn:loopdz}. 

To see uniqueness, 
by taking coefficients of powers of $\e^{2g-2}$, $g\ge1$, in the loop equation~\eqref{loopequationfinal} 
we see that the loop equation~\eqref{loopequationfinal}  is equivalent to
\begin{align}
	&  \sum_{k\ge0} \, \biggl(\p^k\bigl(W^2\bigr) \+ \sum_{j=1}^{k}\,\binom{k}{j} \, 
	\p^{j-1}\big(W\big) \, \p^{k+1-j} \bigl(W\bigr) \biggr)^- \, \frac{\p F^\varphi_g}{\p V_k} \label{loopequationfinalfg11108} \\
	& \= \frac{1}{2}\,\sum_{k,\ell\ge0} \, \bigl(\p^{k+1}\big(W\big) \, \p^{\ell+1}\big(W\big) \bigr)^- \,
	\bigg(\frac{\p^2 F^\varphi_{g-1}}{\p V_k\p V_\ell} 
	\+ \sum_{m=1}^{g-1} \frac{\p F^\varphi_{m}}{\p V_k}\frac{\p F^\varphi_{g-m}}{\p V_\ell}\bigg)\nn\\
	& \qquad + \, \frac{1}{16} \, \sum_{k\ge0} \, \p^{k+2}\Bigl(\frac{1}{\Delta^2}\Bigr) \,\frac{\p F^\varphi_{g-1}}{\p V_k} 
	\+ \frac{1}{16} \, \frac{1}{\Delta^2} \, \delta_{g,1} \,,\qquad g\geq1 \nn
\end{align}
(setting $F^\varphi_{0}=0$). For each $g\ge1$, since $F^\varphi_{g}$ is a function of $V_0,\dots,V_{3g-2}$ 
the sum $\sum_k$ on the left-hand side of~\eqref{loopequationfinalfg11108} is actually is finite sum, and 
by comparing coefficients of negatives powers of~$\Delta$ we find that~\eqref{loopequationfinalfg11108} 
is equivalent to the following triangular inhomogeneous linear system for the gradients of $F^\varphi_{g}$:
\beq
C_g \, \biggl(\frac{\p F^\varphi_{g}}{\p V_0}, \dots, \frac{\p F^\varphi_{g}}{\p V_{3g-2}}\biggr)^T \= M_g \, ,
\eeq
where $M_g$ is a column vector which is determined by $F^\varphi_{1}, \dots, F^\varphi_{g-1}$ and~$W$, and 
$C_g$ is an upper triangular matrix determined by~$W$. Moreover, by a straightforward calculation we find 
\beq
\det C_g \= \prod_{j=0}^{3g-2} \frac{(2j+1)!!}{2^j} \, \varphi'(V)^{j-1} \;\neq\; 0 \, ,
\eeq
which implies that~\eqref{loopequationfinalfg11108} gives a recursive formula for the gradients of $F^\varphi_{g}$, $g\ge1$.
For $g=1$, 
equation~\eqref{loopequationfinalfg11108} reads
\begin{align}
& \biggl(\frac32 \, \frac{V_1}{\Delta^2} - \frac32 \, \frac{\varphi''(V)}{\varphi'(V)^2} \, \frac{V_1}{\Delta}\biggr) \, \frac{\p F^\varphi_1}{\p V_1} 
\+ \frac{1}{\varphi'(V) } \, \frac1\Delta \, \frac{\p F^\varphi_1}{\p V_0}  \= \frac1{16 \, \Delta^2} \,.  \nn
\end{align}
By equating the coefficients of $\Delta^{-1}$ and $\Delta^{-2}$ to~0, we get
\beq
\frac{\p F^\varphi_1}{\p V_{1}} \= \frac1{24 \, V_1} \,, \quad 
\frac{\p F^\varphi_1}{\p V} \= \frac3{2} \, V_1 \, \frac{\varphi''(V)}{\varphi'(V)} \, \frac{\p F^\varphi_1}{\p V_{1}} \= \frac1{16} \,  \frac{\varphi''(V)}{\varphi'(V)} \,,
\eeq
which agrees with~\eqref{genus1K1}. 
For $g\ge2$, the homogeneity~\eqref{homogfvpg} fixes $F^\varphi_g$ by its gradients. 
\end{proof}

When $g=2$, the expression of $F^\varphi_2$ obtained using a computer algorithm designed from the above theorem  
coincides with the one given by~\eqref{expF2phi}. 
We have made this double check also for $g=3,4,5$ with a simple home computer.

We call \eqref{equation80} or~\eqref{loopequationfinal} or~\eqref{loopequationfinalfg11108} the 
{\it Dubrovin--Zhang type loop equation for the WK mapping partition function}, for short,  
 the {\it loop equation}.
 
 \begin{remark}\label{remarkloopgradient}
 We note that the existence of solution to the loop equation~\eqref{equation80} is a non-trivial fact. Our construction proves this 
 existence.
\end{remark}

\begin{remark}
For the case that $\varphi(V)=V$, the loop equation~\eqref{loopequationfinal} reduces to 
\begin{align}
&  - \sum_{k\ge0} \,  \biggl(\p^k\biggl(\frac1{\Delta} \biggr) \+ \sum_{m=1}^k \binom{k}{m} \p^{m-1}\biggl(\frac1{\sqrt{\Delta}}\biggr) \, 
\p^{k-m+1} \biggl(\frac1{\sqrt{\Delta}}\biggr) \biggr) \, \frac{\p F^\varphi_{\rm h.g.}}{\p V_{k}} 
\nn\\
& \+ \frac{\e^2}2 \, \sum_{k_1, \, k_2 \ge 0} \, \p^{k_1+1} \biggl(\frac1{\sqrt{\Delta}}\biggr) \, 
\p^{k_2+1} \biggl(\frac1{\sqrt{\Delta}}\biggr) \, \biggl(\frac{\p^2 F^\varphi_{\rm h.g.}}{\p V_k\p V_\ell} 
\+ \frac{\p F^\varphi_{\rm h.g.}}{\p V_k}\frac{\p F^\varphi_{\rm h.g.}}{\p V_\ell}\biggr)\nn\\
& \+ \frac{\epsilon^2}{16} \, \sum_{k\ge0} \, \p^{k+2} \biggl(\frac{1}{\Delta^2}\biggr) \, 
\frac{\p F^\varphi_{\rm h.g.}}{\p V_{k}} \+ \frac{1} {16 \, \Delta^2} \= 0\,. \nn
\end{align}
This loop equation coincides with the loop equation for the Witten--Kontsevich partition function, derived by Dubrovin and Zhang 
in~\cite{DZ-norm}.
\end{remark}

Introduce the polynomial ring 
\beq
\mathcal{R} \= \QQ[w_1,w_2,\dots; \ell_1,\ell_2,\dots]\,.
\eeq
As a vector space $\R$ decomposes into a direct sum of homogeneous subspaces
\beq
\R \= \oplus_{m\ge0} \R^{[m]}\,,
\eeq
where elements in $\R^{[m]}$ are weighted homogeneous polynomials of degree~$m$ in 
variables $w_i$ and $\ell_i$ of weight~$i$ ($i\ge 1$).
In the remainder of this section we will give a more elementary description of the functions 
 $F_g^\varphi$ by showing that 
\begin{align}
F^\varphi_g(V,V_1,\dots,V_{3g-2}) \= V_1^{2g-2} \, P_g\Bigl(\frac{V_2}{V_1^2}, \frac{V_{3}}{V_1^{3}},\dots;  l_1(V), l_2(V), \dots \Bigr) \quad (g\ge2)\,,
\label{fp59}
\end{align}
where $l_k(V)$ is defined by 
\beq\label{deflkv63}
l_k(V) \:= \Bigl(\frac{d}{dV}\Bigr)^k (\log \varphi'(V)) \,, \quad k\ge0 \,,
\eeq
and $P_g=P_g(w_1,w_2,\dots; \ell_1,\ell_2,\dots)$ is a polynomial in $\P^{[3g-3]}$. 
For instance, for $g=2$, 
\begin{align}
& P_2 \= \frac{w_3}{1152} \,-\, \frac{7 \, w_1 \, w_2}{1920} \+ \frac{w_1^3}{360} 
\+ \frac{\ell_1 \, w_2}{320} \, - \, \frac{11 \, \ell_1 \, w_1^2}{3840} \\
& \quad + \, \biggl( \frac{5 \, \ell_2 }{768} + \frac{7 \, \ell_1^2 }{2560} \biggr) \, w_1 \+ 
\biggl( \frac{\ell_1 \, \ell_2}{192} + \frac{\ell_1^3}{11520} + \frac{\ell_3}{384} \biggr) \,,\nn
\end{align}
which is much simpler to read than the equivalent expression~\eqref{expF2phi} for $F_2^\varphi$ and  
belongs to $\R^{[3]}$.  
More generally, we will 
show that 
\beq\label{fkpk59}
\frac{\p F_g^\varphi}{\p V_k} \= V_1^{2g-2-k} \, 
P_{g,k}\Bigl(\frac{V_2}{V_1^2}, \frac{V_{3}}{V_1^{3}},\dots;  l_1(V), l_2(V), \dots \Bigr) \quad (0\leq k\leq 3g-2)
\eeq
for some polynomials $P_{g,k}$ in $\R^{[3g-2-k]}$. Notice that this 
formula, unlike~\eqref{fp59}, is true also in genus~1, with 
\beq\label{pini59}
P_{1,0}\=\frac{\ell_1}{16}\,, \quad P_{1,1}\=\frac1{24} \,,
\eeq
as we see immediately from equation~\eqref{genus1K1}.

Let us introduce further some notations. 
First, we can show that $\p^k(W)$ for every $k\ge1$ is $\frac{V_1^k}{\sqrt{\varphi'(V) \Delta}}$ times a polynomial 
in $\frac1y=\frac{\varphi'(V)}{\Delta}$, namely, 
\begin{align}
\p^k (W) \= \frac{V_1^k}{\sqrt{\varphi'(V) \Delta}} \, 
\sum_{n=1}^k M_{k,n}\Bigl(\frac{V_2}{V_1^2}, \dots, \frac{V_{k+1-n}}{V_1^{k+1-n}}; l_1(V),\dots, l_{k-n}(V)\Bigr) \, y^{-n} \,,
\label{dkw420}
\end{align}
where $M_{k,n}(w_1,\dots,w_{k-n}; \ell_1,\dots, \ell_{k-n})\in \R^{[k-n]}$
with $M_{1,n}=\delta_{n,1}/2$.
Second, for $k\ge0$, we define $Y_k=\sum_{j=0}^k \binom{k+1}{j+1} \p^j(W) \p^{k-j}(W)$. 
Then we have that there exist 
$$Y_{k,n}(w_1,\dots,w_{k-1}; \ell_1,\dots, \ell_n)\in \R^{[n]}\,,$$ 
such that 
\beq\label{Ykexpr422}
Y_k \= \frac{V_1^k}{\varphi'(V) \, \Delta} \, 
\sum_{n\ge0} \, Y_{k,n} \Bigl(\frac{V_2}{V_1^2}, \dots, \frac{V_{k}}{V_1^{k}}; l_1(V),\dots,l_{n}(V)\Bigr) \, y^{n-k} 
\eeq
with $Y_{k,0}=(2k+1)!!/2^{k}$, $k\ge0$.
Third, for $k\ge1$ we have
\beq 
\p^k\Bigl(\frac1{\Delta^2}\Bigr) \,=:\, \frac{V_1^k}{\varphi'(V) \Delta} \, \sum_{m=3}^{k+2} \, 
Q_{k,m}\Bigl(\frac{V_2}{V_1^2}, \dots, \frac{V_{k+3-m}}{V_1^{k+3-m}}; l_1(V),\dots,l_{k+2-m}(V)\Bigr) \,  y^{1-m} \,, \label{Qkexpr56}
\eeq
for some  
$Q_{k,m}=Q_{k,m}(w_1,\dots,w_{k+2-m}; \ell_1,\dots, \ell_{k+2-m})\in\P^{[k+2-m]}$ with $Q_{1,3}=2$.

We also introduce the first-order 
differential operators $D_k^{[m]} : \R \to \R$ by 
$$D_k^{[m]} \,:=\, D_k \+ \delta_{k,1}\,m \,, $$ 
where $D_k$ is the derivation defined by 
 $$  D_0 \= \sum_{i\ge1} \ell_{i+1} \frac{\p }{\p \ell_i}\,, \quad D_1 \= - \sum_{i\ge1} (i+1) w_i \frac{\p }{\p w_i} \,, \quad 
 D_k \=  \frac{\p }{\p w_{k-1}} ~(k\ge2)\,. $$  
The operators $D_k$, $k\ge0$, simply correspond to 
$V_1^k \frac{\p }{ \p V_k}$ when applied to the function  $P(V_2/V_1^2, V_3/V_1^3, \dots; l_1(V), l_2(V),\dots)$.

\begin{theorem}\label{structureFgphi719}
The functions $F_g^\varphi$ and $\p F^\varphi_g/\p V_k$ are given by 
equations~\eqref{fp59}  
and~\eqref{fkpk59}, where the polynomials $P_{g,k}\in\R^{[3g-2-k]}$ $(0\leq k\leq 3g-2)$ are defined by the initial values~\eqref{pini59} and the recursion
\begin{align}
	&  \frac{(2k+1)!!}{2^k} \, P_{g,k} \= \frac{1}{16} \, \sum_{\ell=0}^{3g-5} \,
Q_{\ell+2,k+1} \, P_{g-1,\ell} - \sum_{j=k+1}^{3g-2} \, Y_{j,j-k} \, P_{g,j} \label{loop57}\\
	& + \, \frac{1}{2}\,\sum_{k_1, k_2=0}^{3g-5} \, \sum_{\substack{ 1\le n_1\le k_1+1 \\ 1\le n_2\le k_2+1 \\ n_1+n_2=k}}
M_{k_1+1,n_1} \, M_{k_2+1,n_2}\, 
\biggl(
D^{[2g-4-k_1]}_{k_2} (P_{g-1,k_1})  
	\+ \sum_{m=1}^{g-1} P_{m,k_1} P_{g-m,k_2}\biggr)  \nn
\end{align}
for $g\ge2$, and the polynomial $P_g \in \R^{[3g-3]}$ is defined by 
\beq\label{PS56}
P_g \= \frac1{2g-2} \, \biggl(P_{g,1} + \sum_{k\geq2} \, k \, w_{k-1} \, P_{g,k} \biggr)  \qquad  (g\ge 2) \,.  
\eeq  
Moreover, the polynomials $P_g$ and $P_{g,k}$ are related by
\beq\label{new116}
P_{g,k} \= \bigl(D_k+(2g-2)\delta_{k,1}\bigr) \, P_g \qquad (g\ge2)\,.
\eeq
\end{theorem}

\begin{proof}
Substituting \eqref{dkw420}--\eqref{Qkexpr56} in the loop equation~\eqref{loopequationfinalfg11108}, and  
 comparing the coefficients of powers of~$y$ we obtain~\eqref{loop57}. 
Dividing~\eqref{homogfvpg} by $V_1^{2g-2}$, we find that for $g\ge2$ the polynomials $P_g$ can be constructed by $P_{g,k}$ by~\eqref{PS56}.
\end{proof}

\begin{remark}
Formulas~\eqref{pini59} and~\eqref{loop57} 
define $P_{g,k}$ for all~$g$ and~$k$ explicitly by induction. But the fact that these polynomials and the polynomial
$P_g$ defined by~\eqref{PS56} are related by~\eqref{new116} is not at all obvious from this
definition.  This corresponds to our previous Remark~\ref{remarkloopgradient} about the existence of a solution to the 
loop equation.
\end{remark}

\begin{remark}
Note that the functions $M_k(V_0,\dots,V_k)$ given in Lemma~\ref{vVmaplemma} for $k\ge1$ have the following more accurate form
\beq
\frac{M_k (V_0,\dots,V_k)}{ \varphi'(V)^{\frac{k}2+1} V_1^k} \= N_k\biggl(\frac{V_2}{V_1^2},\dots,\frac{V_k}{V_1^k}; l_1(V),\dots,l_{k-1}(V)\biggr)\,,
\eeq
where $N_k=N_k(w_1,\dots,w_{k-1}; \ell_1,\dots, \ell_{k-1})\in\R^{[k-1]}$ with $N_1=1$. 
Then by using \eqref{fgphi1001} and~\eqref{WK3gminus2}
we get the expressions of the polynomials $P_g$, $g\ge2$, from $F_g^{\rm WK}$.
For the reader's convenience, let us provide here the expressions for the first few $N_k$, $k\ge1$:
\beq
N_1 \=  1\,, \quad N_2\= 2 \, \ell_1 \+ w_1\,, \quad 
N_3 \= \frac{25}{4} \, \ell_1^2 \+ \frac{5}{2} \, \ell_2 \+ \frac{15}{2} \, \ell_1 \, w_1 \+ w_2\,.
\eeq
\end{remark}

\section{Review of Hamiltonian and bihamiltonian evolutionary PDEs}\label{newnewsection7}
In this section we review basic terminologies about evolutionary PDEs and Poisson structures, referring
 to~\cite{Du06, DLZ06, DZ-norm} for more details.  The evolutionary PDEs considered in this paper are always 
in (1+1) dimensions, meaning that the unknown functions have one space variable and one time variable, and
also, the number of unknown functions will be one. (We note that most of the terminology reviewed in this 
section and the classification projects in this section and the next can be generalized without much difficulty 
to the case when there are several unknown functions, often called the 
multicomponent case; we also note that in the last section of this paper, a more general situation
is actually briefly discussed.).    For readers from other areas, we also recall that  ``evolutionary" simply means that 
the PDE expresses the partial derivative of the unknown function with respect to 
the time variable as a function of the partial derivatives with respect to the space variable.  

Let $\mathcal{A}_{U}=\mathcal{S}(U)[U_1,U_2,\cdots]$ be the differential polynomial ring of~$U$, where
$\mathcal{S}(U)$ is some suitable
ring of functions on~$U$. For instance, $\mathcal{S}(U)$ could be $\mathcal{O}_c(U)$, the ring of 
power series in~$U-c$ for some constant~$c$ (we often take $c=0$).
Let $\p:=\sum_{m\ge0} U_{m+1} \p / \p U_m$ 
be a derivation.
When $U$ is taken as a function of~$X$, we identify $U_m$ with $\p^m U/\p X^m$, $m\ge0$, and $\p$ with $\p/\p X$.
Define a gradation $\deg$ on $\mathcal{A}_{U}$ by the degree assignments 
$\deg U_m=m$ ($m\ge1$), and we use $\mathcal{A}_{U}^{[k]}$ to denote 
the set of elements in $\mathcal{A}_{U}$ that are graded homogeneous of degree~$k$ with respect to~$\deg$. For $\ell\in\ZZ$, we also 
denote 
\beq
\mathcal{A}_{U}[[\e]]_{\rm \ell}=\{a\in \mathcal{A}_{U}[[\e]] \,|\, {\rm gr} \, a = \ell \, a\}\,,
\eeq 
where  
$${\rm gr} \= - \e \, \frac{\p}{\p \e} \+ \sum_{m\ge1} \, m\, U_m \, \frac{\p}{\p U_m} \,.$$ 
An element $M$ in $\mathcal{A}_{U}[[\e]]_{\rm 0}$ can be written in the form $M = \sum_{k\ge0} M^{[k]}(U,U_1,\dots,U_k) \e^k$, where $M^{[k]}(U,U_1,\dots,U_k) \in \mathcal{A}_{U}^{[k]}$, $k\ge0$.

A derivation $D:\mathcal{A}_{U}[[\e]] \to \mathcal{A}_{U}[[\e]]$ is called {\it admissible} if 
it commutes with~$\p$ and~$\e$. 
Following Dubrovin and Novikov~\cite{Du91, DN83, DN89, N92}, we call an admissible derivation $D:\mathcal{A}_{U}[[\e]] \to \mathcal{A}_{U}[[\e]]$ 
a {\it derivation of hydrodynamic type} if $D(U)$ has the form
\beq\label{hydropde1120abstract}
D(U) \= S(U) \, U_1 \,, \quad  S(U) \in \mathcal{S}(U)\,.  
\eeq
If we replace~$D$ by $\p / \p T$ and think of~$U$ as a function of $X$ and~$T$, then 
we get a {\it one-component evolutionary PDE of hydrodynamic type}~\cite{Du91, DN83, DN89, N92}: 
\beq\label{hydropde1120}
\frac{\p U}{\p T} \= S(U) \, \frac{\p U}{\p X} \,.
\eeq
We sat that an admissible derivation $D:\mathcal{A}_{U}[[\e]] \to \mathcal{A}_{U}[[\e]]$ 
 is of the {\it Dubrovin--Zhang normal form} if $D(U)$ has the form
\beq\label{DZnormpde1120abstract}
D(U) \= \sum_{k\ge0} \e^k S_k(U,U_1,\dots,U_{k+1})\,, \quad S_k \in \mathcal{A}^{[k+1]}_U\,, \quad S_0(U,U_1) \= S(U) \, U_1\,.
\eeq
We also call~\eqref{DZnormpde1120abstract} a {\it perturbation} of~\eqref{hydropde1120abstract}, or say that \eqref{hydropde1120abstract} is 
the {\it dispersionless limit} (the $\epsilon\to 0$ limit) 
of~\eqref{DZnormpde1120abstract}. 
A derivation of the Dubrovin--Zhang normal form 
is called an {\it infinitesimal symmetry} of~\eqref{DZnormpde1120abstract} 
if it commutes with~\eqref{DZnormpde1120abstract}. 

For any $R(U)\in  \mathcal{S}(U)$, 
the derivation of hydrodynamic type~$D'$, specified by 
\beq\label{symm1121}
D'(U) \= R(U) \, U_1\,,
\eeq 
is an infinitesimal symmetry of~\eqref{hydropde1120abstract}. 
We call the following family of derivations of hydrodynamic type $D_S$, $S \in \mathcal{O}_c(U)$, defined by 
\beq\label{RHE1122}
D_S(U) \= S(U) \, U_1\,,
\eeq
the {\it abstract local RH hierarchy} (see~e.g.~\cite{Du06}), 
sometimes simply the {\it abstract RH hierarchy}. 
It is obvious that the derivations in~\eqref{RHE1122} pairwise commute. 
When we take the countable subfamily of derivations $(S_i(U)=U^i/i!)_{i\ge0}$  
and consider $U$ as a function of $X$ and $T_i$, $i\ge0$, then the 
equations $\p U/\p T_i = S_i(U) \p U/\p X$, $i\ge0$, are nothing but the RH hierarchy~\eqref{RHhierarchy}, where we 
 identify $X$ with~$T_0$ as we do before. 
In practice, some other interesting countable subfamily of derivations like $(\sin(ku))_{k\ge1}$, $(e^{ku})_{k\ge1}$, \dots, 
can also be taken and the resulting family of equations are an integrable hierarchy which we call 
the {\it chord RH hierarchy} (sometimes simply still the {\it RH hierarchy}). 

By a {\it perturbation of the abstract local RH hierarchy} \eqref{RHE1122}, we mean a family of 
derivations $D_S(U)$, $S\in \mathcal{O}_c(U)$, each being given by~\eqref{DZnormpde1120abstract}. We say that  
$D_S(U)$ is {\it integrable} if $D_{S_2} D_{S_1}(U) = D_{S_1} D_{S_2}(U)$, $\forall\, S_1,S_2\in \mathcal{O}_c(U)$.

Denote by $\int: \mathcal{A}_U \rightarrow \mathcal{A}_U/ \p \mathcal{A}_U$ the projection, which extends 
termwise to a projection on $\mathcal{A}_U[[\e]]_0$. 
Elements in $\mathcal{A}_U[[\e]]_0/ \p \mathcal{A}_U[[\e]]_{-1}=:\F$ are called {\it local functionals}. 
The {\it variational derivative} of a local function $\int h$ with respective to~$U$ is defined by 
\beq\label{defvard1122}
\frac{\delta \int h}{\delta U} \= \sum_{k\ge0} \, (-\p)^k \biggl(\frac{\p h}{\p U_k}\biggr) \,.
\eeq
Clearly, if $h\in \mathcal{A}_U[[\e]]_0$, then $\frac{\delta \int h}{\delta U}\in \mathcal{A}_U[[\e]]_0$.
Also, for $a\in\mathcal{A}_U[[\e]]_0$, it is known that $a \in \p \mathcal{A}_U[[\e]]_{-1}$  
if and only if the right-hand side of~\eqref{defvard1122} with $h$ replaced by~$a$ vanishes, 
so~\eqref{defvard1122} is well defined. 

Let $P$ be an operator of the form
\beq\label{expandP1122}
P \= \sum_{k\ge0} \e^k P^{[k]} \,, \quad P^{[k]} \= \sum_{j=0}^{k+1} A_{k,j} \, \p^j\,, \quad 
A_{k,j} \in \mathcal{A}_U^{[k+1-j]} \,.
\eeq 
Such an operator~$P$ defines a bracket $\{\,,\,\}_P: \F \times \F \rightarrow \F$ via 
\beq
\biggl\{\int F, \int G \biggr\}_P \= \int \, \frac{\delta F}{\delta u} \, P \biggl(\frac{\delta G}{\delta u}\biggr) \,, \quad \forall\, F,G\in \mathcal{A}_U[[\e]]_0\,.
\eeq
This bracket is obviously bilinear. We say that $\{\,,\,\}_P$ is {\it Poisson} if it is anti-symmetric and satisfies the Jacobi identity. 
We call that $P$ a {\it Poisson or hamiltonian operator} if $\{\,,\,\}_P$ is a Poisson bracket. An equivalent criterion of the operator~$P$ to be 
Poisson is that $$[P,P]\=0\,,$$ where $[\,,\,]$ denotes the Schouten--Nijenhuis bracket 
(see~e.g.~\cite{DZ-norm}). 
The part $P^{[0]}$ in~\eqref{expandP1122} is called the {\it dispersionless limit} of~$P$. Obviously, if $P$ is Poisson then $P^{[0]}$ is Poisson.
A Poisson operator like $P^{[0]}$ is called a {\it Poisson operator of hydrodynamic type}.  
According to Dubrovin and Novikov~\cite{DN83}, a Poisson operator of hydrodynamic type corresponds to 
a contravariant flat pseudo-Riemannian metric (true also for the multi-component case), i.e., $P^{[0]}$ must have the form
\beq
P^{[0]} \= g(U) \, \p_X \+ \frac12 \, g'(U) \, U_X\,,
\eeq 
where $g(U)$ is a contravariant metric (automatically flat for our one-component case).

We call~\eqref{DZnormpde1120abstract} a {\it hamiltonian derivation of the Dubrovin--Zhang normal form} 
if there exists a Poisson operator~$P$ and an element $h\in\mathcal{A}_U[[\e]]_0$, such that 
\beq\label{hampde1121}
D(U) \= P\biggl(\frac{\delta \int h}{\delta U}\biggr) \,.
\eeq 
We call~$\int h$ the {\it hamiltonian} of~\eqref{hampde1121} and~$h$ the {\it hamiltonian density}. 
A local functional $\int r$, $r\in\mathcal{A}_U[[\e]]_0$, is called a {\it Casimir} for the Poisson operator~$P$ if 
\beq
P\biggl(\frac{\delta \int r}{\delta U}\biggr) \= 0\,,
\eeq
with $r$ being called the {\it Casimir density}. 

We will call any transformation of the form
\beq\label{UW1121}
U \; \mapsto \; M \=  \sum_{k\ge0} \e^k M^{[k]}(U,U_1,\dots,U_k)\in\mathcal{A}_{U}[[\e]]_{\rm 0}\,, \quad 
M^{[0]}(U) \in \mathcal{S}(U)^\times\,,
\eeq
a {\it Miura-type transformation}. These transformations form a group, called the {\it Miura group}, 
which contains the local diffeomorphism group as a subgroup. For more details about Miura-type transformations see e.g.~\cite{DZ-norm, RJ}. 

A further extension is given by the {\it quasi-Miura transformations}
\beq\label{quasidef1121}
U \; \mapsto \; Q \=  \sum_{k\ge0} \e^k Q^{[k]}(U,U_1,\dots,U_{N_k})\,,
\eeq
where $Q^{[k]}(U,U_1,\dots,U_{N_k})$, $k\geq0$, are usually still required to have polynomial dependence in $U_2, \dots, U_{N_k}$ for some 
integers $N_k$ but are now allowed to have rational dependence in~$U_1$.

The class of derivations of the Dubrovin--Zhang normal form, the class of hamiltonian derivations of the Dubrovin--Zhang normal form and the 
class of Poisson operators are invariant under the Miura-type transformations~\cite{DZ-norm}. 
In particular, the hamiltonian derivation of the Dubrovin--Zhang normal form~\eqref{hampde1121} 
under the Miura-type transformation~\eqref{UW1121} 
transforms to the hamiltonian derivation of the Dubrovin--Zhang normal form given by 
\beq
D(M) \= \widetilde P\biggl(\frac{\delta \int h}{\delta M}\biggr) \,, 
\eeq
where $h$ is understood as an element in $\mathcal{A}_M[[\e^2]]_0$, and 
\begin{align}
&\widetilde P \= \sum_{k,\ell\ge0} \, (-1)^\ell \, \frac{\p M}{\p U_k} \circ \p^k \circ P \circ \p^\ell \circ \frac{\p M}{\p U_\ell} \, . \label{Poissontrans1121} 
\end{align}
To make notations compact we will often write $P$ as $P(U)$ and $\widetilde P$ as $P(M)$, when 
$U$ and $M$ are related by a Miura-type transformation. 

If we apply the quasi-Miura transformation~\eqref{quasidef1121} to~\eqref{hampde1121}, the Poisson operator 
still transforms under the rules \eqref{Poissontrans1121} 
but the resulting operator could have rational dependence in~$M_1$ 
and the variational derivative of the hamiltonian could have rational dependence in~$M_1$ (here the definition 
of the variational derivative is extended again with the same rule~\eqref{defvard1122}). 

Two Poisson operators $P_1,P_2$ are called {\it compatible} 
if an arbitrary linear combination of $P_1, P_2$ is a Poisson 
operator. When $P_1, P_2$ are compatible, 
we call $P_2+\lambda \, P_1$ the {\it Poisson pencil} associated to $P_1, P_2$. 
Following Dubrovin~\cite{Du96}, let us start with 
considering the dispersionless limit $P_2^{[0]}+\lambda \, P_1^{[0]}$.
Fix $P_2^{[0]}(U)+\lambda \, P_1^{[0]}(U)$ an arbitrary Poisson pencil of hydrodynamic type.
According to Dubrovin~\cite{Du96} (see also~\cite{DN83, DN89}), it 
corresponds to a flat pencil, that is, a pencil of flat contravariant pseudo-Riemannian metrics $g_2(U)+\lambda \, g_1(U)$. 
The associated {\it canonical coordinate} of the pencil $u=u(U)$ is defined by $u=g_2(U)/g_1(U)$. 
The flat pencil in the $u$-coordinate reads $u \, g(u) + \lambda \, g(u)$ with $g(u)=g_1(U) u'(U)^2$.
A Poisson pencil $P_2(U)+\lambda \, P_1(U)$ with the hydrodynamic limit being $P_2^{[0]}(U)+\lambda \, P_1^{[0]}(U)$ 
is then characterized by the 
so-called {\it central invariant}, denoted $c(u)$, defined in~\cite{DLZ06, LZ05, L02}. 
On one hand, two Poisson pencils having the same hydrodynamic limit $P_2^{[0]}(U)+\lambda \, P_1^{[0]}(U)$ (or say the same $g(u)$) 
are equivalent under Miura-type transformations if and only if they have the same central invariant~\cite{DLZ06}.  
On the other hand, it is 
shown in~\cite{CPS16, CPS18} (see also~\cite{LZ13} for the case $g(u)=u$) that, for any given function $c(u)$ 
there exits a Poisson pencil $P_2(U)+\lambda \, P_1(U)$, with the hydrodynamic limit 
$P_2^{[0]}(U)+\lambda \, P_1^{[0]}(U)$ and with the central invariant $c(u)$; the proof is based on a subtle computation 
of the bihamiltonian cohomology introduced in~\cite{DZ-norm} and developed in~\cite{DLZ06, LZ05}.
For example, 
the Poisson pencil corresponding to the pair 
\beq
(g(u),c(u))\=\Bigl(u,\frac1{24}\Bigr) 
\eeq 
can be obtained from the bihamiltonian structure discovered by Magri~\cite{M78} of the 
KdV equation~\eqref{kdvequation1002} (see~e.g.~\cite{DLZ06}). 

For an element $g(U)\in \mathcal{O}_c(U)$ that is not identically zero, 
the abstract local RH hierarchy~\eqref{RHE1122} can be written in the form:
\beq\label{hamRH1127}
D_S(U) \= \Bigl(g(U) \, \p + \frac12 \, g'(U) \, U_1 \Bigr)\biggl(\frac{\delta \int h^{[0]}_S}{\delta U}\biggr) \,,
\eeq
where $h^{[0]}_S$ is a solution to the following ODE
\beq\label{HS1127}
g(U) \, \bigl(h^{[0]}_S\bigr)'' \+ \frac12 \, g'(U) \, \bigl(h^{[0]}_S\bigr)' \= S(U) \,.
\eeq
Obviously, up to a trivial additive constant, the $h^{[0]}_S$ is unique up to the addition of a Casimir density for the Poisson 
operator $g(U) \, \p_X + \frac12 \, g'(U) \, U_X$.

Let us consider the hamiltonian perturbation of the 
abstract local RH hierarchy~\eqref{hamRH1127}:
\beq\label{ham115}
D_S(U) \= P(U) \biggl( \frac{\delta \int h_S}{\delta U} \biggr)\,, \quad S\in \mathcal{O}_c(U)\,,
\eeq
where $P(U)$ is a Poisson operator of the form~\eqref{expandP1122} with
$P^{[0]}(U)=g(U) \, \p_X + \frac12 \, g'(U) \, U_X$, and $h_S$ are hamiltonian densities of the form
\beq
h_S \= \sum_{k\ge0} \e^k h_S^{[k]}\,, \quad h_S^{[k]} \in \mathcal{A}_U^{[k]}\,,
\eeq 
with $h_S^{[0]}$ given by~\eqref{HS1127}. Here, we point out that our labellings of the derivations and of 
the hamiltonian densities 
for the abstract local RH hierarchy and for its perturbation are different from the one used in~\cite{Du06}. 

According to~\cite{DMS05, DZ-norm, Getzler02}, the Darboux theorem holds for the hamiltonian operator~$P(U)$, 
namely, there exists a Miura-type transformation $U \mapsto M$ of the form
\beq
M \= \sum_{k\ge1} \, \e^k \, M^{[k]}(U,U_1,\dots,U_k) \,, \quad M^{[0]}(U) \= \int_{U^*}^U \, \frac{1}{\sqrt{g(U')}} \, dU'\,,
\eeq
reducing~$P(U)$ to $P(M)=\p$.  

For simplicity, we shall consider in this paper that $h_S$ written in the $M$-coordinate are power series of~$\e^2$. 

Before proceeding we introduce some notations. 
A {\it partition} is a non-increasing infinite sequence of non-negative integers $\mu=(\mu_1,\mu_2,\dots)$. 
The number of non-zero components of~$\mu$ is called the {\it length} of~$\mu$, denoted by~$\ell(\mu)$. 
The sum $\sum_{i\ge1} \mu_i$ is called the {\it weight} of~$\mu$, denote by~$|\mu|$. 
The set of all partitions is denoted by $\mathcal{P}$, and the set of partitions of weight~$d$, $d\ge0$, is denoted by 
$\mathcal{P}_d$. If $\ell(\mu)>0$, 
we often write $\mu$ as $(\mu_1,\dots,\mu_{\ell(\mu)})$; 
otherwise, we write $\mu$ either as $(0)$ or as $(\,)$. 
Denote $\mu+1=(\mu_1+1,\dots,\mu_{\ell(\mu)}+1)$ if $\ell(\mu)>0$, and $(\,)+1=(\,)$ otherwise. 
We use ${\rm mult}_i(\mu)$ to denote the multiplicity of~$i$ in~$\mu$, $i\ge1$, 
and denote ${\rm mult}(\mu)!=\prod_{i=1}^\infty {\rm mult}_i(\mu)!$. 
For any sequence of indeterminates $(y_1, y_2, \dots)$, 
 $y_\mu:=\prod_{i=1}^{\ell(\mu)} y_{\mu_i}$ (clearly, $y_{(\,)}=1$).

S.-Q.~Liu and Y.~Zhang found (see also~\cite{BDGR20, Du06, DLYZ16}) that 
performing a {\it canonical} \cite{Du06} Miura-type transformation 
$M\mapsto w=M+\cdots$, 
that is, a Miura-type transformation 
keeping the Poisson operator~$\p$ invariant, yields the following unique 
standard form:
\beq\label{hamperturb1127}
D_S(w) \= \p \biggl( \frac{\delta \int h_S}{\delta w} \biggr) \,, 
\eeq
with the derivation $D_{M^{[0]}(U)}$ satisfying 
\beq\label{standardform1113}
D_{M^{[0]}(U)}(w) \= \p \biggl(\frac{\delta \int h_{M^{[0]}(U)}}{\delta w}\biggr) \,, 
\eeq
where
\begin{align}
& h_{M^{[0]}(U)} \= \frac{w^3}6-\frac{\e^2}{24} a_0(w) w_1^2 \+\sum_{g\ge2} \e^{2g}
\sum_{\lambda\in \mathcal{P}_{2g} \atop 
\ell(\lambda)>1, \, \lambda_1=\lambda_2} 
\alpha_{\lambda}(w) w_\lambda \,. \label{hstand1128}
\end{align}
Here $a_0(w)$ and $\alpha_\lambda(w)$  with $\lambda\in \mathcal{P}_{2g}$ ($g\ge 2$), 
$\ell(\lambda)>1$, $\lambda_1=\lambda_2$,  are functions of~$w$.
\begin{remark}\label{standardformremark}
It was conjectured by S.-Q.~Liu and Y.~Zhang that if one imposes the integrability to the above standard form~\eqref{hstand1128}, then  
the functions $\alpha_\lambda(w)$ appearing in~\eqref{hstand1128}  
 are uniquely determined by $\alpha_{2^m}(w)$, $m\ge2$, and the functions $\alpha_{2^m}(w)$, $m\ge2$,  are free
 functional parameters.
In~\cite{DLYZ16} (see also~\cite{BDGR20}) 
it is indicated that if one further imposes a symmetry condition~\cite{DLYZ16, DZ-norm} for the 
hamiltonian densities (so-called {\it $\tau$-symmetry}) then the functional parameters 
$a_0(w)$ and $\alpha_\lambda(w)$ appearing in~\eqref{hstand1128} all become constants. 
In this paper, we will impose a new condition, as already mentioned in Introduction, i.e., to require the hamiltonian 
 system to possess a $\tau$-structure (see the next section for the details).
\end{remark}

In~\cite{Du06}, B.~Dubrovin considers the bihamiltonian test for the integrable hamiltonian perturbation and
 obtained the following theorem.
 
\vspace{3mm}
 
\noindent {\bf Theorem A} (Dubrovin~\cite{Du06})  {\it For $a_0(w),q(w), q'(w)$ all not identically~0, let 
 $P_1$, $P_2$ be the Poisson operators of the form: $P_1=\p$, and 
\beq\label{duthm-2}
P_2(w) \= q(w) \, \p \+ \frac12 \, q'(w) \, w_1 \+ \cdots \,,
\eeq 
where ``$\cdots$" contains higher order terms in~$\e$.
The two commutativity properties
\beq\label{duthm-1}
\biggl\{\int h_{S_1}, \int h_{S_2} \biggr\}_{P_1}\= O(\e^6)\,, \,\; \biggl\{\int h_{S_1}, \int h_{S_2} \biggr\}_{P_2} \= O(\e^6)\,, \quad \forall \, S_1,\, S_2\,,
\eeq
hold if and only if}
\beq\label{duthm-3}
\alpha_{2^2}(w) \= \frac{a_0(w)^2}{960} \biggl(5\,\frac{a_0'(w)}{a_0(w)} - \frac{q''(w)}{q'(w)} \biggr) \,.
\eeq

The explicit expression for the $\e^2$-term in $P_2(w)$ is given in the Appendix of~\cite{Du06}.

We have the following proposition. 
\begin{prop}\label{centralprop1202}
The central invariant for the pencil $P_2+\lambda P_1$ in Theorem~A is 
\beq\label{centalcu1203}
c(u) \= \frac{1}{24} \, \frac{a_0(q^{-1}(u))}{q'(q^{-1}(u))}\,,
\eeq
where $u=q(w)$ is the canonical coordinate for the pencil $P_2(w)+\lambda P_1(w)$.
\end{prop}

\begin{proof}
Let us perform the following Miura-type transformation:
\beq
u \= q(w) \,.
\eeq
The Poisson operators $P_1,P_2$ in the $u$-coordinate read:
\begin{align}
& P_1(u) \= \frac12 \, q'(q^{-1}(u))^2 \circ \p \+ \frac12 \, \p \circ q'(q^{-1}(u))^2 \,, \\
& P_2(u) \= \frac12 \, u \, q'(q^{-1}(u))^2 \circ \p \+ \frac12 \, \p \circ u \circ q'(q^{-1}(u))^2 \+  \cdots \,.
\end{align}
Using \cite[formula~(1.49)]{DLZ06} and using \cite[Appendix]{Du06} we obtain the expression~\eqref{centalcu1203} 
of the central invariant $c(u)$.
The proposition is proved.
\end{proof}

\section{Hamiltonian and bihamiltonian perturbations possessing a $\tau$-structure} \label{newnewsection8}
Driven by topological field theories and the 
 Witten--Kontsevich theorem (see~\cite{DVV91, DW90, Du96, DZ-norm, Ko92, Wi91}), 
the {\it $\tau$-structure} for the KdV hierarchy (see~\cite{BDY16, Dickey, DYZ21, DZ-norm}) 
becomes an important notion in the theory of integrable systems.  It still makes sense to speak of a 
$\tau$-structure for more general evolutionary systems (we will give a precise definition in a moment).
One of our main objects for the rest of the paper is to 
give conjectural classifications of hamiltonian and of bihamiltonian systems 
possessing a $\tau$-structure with the help of the group~$\G$.

It can be shown (see~e.g.~\cite{BDY16, Dickey, DW90, DYZ21, DZ-norm, Ko92, Wi91}) 
that there exist unique elements $\Omega^{\scriptscriptstyle \rm KdV}_{i,j}\in \mathcal{A}_u[\e^2]$, $i,j\ge0$, such that  
\beq
\e^2 \, \frac{\p^2 \log \F^{\scriptscriptstyle\rm WK}(\bt;\e)}{\p t_i \p t_j} 
\= \Omega^{\scriptscriptstyle \rm KdV}_{i,j}\big|_{u_k \mapsto \p_x^k(u^{\scriptscriptstyle \rm WK}(\bt;\e)), \, k\ge 0}  \,, \quad i,j\ge0\,.
\eeq
For instance, 
\beq
\Omega^{\scriptscriptstyle \rm KdV}_{0,0} \= u \,, \quad \Omega^{\scriptscriptstyle \rm KdV}_{0,1} \= \frac{u^2}2 \+ \e^2 \, \frac{u_2}{12}\,, 
\quad \Omega^{\scriptscriptstyle \rm KdV}_{1,1} \= \frac{u^3}3 \+ \e^2 \biggl(\frac{u u_2}6 \+ \frac{u_1^2}{24}\biggr) \+ \e^4 \, \frac{u_4}{144}\,.
\eeq
The polynomials $(\Omega^{\scriptscriptstyle \rm KdV}_{i,j})_{i,j\ge0}$ form a $\tau$-structure for the KdV hierarchy~\eqref{kdvflows26} and hence defines $\tau$-functions (see~\cite{BDY16, Dickey, DYZ21, DZ-norm, Sato}). 

Constructively, $\Omega^{\scriptscriptstyle \rm KdV}_{i,j}$ can be obtained in the following way. 
Introduce
\beq\label{defU813kdv1223}
u^{\scriptscriptstyle \rm WK}(\bt;\e) \,:=\, \e^2 \, \frac{\p^2 \F^{\scriptscriptstyle\rm WK}(\bt;\e)}{\p x^2} 
\= E(\bt) \+ \sum_{g\ge1} \, \e^{2g} \, \frac{\p \F^{\scriptscriptstyle\rm WK}_g(\bt)}{\p x^2} \,. 
\eeq
Here $x=t_0$. From~\eqref{WK3gminus2} we know that 
it leads to the quasi-Miura transformation 
\beq\label{qM92WK1}
v \;\mapsto\; u \= v \+ \sum_{g\ge1} \, \e^{2g} \, \p^2 \bigl(F^{\scriptscriptstyle\rm WK}_g\bigr)\,,
\eeq
which transforms the abstract local RH hierarchy $D_S(v) = S(v) \, v_1$ to
\begin{align}
 D_S(u) \= &
S(u) \, u_1 \+   
 \biggl(\frac{S'(u)}{12} \, u_3 
 + \frac{S''(u) }6 \, u_1 u_2 +  \frac{S'''(u)}{24} \, u_1^3 \biggr) \epsilon^2 \label{localkdvhierarchy1223}\\
& + \biggl(
\frac{S''(u)}{240} \, u_5+ \frac{S^{(3)}(u)}{80} \, u_4 u_1 + \frac{S^{(3)}(u)}{48} \, u_3 u_2 + \frac{23 \, S^{(4)}(u) }{1440} \, u_3  u_1^2 \nn\\
& \qquad +\frac{31 \, S^{(4)}(u) }{1440} \, u_2^2 u_1+\frac{S^{(5)}(u)}{90} \, u_2 u_1^3  + \frac{S^{(6)}(u)}{1152} \, u_1^5
\biggr) \, \epsilon^4
 \+ \cdots \,. \nn
\end{align}
By the Witten--Kontsevich theorem, when $S(u)=u^i/i!$ ($i\ge0$), 
equations~\eqref{localkdvhierarchy1223} are the abstract KdV hierarchy (see~\eqref{kdvflows26} and~\cite{BDY16, Dickey, DYZ21, DZ-norm}).
In general, $D_S$ commutes with $D_{u^i/i!}$, $i\ge0$. We call~\eqref{localkdvhierarchy1223} the {\it abstract local KdV hierarchy}.
Note that 
\beq
\e^2 \, \frac{\p^2 \log \F^{\scriptscriptstyle\rm WK}(\bt;\e)}{\p t_i \p t_j} 
\= \frac{E(\bt)^{i+j+1}}{i! \, j! \, (i+j+1)} \+ \sum_{g\ge1} \e^{2g} \, \frac{\p^2 \log \F_g^{\scriptscriptstyle\rm WK}(\bt)}{\p t_i \p t_j} \,, \label{omegaij114}
\eeq
where we used~\eqref{defFZ813} and~\eqref{twopointgenus0formula}. 
By~\eqref{WK3gminus2} we know that the right-hand side of~\eqref{omegaij114} can be represented by the jets 
$v,v_1,v_2$, $\cdots$. Substituting the inverse of the quasi-Miura transformation into the jet representation of 
the right-hand side of~\eqref{omegaij114} we get $\Omega^{\scriptscriptstyle \rm KdV}_{i,j}$.

As in~\cite{DYZ21} (cf.~\cite{BDY16, DZ-norm}), we say that a perturbation of the abstract local RH hierarchy (see~\eqref{DZnormpde1120abstract})
possesses a \hbox{{\it $\tau$-structure}} if there exist $\Omega_{S_1,S_2} \in \mathcal{A}_U[[\e]]_0$, $S_1,S_2 \in \mathcal{O}_c(U)$, such that $\Omega_{1,1}^{[0]} \in \mathcal{O}_c(U)^\times$ and 
\beq
\Omega_{S_1,S_2} \= \Omega_{S_2,S_1}\,, \quad D_{S_1} (\Omega_{S_2,S_3}) \= D_{S_2} (\Omega_{S_1,S_3})\,, \quad \forall\, 
S_1,S_2,S_3\in \mathcal{O}_c(U)\,.
\eeq
It can be easily verified that the existence of a $\tau$-structure implies integrability \cite{DYZ21} (cf.~also~\cite{DLYZ16}). 
More general setups for this principle are given in~\cite{VY}. 
We refer also to a related forthcoming joint work announced in~\cite{LYZZ22}. 

We are ready to give an axiomatic way to approach the class of (bi)-hamiltonian perturbations of 
the RH hierarchy possessing   
 a $\tau$-structure. 

\begin{lemma}\label{haminvtaulemma1203}
The class of hamiltonian perturbations of the RH hierarchy 
possessing a $\tau$-structure is invariant under Miura-type transformations.
\end{lemma}
\begin{proof}
We already know that the class of hamiltonian perturbations of the RH hierarchy is invariant under Miura-type transformations. 
It is also obvious that under a Miura-type transformation a $\tau$-structure remains a $\tau$-structure.
\end{proof}
  
Let us now consider hamiltonian perturbations of the RH hierarchy possessing a $\tau$-structure with a fixed choice of~$P^{[0]}$.
First we use Miura-type transformations reducing the consideration to the standard form~\eqref{hamperturb1127}--\eqref{hstand1128}.
Then with a help of a computer program we find that 
the requirement of existence of a $\tau$-structure implies that 
the functions 
$\alpha_{3^2}(w)$, $\alpha_{2^13^2}(w)$, $\alpha_{4^2}(w)$, 
$\alpha_{2^23^2}(w)$, $\alpha_{2^1 4^2}(w)$, $\alpha_{5^2}(w)$
are uniquely determined by $a_0(w)$, $\alpha_{2^2}(w)$, $\alpha_{2^3}(w)$, $\alpha_{2^4}(w)$, $\alpha_{2^5}(w)$ 
(agreeing with Remark~\ref{standardformremark}), 
and that the functions $\alpha_{2^2}(w)$, $\alpha_{2^3}(w)$, $\alpha_{2^4}(w)$ must have the expressions 
\begin{align}
& \alpha_{2^2} \= \frac{a_0 a_0' }{240} \+ q_1 a_0^3 \,, \label{a1class1202} 
\end{align}
\begin{align}
& \alpha_{2^3} \= \frac{31 a_0''' a_0^2}{96768} \+ \frac{527 q_1 a_0^3 a_0''}{1008} \+ \frac{1800 q_1^2 a_0^4 a_0'}{7} \+ \frac{499 q_1 a_0^2 a_0'^2}{336} 
\label{a2class1202}\\
& \qquad + \, \frac{23 a_0'^3}{45360} \+ \frac{1613 a_0 a_0' a_0''}{967680} \+ q_2 a_0^6 \,, \nn 
\end{align}
\begin{align}
& \alpha_{2^4} \= \frac{913 a_0''''' a_0^3}{46448640} \+ \frac{1795 q_1 a_0'''' a_0^4}{32256} \+ \frac{10357 q_1^2 a_0''' a_0^5}{168} \+ 25920 q_1^3 a_0^6 a_0''
 \label{a3class1202}\\
& \qquad  \+ \frac{167 q_2 a_0^6 a_0''}{105} \+ \frac{23087 q_1 a_0^3 a_0''^2}{40320} 
\+ 155520 q_1^3 a_0^5 a_0'^2 \+ \frac{12528 q_1 q_2 a_0^7 a_0'}{7} \nn\\
& \qquad + \, \frac{15635 q_1^2 a_0^3 a_0'^3}{14} \+  \frac{593 q_2 a_0^5 a_0'^2}{70} \+ \frac{20893 q_1 a_0 a_0'^4}{20160} \+ \frac{7733 a_0'''' a_0^2 a_0'}{30965760}\nn\\
& \qquad + \, \frac{212591 a_0''' a_0^2 a_0''}{464486400} \+ \frac{47953 q_1 a_0''' a_0^3 a_0'}{60480} \+ 
\frac{56519 a_0''' a_0 a_0'^2}{66355200} \+  \frac{48785 q_1^2 a_0^4 a_0' a_0''}{56} \nn\\
& \qquad + \, \frac{733 q_1 a_0^2 a_0'^2 a_0''}{224} \+ 
\frac{70229 a_0 a_0' a_0''^2}{58060800}  \+ \frac{1049357 a_0'^3 a_0''}{1393459200} \+ q_3 a_0^9 \,, \nn
\end{align}
where $q_1, q_2, q_3$ are arbitrary parameters (independent of~$w$) and where the arguments 
of the functions $a_0(w)$, $\alpha_{2^2}(w)$, $\alpha_{2^3}(w)$, $\alpha_{2^4}(w)$
 have been omitted. More generally, we expect that there are unique expressions giving 
 all $\alpha_\lambda$ 
 as polynomials in $a_0, a_0', a_0'', \thin \dots$ and constants $q_1$, $q_2, \thin \dots$, where 
 $q_i$ first appears linearly in~$\alpha_{2^{i+1}}(w)$. 
 This implies in particular that 
 if $a_0(w)$ is a constant function, then all of the $\alpha_\lambda(w)$ 
 are constants\footnote{This occurs, for example, in the presence of $\tau$-symmetry, 
 as shown in~\cite{BDGR18, BDGR20, DLYZ16}.}.
 Moreover, we expect that except for the term $q_i a_0^{3i}$ 
all terms in $\alpha_{2^{i+1}}$ contain higher derivatives of~$a_0$, so that when $a_0$
is a constant, then $\alpha_{2^{i+1}}$ is simply~$q_ia_0^{3i}$.
  
We continue to consider bihamiltonian perturbations of the (local) RH hierarchy 
possessing a $\tau$-structure. Of course, this class of perturbations is again invariant under 
Miura-type transformations (see~Lemma~\ref{haminvtaulemma1203}).
We reduce the considerations to the standard form as above, and the bihamiltonian axiom will further impose restrictions on~$q_i$'s. 
Note that in Theorem~A, B.~Dubrovin already did the bihamiltonian test 
for integrable hamiltonian perturbations (hamiltonian perturbations with a $\tau$-structure belong to this class)
up to order~$4$ in~$\e$. So by using~\eqref{a1class1202} and by using formula~\eqref{duthm-3} of Theorem~A,  we find 
\beq\label{qw1202}
q(w) \= \left\{\begin{array}{cc} 
C_1\int^w_{w^*} \, a_0(w') \, dw' \+ C_2 \,, & q_1 \= 0 \,,\\
\\
C_1 \frac{1-\exp\bigl(-960 \, q_1 \, \int^w_{w^*} \, a_0(w') \, dw'\bigr)}{960 \, q_1} \+ C_2\,, & q_1 \, \neq \, 0 \,, \\
\end{array}\right.
\eeq
where $C_1, C_2$ are arbitrary constants (that can depend on~$q_i$'s) and $C_1\neq0$. Continuing Dubrovin's bihamiltonian test, up to the order~8 in~$\e$, we find that 
\beq\label{q2q31129}
q_2 \= \frac{6400}{3} \, q_1^3 \,, \quad q_3 \= 0 \,.
\eeq
We expect that $q_4, q_5$,
$\cdots$ are also determined by~$q_1$ and $a_0(w)$. Note that since 
$q_2, q_3$ do not depend on $a_0(w)$, we can further expect this to be true for $q_4, q_5$, $\cdots$; 
with this consideration, we can restrict to the simple case $a_0(w)\equiv 1$ and the corresponding 
bihamiltonian test allows us to compute two more values:
\beq
q_4 \= -\frac{36805017600000}{77} \, q_1^7\,,  \quad q_5 \= -\frac{45612552683520000000}{7007} \, q_1^9\,.
\eeq

We have the following proposition. 
\begin{prop}
The central invariant for the pencil $P_2+\lambda P_1$ is given by 
\beq
c(u) \= \left\{\begin{array}{cc} 
\frac1{24 \, C_1}\,, & q_1 \= 0 \,,\\
\\
\frac1{24 \, (C_1-960 \, q_1 \, (u-C_2))}\,, & q_1 \, \neq \, 0 \,. \\
\end{array}\right.
\eeq
\end{prop}
\begin{proof}
By using Proposition~\ref{centralprop1202} and the expression~\eqref{qw1202}.
\end{proof}

Using equations~\eqref{a1class1202}, \eqref{a2class1202}, \eqref{q2q31129}, the above proposition, 
\cite[Theorem~1.8]{DLZ06} and 
\cite[Theorems~1 and~2]{LZ05}, we arrive at
\begin{theorem}\label{theorembiham43}
A bihamiltonian perturbation of the RH hierarchy with the central invariant $c(u)\not\equiv0$
admits a $\tau$-structure up to the $\epsilon^8$ approximation if and only if 
$1/c(u)$ is an affine-linear function of~$u$. 
\end{theorem}
The following theorem
will be proved in the next section.
\begin{theorem}\label{theoremnew1}
The statement in Theorem~\ref{theorembiham43} holds to all orders in~$\e$.
\end{theorem}
 
Notice that when there is a(n) (approximated) bihamiltonian structure, there is a choice of the associated 
Poisson pencil. Namely, consider the following change of the choice of Poisson pencil:
\beq
\widetilde P_1 \,:=\, c \, P_2 \+ d \, P_1\,, \quad  \widetilde P_2 \,:=\, a \, P_2 \+ b \, P_1\,, \quad ad-bc\neq 0 \, . 
\eeq
Here $a,b,c,d\in\CC$ are constants.
The canonical coordinate of $\widetilde P_1, \widetilde P_2$, denoted $\tilde u$, is related to~$u$ by
$$
\tilde u \= \frac{a \, u+b}{c \, u+d} \, .
$$
The pair of functions $(\tilde g, \tilde c)$ that characterizes the pencil $\widetilde P_2 + \lambda \widetilde P_1$ are given by
\begin{align}
& \tilde g (\tilde u) \= \frac{(ad-bc)^2}{(c \, u+d)^3} \, g(u) \,, \\
& \tilde{c}(\tilde u) \= \frac{c \, u+d}{ad-bc} \, c(u)\,. \label{tildec1204}
\end{align}
In particular, formula~\eqref{tildec1204} was obtained in~\cite{DLZ06}.
So, if the central invariant $c(u)$ of a Poisson pencil satisfies that $1/c(u)$ is an affine-linear function of~$u$, 
then it is always possible to choose properly the pencil so that the central invariant is $1/24$.

Hence the above Theorem~\ref{theoremnew1} can be more compactly reformulated as follows.

\smallskip

\noindent {\bf Theorem~10$'$.} {\it A bihamiltonian perturbation of the abstract local RH hierarchy possesses a $\tau$-structure if and only if 
under a proper choice of the associated Poisson pencil the central invariant is $1/24$.}

\smallskip

Here we recall again that according to~\cite{CPS16, CPS18, LZ13}, 
the existence of a bihamiltonian perturbation with central invariant $1/24$ is known (actually for arbitrary function $c(u)$ 
the existence is also known). 

\begin{remark}
Theorem~10$'$ was known for the case of the flat-exact Poisson pencils~\cite{DLZ18}, 
where the $\tau$-structure is associated to {\it $\tau$-symmetry}~\cite{DLYZ16, DZ-norm}. The flat-exact condition implies $g(u)=u$ which is a special case in our general consideration.
\end{remark}

\section{The WK mapping hierarchy and the WK mapping universality} \label{section6}
In this section, we introduce the hierarchy of equations associated to the 
WK mapping partition function, call it the {\it WK mapping hierarchy}, and 
prove it to be integrable and bihamiltonian with the central invariant 1/24. 
Then we propose and prove the WK mapping universality.

\subsection{The WK mapping hierarchy} 
For an arbitrary element $\varphi\in\G$, let $\F^\varphi(\bT;\e)$ be the WK mapping free energy. 
Introduce
\beq\label{defU813}
U^\varphi(\bT;\e) \,:=\, \e^2 \, \frac{\p^2 \F^\varphi(\bT;\e)}{\p X^2} 
\= E(\bT) \+ \sum_{g\ge1} \, \e^{2g} \, \frac{\p \F^\varphi_g(\bT)}{\p X^2} \,. 
\eeq
Here $X=T_0$, and we used 
\eqref{FZphig1001} and Theorem~\ref{thmgenus0}. 
From Proposition~\ref{jetreprep} we know that~\eqref{defU813} leads to a quasi-Miura 
transformation 
\beq\label{qM92phi2}
V \;\mapsto\; U^\varphi  \= V \+ \sum_{g\ge1} \, \e^{2g} \, \p^2 \bigl(F^\varphi_g\bigr)\,.
\eeq
It transforms the abstract local RH hierarchy $D_S(V) = S(V) \, V_1$ to 
\begin{align}
& D_S(U^\varphi) \= 
S \, U^\varphi_1 \+ \Biggl(\frac{S'}{12} \, U^\varphi_3 \+ \biggl(\frac{S''}6 \+ \frac{S'}8 \, \frac{\varphi''}{\varphi'}\biggr) 
\, U^\varphi_1 U^\varphi_2   \label{mappinghierarchy}\\
& \quad + \, \biggl(\frac{S'''}{24} \+ \frac{S''}{16} \, \frac{\varphi''}{\varphi'}  \+ \frac{S'}{16} \, 
\biggl(\frac{\varphi'''}{\varphi'}-\frac{\varphi''^2}{\varphi'^2}\biggr) \biggr)
\, \bigl(U^\varphi_1\bigr)^3 \Biggr) \, \e^2 \+ \cdots\,,\nn
\end{align}
which (by Proposition~\ref{jetreprep} and the Witten--Kontsevich theorem) is 
a perturbation of the abstract local KdV hierarchy~\eqref{localkdvhierarchy1223}. 
In~\eqref{mappinghierarchy}, $\varphi=\varphi(U^\varphi)$, $S=S(U^\varphi)$, 
and $\varphi^{(k)}=\varphi^{(k)}(U^\varphi)$, $S^{(k)}=S^{(k)}(U^\varphi)$ for $k\ge1$.
We call~\eqref{mappinghierarchy} the {\it abstract local WK mapping hierarchy associated to~$\varphi$}, 
for short the {\it abstract local WK mapping hierarchy}. 
In particular, $D_{U^\varphi}(U^\varphi)$ reads as follows:
\begin{align}
& D_{U^\varphi}(U^\varphi) \\
& = \; U^\varphi U^\varphi_1 \+\epsilon^2 \biggl(\frac{1}{12} \, U^\varphi_3 + \frac{\varphi''}{8 \, \varphi'} \, U^\varphi_1 U^\varphi_2 + 
\biggl(\frac{\varphi'''}{16 \, \varphi'}-\frac{\varphi''^2}{16 \, \varphi'^2}\biggr) \, (U^\varphi_1)^3 \biggr) \nn \\ 
&+ \, \epsilon^4 \biggl(
\frac{\varphi''}{480 \, \varphi'} \, U^\varphi_5   
+\biggl(\frac{7 \, \varphi'''}{480 \, \varphi'}
-\frac{11 \, \varphi''^2}{960 \, \varphi'^2}\biggr) \, U^\varphi_4 U^\varphi_1 + \biggl(\frac{ \varphi'''}{48 \, \varphi'} 
-\frac{\varphi''^2}{192 \varphi'^2}\biggr) \, U^\varphi_3 U^\varphi_2 
\nn\\
& \quad + \biggl(\frac{53 \, \varphi''^3}{2880 \varphi'^3} 
-\frac{11 \, \varphi''' \varphi''}{240 \, \varphi'^2} 
+\frac{9 \, \varphi^{(4)}}{320 \, \varphi'}
\biggr) \, U^\varphi_3 (U^\varphi_1)^2\nn\\
& \quad + \biggl(\frac{17 \, \varphi^{(4)}}{480 \, \varphi'}
-\frac{7 \, \varphi''^3}{1440 \, \varphi'^3}
-\frac{7 \, \varphi''' \varphi''}{240 \, \varphi'^2} \biggr) \,  (U^\varphi_2)^2 U^\varphi_1\nn\\
& \quad + \biggl(\frac{5 \, \varphi^{(5)}}{192 \, \varphi'}
+\frac{71 \, \varphi''^4}{1920 \, \varphi'^4}
-\frac{13 \, \varphi'''^2}{1920 \, \varphi'^2} -\frac{13 \, \varphi^{(4)} \varphi''}{320 \, \varphi'^2}
-\frac{\varphi''' \varphi''^2}{64 \, \varphi'^3}\biggr) \, U^\varphi_2 (U^\varphi_1)^3 \nn\\
& \quad + \biggl(\frac{\varphi^{(6)}}{384 \, \varphi'}
-\frac{23 \, \varphi''^5}{640 \, \varphi'^5}
-\frac{\varphi^{(5)} \varphi''}{192 \, \varphi'^2}
+\frac{\varphi''' \varphi^{(4)}}{768 \, \varphi'^2} \nn\\
& \quad \quad -\frac{11 \, \varphi^{(4)} \varphi''^2}{1920 \, \varphi'^3}
+\frac{11 \, \varphi''' \varphi''^3}{160 \, \varphi'^4}
-\frac{33 \, \varphi'''^2  \varphi''}{1280 \, \varphi'^3}\biggr) \, (U^\varphi_1)^5
\biggr) \+ \cdots\,. \nn
\end{align}
Alternatively, 
\begin{align}
& D_{U^\varphi}(U^\varphi) \= U^\varphi U^\varphi_1 \+ \biggl(\frac{1}{12} \, U^\varphi_3 
+ \frac{l_1}8 \, U^\varphi_1 U^\varphi_2 + 
\frac{l_2}{16} \, (U^\varphi_1)^3 \biggr) \epsilon^2 \label{firstflow63}\\ 
&+ \, \biggl(
\frac{l_1}{480} \, U^\varphi_5   
+\biggl(\frac{7 \, l_2}{480} + \frac{l_1^2}{320}\biggr) \, U^\varphi_4 U^\varphi_1 
+ \biggl(\frac{l_2}{48}+\frac{l_1^2}{64}\biggr) \, U^\varphi_3 U^\varphi_2 
\nn\\
& \quad + \biggl(\frac{9 \, l_3}{320} + \frac{37 \, l_1 l_2} {960} + \frac{l_1^3}{1440} \biggr) \, U^\varphi_3 (U^\varphi_1)^2\nn\\
& \quad + \biggl(\frac{17\, l_3}{480} + \frac{37 \, l_1 l_2} {480} + \frac{l_1^3}{720} \biggr) \,  (U^\varphi_2)^2 U^\varphi_1\nn\\
& \quad + \biggl(\frac{5 \, l_4}{192} + \frac{61 \, l_1 l_3}{960} + \frac{137 \, l_2^2}{1920} 
+ \frac{l_1^2 l_2}{192} \biggr) \, U^\varphi_2 (U^\varphi_1)^3 \nn\\
& \quad + \biggl(\frac{l_5}{384} + \frac{l_1 l_4}{128} + \frac{7 \, l_2 l_3}{256} + \frac{l_1^2 l_3}{1280} 
+ \frac{l_1 l_2^2}{640}  \biggr) \, (U^\varphi_1)^5
\biggr) \epsilon^4 \+ \cdots\,, \nn
\end{align}
where $l_k=l_k(U^\varphi)$ are defined in~\eqref{deflkv63}.
The abstract local WK mapping hierarchy~\eqref{mappinghierarchy} reduces to~\eqref{localkdvhierarchy1223} when $\varphi(V)=V$. 
Recalling that $S(V)\in \mathcal{O}_c(V)$, we note that for $c\ne0$ one should 
modify the infinite group $\G$ to $\G=V-c+(V-c)^2 R[[V-c]]$, which does not affect the previous formulations.
By construction, the power series $U^\varphi(\bT;\e)$ 
satisfies the following hierarchy of evolutionary PDEs: 
\beq
\frac{\p U^\varphi(\bT;\e)}{\p T_i} \= D_{(U^\varphi)^i/i!}(U^\varphi(\bT;\e))\,, \quad i\ge0\,,
\eeq
which we call the {\it mapping WK hierarchy associated to~$\varphi$}. 
Here on the right-hand side it is understood that one replaces $U^\varphi_k$ by $\p U^\varphi(\bT;\e)/ X^k$, $k\ge1$, with $X=T_0$.

\begin{remark}
Surprisingly, the rational numbers appearing on the right-hand side of~\eqref{firstflow63} are all {\it positive}. 
This may hold to all orders in~$\e$.
\end{remark}

{\it A priori} the coefficient of each power of~$\e^2$ on the right-hand side of~\eqref{mappinghierarchy} 
could be a polynomial in $(U_1^{\varphi})^{\pm1}$, $U_2^\varphi$, $U_3^\varphi$, \dots, 
but with the help of a general Mathematica package\footnote{The package is 
based on the method given in~\cite{Du08, DZ-norm, LZ13}.} 
designed by Joel Ekstrand, 
 we find that up to and including the term~$\e^{10}$ there are never any negative powers of~$U_1^\varphi$.
We will prove the following theorem.
\begin{theorem}\label{conj0307}
The abstract local WK mapping hierarchy~\eqref{mappinghierarchy} has {\it polynomiality}: for any $S$
the right-hand side of~\eqref{mappinghierarchy} belongs to $\mathcal{A}_{U^\varphi}[[\e^2]]_1$. 
\end{theorem}  

We first prove a special case 
of Theorem~\ref{conj0307}. 
\begin{prop}\label{localkdv63}
Theorem~\ref{conj0307} holds when $\varphi={\rm id}$, i.e., for the local KdV hierarchy~\eqref{localkdvhierarchy1223}. More precisely, the abstract local KdV hierarchy~\eqref{localkdvhierarchy1223} has the form:
\begin{align} 
D_S(u) & \= 
\p \biggl(\int^{u} S \+ \sum_{g\ge1} \e^{2g} \sum_{\lambda \in \mathcal{P}_{2g}} \, 
 K_\lambda \, S^{(\ell(\lambda)+g-1)}(u) \, u_\lambda\biggr)\label{mappinghierarchykdvdivergence} \\
 & \= S(u) \, u_1 \+ \sum_{g\ge1} \e^{2g} \sum_{\mu \in \mathcal{P}_{2g+1}} \, 
 G_\mu \, S^{(\ell(\mu)+g-1)}(u) \, u_\mu \,, \label{mappinghierarchykdv}
\end{align}
where $K_\lambda, G_\mu$  are rational numbers. 
\end{prop}
\begin{proof}
For $S(u)=u^i/i!$ ($i\ge0$), it is known that~\eqref{localkdvhierarchy1223}, i.e., the abstract KdV hierarchy, can be written as
\beq\label{abslokdvv2}
D_{u^i/i!}(u) \= \p(h_{i-1}(u,u_1,u_2,\dots,u_i))\,, \quad i\ge0\,,
\eeq
where $h_k=h_k(u,u_1,u_2,\dots,u_{k+1})\in \QQ[u,u_1,\dots,u_{k+1}][\e^2]_0$, $k\ge-1$, are 
$\tau$-symmetric hamiltonian densities for the KdV hierarchy \cite{DZ-norm} (see also~\cite{BDY16, DYZ21}), which satisfy 
\beq
h_k - \frac{u^{k+2}}{(k+2)!} \, \in \, \e^2 \cdot \QQ[u,u_1,\dots,u_{k+1}][\e^2]_{-2} ~(\forall\,k\ge-1), 
\eeq
\beq
 h_{-1} \= u\,, \quad D_{u^j/j!} (h_{i-1}) \= D_{u^i/i!} (h_{j-1}) ~ (\forall\,i,j\ge0)\,,
\eeq
and 
\beq\label{partialhipartialu}
\frac{\p h_i}{\p u} \= h_{i-1} \,, \quad i\ge0\,.
\eeq
For $S(u)=\sum_{m\ge0} a_m u^m/m!$ 
with $a_m$ being arbitrarily given constants, we have
\beq\label{DSusumsimple}
D_S(u) \= \sum_{m\ge0} a_m D_{u^m/m!}(u) \,.
\eeq
The expression~\eqref{mappinghierarchykdvdivergence} follows from~\eqref{partialhipartialu}, \eqref{DSusumsimple} and~\eqref{abslokdvv2}. 
For $S\in \mathcal{O}_c(u)$ with $c\neq 0$, the proof is then similar. Equation~\eqref{mappinghierarchykdv} follows from~\eqref{mappinghierarchykdvdivergence}.
\end{proof}

We note that assuming polynomiality the precise form~\eqref{mappinghierarchykdvdivergence} can also 
be obtained as a result of the quasi-trivial transformation combined with~\eqref{WK3gminus2}--\eqref{homowk-2}.

\begin{remark}
Let us briefly describe another way of defining the abstract local KdV hierarchy $D_S$, $S\in\mathcal{O}_c(u)$. 
Define $D_u$ as an admissible derivation such that $D_u(u)=u u_1 + \e^2 \frac{u_3}{12}$. 
Require $D_S$ to be the admissible derivation on $\mathcal{A}_u[\e^2]$ satisfying 
\beq\label{deflocalkdvunique}
D_S(u)-S(u) \, u_1 \, \in \, \e \cdot \mathcal{A}_u[\e]_{-2}\,, \quad  [D_S,D_u]\=0 \,.
\eeq
For the uniqueness of $D_S$ for any $S\in\mathcal{O}_c(u)$ see e.g.~\cite{Buryak, LZ06}. The existence of $D_{u^i/i!}$, $i\ge0$, is well known. 
For $S(u)=\sum_{m\ge0} a_m u^m/m!$, let $D_S(u)$ be assigned as the right-hand side of~\eqref{DSusumsimple}, then it can be checked that $D_S$ 
satisfies~\eqref{deflocalkdvunique}. For a general $S\in\mathcal{O}_c(u)$ the proof of existence is similar. 
\end{remark}

It is known (see e.g.~\cite{BDY16, Buryak, Dickey, DW90, DYZ21, DZ-norm, Ko92, Wi91}) that 
$\Omega^{\scriptscriptstyle \rm KdV}_{i,j}$, $i,j\ge0$, actually all belong to $\mathcal{A}_{u}[[\e^2]]_0$. Then by 
an argument similar to the proof of Proposition~\ref{localkdv63} we have
$\Omega_{S_1,S_2}^{\scriptscriptstyle \rm KdV} = \int^u S_1 S_2 + O(\e^2)
\in \mathcal{A}_{u}[[\e^2]]_0$, $\forall\,S_1,S_2 \in \mathcal{O}_c(u)$. 
Here $\Omega_{S_1,S_2}^{\scriptscriptstyle \rm KdV}$ is defined as the 
substitution of the inverse of the quasi-Miura type transformation~\eqref{qM92WK1}  
in~$\int^u S_1 S_2+ \sum_{g\ge1} \e^{2g} D_{S_1} D_{S_2} (F^{\scriptscriptstyle\rm WK}_g)$ 
(similar to the definition of $\Omega^{\scriptscriptstyle \rm KdV}_{i,j}$; see~\eqref{omegaij114}).

In order to prove Theorem~\ref{conj0307} for general $\varphi$, we will prove a 
stronger statement. (In Section~\ref{section8} we will give a more direct proof of a generalization 
of Theorem~\ref{conj0307}.)
Define two operators $P^\varphi_1$ and $P^\varphi_2$ by
\begin{align}
P^\varphi_1 \,:=\, 
\sum_{k,\ell\ge0} \, \frac{\p U^\varphi}{\p V_k} \circ \p^k 
\circ \biggl( \frac{1}{2\,\varphi'(V)} \circ \p \+ \p \circ 
\frac{1}{2\,\varphi'(V)} \biggr) \circ (-\p)^\ell \circ \frac{\p U^\varphi}{\p V_\ell} \,,\label{Poisson92-1} \\
P^\varphi_2 \,:=\, 
\sum_{k,\ell\ge0} \, \frac{\p U^\varphi}{\p V_k} \circ \p^k \circ \biggl(
\frac{\varphi(V)}{2\,\varphi'(V)} \circ \p \+ \p \circ \frac{\varphi(V)}{2\,\varphi'(V)} \biggr) \circ (-\p)^\ell \circ \frac{\p U^\varphi}{\p V_\ell} \,, \label{Poisson92-2}
\end{align}
where $U^\varphi$ is given by the quasi-Miura transformation~\eqref{qM92phi2}. It readily follows from the definition that 
the operators $P^\varphi_a$, $a=1,2$, have the form:
\begin{align}
& P^\varphi_a(U^\varphi) \= \sum_{g\ge0} \, \e^{2g} \, P^{\varphi,[g]}_a \,,\\
& P^{\varphi,[g]}_a \= \sum_{j=0}^{3g+1} \, A^\varphi_{2g,j;a} \, \p^j\,, 
\quad A_{2g,j;a}^{\varphi} \in \mathcal{O}_c(U^\varphi)\bigl[U^\varphi_1,\dots,U^\varphi_{3g+1},(U^\varphi_1)^{-1}\bigr]\,,\\
& P^{\varphi,[0]}_1  \= \frac{1}{2\,\varphi'(U^\varphi)} \circ \p \+ \p \circ \frac{1}{2\,\varphi'(U^\varphi)}\,, \label{explicitp11126}\\
& P^{\varphi,[0]}_2 \= \frac{\varphi(U^\varphi)}{2\,\varphi'(U^\varphi)} \circ \p \+ \p \circ \frac{\varphi(U^\varphi)}{2\,\varphi'(U^\varphi)}\,, 
\label{explicitp21126}\\
& \sum_{m\ge1} \, m \, U^\varphi_m \, \frac{\p A^\varphi_{2g,j;a} }{\p U^\varphi_m} \= (2g+1-j) A^\varphi_{2g,j;a}  \,.
\end{align}
We know that $[P^\varphi_a(U^\varphi), P^\varphi_b(U^\varphi)]=0$, for arbitrary $a,b\in\{1,2\}$.
The abstract local WK mapping hierarchy~\eqref{mappinghierarchy} can be written in the following form:
\beq\label{WKmappingham1126}
D_S(U^\varphi ) \= P^\varphi_1(U^\varphi) \biggl( \frac{\delta \int h^\varphi_{1;S}}{\delta U^\varphi} \biggr) 
\= P^\varphi_2(U^\varphi) \biggl(\frac{\delta \int h^\varphi_{2;S}}{\delta U^\varphi} \biggr) \,, \quad i\ge0\,.
\eeq
Here, the hamiltonian densities 
$h^\varphi_{1;S}$, $h^\varphi_{2;S}$ are understood as the substitutions of the inverse of 
the quasi-Miura transformation~\eqref{qM92phi2} into
\begin{align}
& h_{1;S}^{\varphi} \= \int_{0}^V \sqrt{\varphi'(x_2)} \int_0^{x_2} \, S(x_1) \, \sqrt{\varphi'(x_1)} \, dx_1 \, dx_2 \,, \label{h1s69}\\
& h_{2;S}^{\varphi} \= \int_{0}^V \sqrt{\frac{\varphi'(x_2)}{\varphi(x_2)}} \int_0^{x_2} \, S(x_1)\, \sqrt{\frac{\varphi'(x_1)}{\varphi(x_1)}} \, dx_1 \, dx_2 \,.
\label{h2s69}
\end{align}

{\it A priori}, the operators $P^\varphi_a(U^\varphi)$ and the variational derivatives of the 
hamiltonian densities $h^\varphi_{a;S}$ with respect to $U^\varphi$, 
$a=1,2$,
could contain negative powers of~$U^\varphi_1$, but just as in the remark preceding Theorem~\ref{conj0307}, 
we can use Ekstrand's Mathematica pakage to check that up to and including the $\e^8$ term 
this does not happen.
Explicit expressions 
for $P_1,P_2$ up to and including~$\e^2$ are given as follows:
\begin{align}
& P^\varphi_1(U^\varphi) \= \frac12 \, \frac{1}{\varphi'} \circ \p \+ \frac12 \, \p \circ \frac{1}{\varphi'} \label{p1explicit1126}\\
& + \, \frac12 \, \Biggl( \biggl(\frac{3 \varphi'''^2}{16 \varphi'^3}-\frac{\varphi^{(5)}}{24 \varphi'^2} +\frac{13 \varphi^{(4)} \varphi''}{48 \varphi'^3} 
-\frac{15 \varphi''' \varphi''^2}{16 \varphi'^4}+\frac{\varphi''^4}{2 \varphi'^5}\biggr) \, (U^\varphi_1)^3 \nn\\
& + \, \biggl(\frac{7 \varphi''' \varphi''}{8 \varphi'^3} -\frac{\varphi^{(4)}}{6 \varphi'^2}-\frac{3 \varphi''^3}{4 \varphi'^4} \biggr) \, U^\varphi_1 U^\varphi_2 
+ \biggl(\frac{ \varphi''^2}{6 \varphi'^3} -\frac{\varphi'''}{12 \varphi'^2}\biggr) \, U^\varphi_3 \nn\\ 
& + \, \biggl( \biggl(\frac{3 \varphi''' \varphi''}{8 \varphi'^3}-\frac{\varphi^{(4)}}{12 \varphi'^2}-\frac{\varphi''^3}{4 \varphi'^4}\biggr) \, (U^\varphi_1)^2 \+ \biggl(\frac{\varphi''^2}{3 \varphi'^3}-\frac{\varphi'''}{6 \varphi'^2}\biggr) \, U^\varphi_2 \biggr) \circ \p \Biggr) \e^2 \+ \cdots \nn
\end{align}
and
\begin{align}
& P^\varphi_2(U^\varphi) \= \frac12 \, \frac{\varphi}{\varphi'} \circ \p \+ \frac12\, \p \circ \frac{\varphi}{\varphi'} \label{p2explicit1126}\\
& + \, \frac{\e^2}2 \, \Biggl(\biggl(-\frac{\varphi \varphi^{(5)}}{24 \varphi'^2}
-\frac{\varphi^{(4)}}{8 \varphi'} +\frac{3 \varphi \varphi'''^2}{16 \varphi'^3} 
+\frac{\varphi \varphi''^4}{2 \varphi'^5}-\frac{\varphi''^3}{4 \varphi'^3} \nn\\
& \qquad\qquad +\frac{13 \varphi \varphi^{(4)} \varphi''}{48 \varphi'^3}-\frac{15 \varphi \varphi''' \varphi''^2}{16 \varphi'^4}
+\frac{19 \varphi''' \varphi''}{48 \varphi'^2} \biggr) \, (U^\varphi_1)^3 \nn\\
& \qquad +\, \biggl(-\frac{\varphi \varphi^{(4)}}{6 \varphi'^2} -\frac{\varphi'''}{3 \varphi'} 
-\frac{3 \varphi \varphi''^3}{4 \varphi'^4}+\frac{3 \varphi''^2}{8 \varphi'^2}
+\frac{7 \varphi \varphi''' \varphi''}{8 \varphi'^3}\biggr) \, U^\varphi_1 U^\varphi_2 \nn\\
& \qquad + \, \biggl( 
\frac{\varphi \varphi''^2}{6 \varphi'^3} - \frac{\varphi \varphi'''}{12 \varphi'^2} 
-\frac{\varphi''}{12 \varphi'}\biggr) \, U^\varphi_3 \nn\\
& \qquad + \, \biggl(
\biggl(-\frac{\varphi \varphi^{(4)}}{12 \varphi'^2} -\frac{\varphi'''}{6 \varphi'}-\frac{\varphi \varphi''^3}{4 \varphi'^4}
+\frac{\varphi''^2}{8 \varphi'^2}+\frac{3 \varphi \varphi''' \varphi''}{8 \varphi'^3}\biggr) \, (U^\varphi_1)^2\nn\\
& \qquad \qquad 
+ \, \biggl(\frac{\varphi \varphi''^2}{3 \varphi'^3}-\frac{\varphi \varphi'''}{6 \varphi'^2}-\frac{\varphi''}{6 \varphi'}\biggr) \, U^\varphi_2\biggr) \circ \p
+\frac{\p^3}{4} \Biggr) \+ \cdots \,. \nn
\end{align}

The following theorem, which is stronger than Theorem~\ref{conj0307}, gives a refined version of Theorem~\ref{thmmainshortversion}.
\begin{theorem}\label{mainconjecture} 
For $a=1,2$, $g\ge0$ and $0\leq j\leq 3g+1$, the elements $A^\varphi_{2g,j;a}$ all belong to $\mathcal{A}_{U^\varphi}^{[2g+1-j]}$. 
Moreover, the variational derivatives of the hamiltonians $\int h^\varphi_{1;S}$ and $\int h^\varphi_{2;S}$ 
with respect to $U^\varphi$ belong to $\mathcal{A}_{U^\varphi}[[\e]]$.
\end{theorem}
\begin{proof}
First of all, we have
\beq
\p \= \sum_{m\ge0} \frac{\p t_m}{\p X} D_{u^m/m!} \,.
\eeq
Here when we use a function of $u$, say $f(u)$, to label a derivation $D_{f(u)}$ we mean  
the corresponding derivation in the abstract local KdV hierarchy. By the definition~\eqref{defsigmaction} 
we can simplify the above equality to 
\beq
\p \= D_{\sqrt{\varphi'(\varphi^{-1}(u))}} \,.
\eeq
By Proposition~\ref{localkdv63} we know that $D_{\sqrt{\varphi'(\varphi^{-1}(u))}}(u)$ has polynomiality 
and of course it commutes with the KdV derivation $D_u(u)$. 
Then, using the results in~\cite{LWangZ}, we know that by taking $\p=\p_X$ as the spatial derivative 
 the abstract local KdV hierarchy is transformed to a bihamiltonian evolutionary system for~$u$ and 
 particularly the $\p_x$-flow of the abstract local KdV hierarchy after the transformation 
is bihamiltonian in Dubrovin--Zhang's normal form. 
Secondly, we note that 
\beq\label{Uu719}
U^{\varphi} \= \e^2 \, \frac{\p^2 \F^\varphi(\bT;\e)}{\p X^2} \= \e^2 \, 
\sum_{i,j\ge0} \frac{\p t_i}{\p X} \frac{\p t_j}{\p X} \, \Omega^{\scriptscriptstyle \rm KdV}_{i,j} 
\= \varphi^{-1}(u) \+ O(\e^2)\,.
\eeq
Since the $\p_x$-flow for $u$ with $\p=\p_X$ as the spatial derivative is an evolutionary PDE 
in Dubrovin--Zhang's normal form
and since $\Omega^{\scriptscriptstyle \rm KdV}_{i,j}\in \mathcal{A}_u[[\e^2]]_0$, we find by 
doing the substitution that $\Omega^{\scriptscriptstyle \rm KdV}_{i,j}$ are power series of $\e^2$ 
with coefficients being polynomials of $\p_X(u), \p^2_X(u)$, \dots.
This means that equation~\eqref{Uu719} gives a Miura-type transformation (with the spacial 
derivative being $\p$). The theorem is proved.
\end{proof}

We note that an equivalent description of the second statement of Theorem~\ref{mainconjecture} 
is that the hamiltonian densities $h^\varphi_{a;S}$, $a=1,2$, modulo certain total $\p$-derivatives,
both belong to $\mathcal{A}_{U^\varphi}[[\e]]$ for any~$S$.  
We also note that the first statement of Theorem~\ref{mainconjecture} implies 
in particular that $A^\varphi_{2g,j;a}=0$ for all $j\geq 2g+2$.

\begin{proof}[Proof of Theorem~\ref{conj0307}]
The theorem follows from Theorem~\ref{mainconjecture}.
\end{proof}

Since we have proved Theorem~\ref{conj0307}, by using the quasi-trivial transformation~\eqref{qM92phi2} with Theorem~\ref{structureFgphi719}, 
we can further prove that the abstract local WK mapping hierarchy~\eqref{mappinghierarchy} has the more precise form:
\beq\label{mappinghierarchydivergence}
D_S(U^\varphi) \= 
\p \Biggl(\int^{U^\varphi} S+\sum_{g\ge1} \e^{2g} \sum_{\lambda \in \mathcal{P}_{2g}} 
\sum_{j=1}^{\ell(\lambda)+g-1}  
Y^\varphi_{\lambda,j}(l_1(U^\varphi),\dots) \, S^{(j)}(U^\varphi) \, U^\varphi_\lambda\Biggr)\,,
\eeq
where 
$Y^\varphi_{\lambda,j}(\ell_1,\dots)\in \mathcal{R}^{[\ell(\lambda)+g-1-j]}$, and 
$l_k(U^\varphi)$ are defined in~\eqref{deflkv63}.

For $\varphi\in\G$, we have
\beq
\e^2 \, \frac{\p^2 \log \F^{\varphi}(\bT;\e)}{\p T_i \p T_j} 
\= \frac{E(\bT)^{i+j+1}}{i! \, j! \, (i+j+1)} \+ \sum_{g\ge1} \e^{2g} \, \frac{\p^2 \log \F_g^{\varphi}(\bT)}{\p T_i \p T_j} \,. \label{omegaijphi114}
\eeq
According to Proposition~\ref{jetreprep}, the right-hand side of~\eqref{omegaijphi114} can be represented by the jets $V, V_1, V_2$, $\dots$. 
Substituting the inverse of the quasi-Miura transformation~\eqref{qM92phi2} into the jet representation of the 
right-hand side of~\eqref{omegaijphi114} we get a power series of~$\e^2$, which we denote by 
$\Omega^{\varphi}_{i,j}$. We have
\beq
\Omega^{\varphi}_{i,j} \= \sum_{i_1,j_1\ge0} 
\frac{\p t_{i_1}}{\p T_i} \frac{\p t_{j_1}}{\p T_j} \, \Omega^{\scriptscriptstyle \rm KdV}_{i_1,j_1} \,, \quad i,j\ge0\,.
\eeq
This implies that 
$\Omega^{\varphi}_{i,j}$ belong to $\mathcal{A}_{U^\varphi}[[\e^2]]$. In Section~\ref{section8} we will give a more general description about this.

\begin{remark}\label{remark12-717}
Recall that the Hodge hierarchy is 
a $\tau$-symmetric integrable Hamiltonian perturbation of the RH hierarchy, depending on an infinite sequence of 
parameters~\cite{DLYZ16} (see also~\cite{BPS12-1, BPS12-2}). 
The Hodge universality conjecture proposed in~\cite{DLYZ16} says that the Hodge hierarchy is a universal object 
in $\tau$-symmetric Hamiltonian integrable hierarchies, meaning that any 
$\tau$-symmetric Hamiltonian integrable hierarchy in the sense of~\cite{DLYZ16} is related to the Hodge hierarchy via a 
Miura-type transformation.  
The WK mapping hierarchy is integrable, Hamiltonian (actually bihamiltonian) and 
possesses a $\tau$-structure. However, in general its hamiltonian densities cannot be chosen 
to satisfy the $\tau$-symmetry condition of~\cite{DLYZ16, DZ-norm}. So our result is consistent with~\cite{DLYZ16}.
\end{remark}

Using the explicit expressions~\eqref{p1explicit1126}--\eqref{p2explicit1126} one can easily compute 
the central invariant of the pencil $P^\varphi_2+\lambda P^\varphi_1$. Although this invariant
can be deduced from the result of~\cite{LWangZ} or Proposition~\ref{centralprop1202}, we give an 
explicit computation.
The canonical coordinate for the pencil $P^\varphi_2(U^\varphi)+\lambda P^\varphi_1(U^\varphi)$ is $\varphi(U^\varphi)$.
Perform the following Miura-type transformation to~\eqref{WKmappingham1126}:
\beq
\hat u \= \varphi(U^{\varphi}) \,.
\eeq
The Poisson operators $P_1,P_2$ in the $\hat u$-coordinate read:
\begin{align}
& P^\varphi_1(\hat u) \= \frac12 \, \varphi'(\varphi^{-1}(\hat u)) \circ \p 
\+ \frac12 \, \p \circ \varphi'(\varphi^{-1}(\hat u))  \+  \cdots \,, \\
& P^\varphi_2(\hat u) \= \frac12 \, \hat u \, \varphi'((\varphi^{-1}(\hat u)) \circ \p 
\+ \frac12 \, \p \circ \hat u \circ \varphi'(\varphi^{-1}(\hat u))  \+  \cdots \,,
\end{align}
where ``$\cdots$" denotes terms containing $\e^2, \e^4$, $\cdots$. In particular, 
the terms containing~$\e^2$ can be obtained from~\eqref{p1explicit1126}--\eqref{p2explicit1126}.
Now using \cite[formula~(1.49)]{DLZ06} we find that the central invariant is the constant-valued function $1/24$ 
for any~$\varphi$. 
So the WK mapping hierarchy leads to a construction of the representatives of Poisson pencils for 
\beq\label{203719}
(g,c) \= \Bigl(\varphi'(\varphi^{-1}(\hat u)),\frac1{24}\Bigr) \,.
\eeq 

Let us now give a proof of Theorem~\ref{theoremnew1} (or equivalently Theorem~10$'$).
\begin{proof}[Proof of Theorem~\ref{theoremnew1}] 
The necessity is already implied by Theorem~\ref{theorembiham43}.
For the sufficiency, using the argument given after the statement of Theorem~\ref{theorembiham43}, we know that, 
under a proper choice of the Poisson pencil for the bihamiltonian perturbation under consideration, 
the central invariant equals $1/24$. Since the WK mapping hierarchy, which is bihamiltonian, also 
has central invariant~$1/24$, the results \cite[Theorems~1 and~2]{LZ05} and \cite[Theorem~1.8]{DLZ06} then imply that 
the given bihamiltonian perturbation is Miura equivalent to the WK mapping hierarchy. 
As the WK mapping hierarchy has a tau-structure, the sufficiency part is proved by recalling Lemma~\ref{haminvtaulemma1203}.
\end{proof}
The above proof
 immediately leads also to a proof of the following theorem, which we call
the {\it WK mapping universality theorem}.
\begin{theorem}  \label{WKmuconj226}
The abstract local WK mapping hierarchy is a universal object for bihamiltonian perturbations of the abstract local RH hierarchy possessing  
a $\tau$-structure.
\end{theorem}

To make the content of Theorem~\ref{WKmuconj226} clearer,
and at the same time as an additional check that the statement is correct, we provide 
some direct verifications up to order $\e^8$ (i.e., up to and including $a_3$), with $q_2$ and $q_3$ given by~\eqref{q2q31129}.
Namely, the following Miura-type transformation 
\beq
U^\varphi \;\mapsto\; w \= M(U^\varphi) \+ 
\sum_{k=1}^4 \e^{2k} \sum_{\lambda \in \mathcal{P}_{2k}} C_\lambda(U^\varphi) \, U^\varphi_{\lambda}  \+ \mathcal{O}(\e^{10})
\eeq 
transforms the WK mapping hierarchy~\eqref{mappinghierarchy} to the standard form \eqref{standardform1113}, \eqref{hstand1128} 
up to and including the $\e^8$ term.
Here, 
\beq
M(U^\varphi) \= \int_0^{U^{\varphi}} \sqrt{\varphi'(y)} \, dy\,,
\eeq
\beq
a_0(w) \= M'(M^{-1}(w)) \,,
\eeq
and $C_\lambda$ for $|\lambda| \le 8$ are explicit expressions, e.g.,
\begin{align}
C_{(2)}(U^\varphi) & \=  -\frac{q}{8} \, (\varphi'(U^\varphi))^{3/2}\,, \\
C_{(1^2)}(U^\varphi) & \=  -\frac{q}8 \, (\varphi'(U^\varphi))^{1/2} \, \varphi''(U^\varphi) \+ \frac{\varphi''(U^\varphi)^2}{24 \, \varphi'(U^\varphi)^{3/2}}
 - \frac{\varphi^{(3)}(U^\varphi)}{48 \sqrt{\varphi'(U^\varphi)}} \, , \\
C_{(1^8)}(U^\varphi) & \= - \frac{107}{185794560} \, \frac{\varphi^{(12)}(U^\varphi)}{\sqrt{\varphi'(U^\varphi)}} 
\+ ~{\rm more~than~two~hundred~terms} \,. 
\end{align}

We end this section with one more remark.

\begin{remark}
Using the Miura-type transformation we can  
  assume that $\Omega_{1,1}=U$ in the setting for 
 hamiltonian perturbations~\eqref{ham115} of the RH hierarchy possessing   
 a $\tau$-structure. When $\Omega_{1,1}=U$, we can define 
the {\it normal Miura-type transformation}, which has the form
\beq
 U \; \mapsto \; \tilde U \= \sum_{k\geq0} \tilde U^{[k]} \, \e^k \= U \+ \e^2 \, \p^2(A(U, U_1, U_2,\dots;\e))
\eeq 
for some $A\in  \mathcal{A}_U[[\e]]_0$.
So normal Miura-type transformations 
are a special class, but whenever we use the word ``normal" we mean that the
transformation does not just act on the hierarchy as usual, 
but also acts on the $\tau$-structure 
by
\beq
 \tilde \Omega_{S_1,S_2} \,:=\, \Omega_{S_1,S_2} \+ \e^2 \, D_{S_1} D_{S_2} (A)\,.
\eeq
Since the normal Miura-type transformation  
changes the resulting $\tau$-function when $A\neq{\rm constant}$, it plays the 
role of choosing different $\tau$-structures. However, in this paper, we do not use this terminology. 
\end{remark}

\section{A special group element and the Hodge--WK correspondence}\label{sectionexample223}
In this section, we will consider the particular example for the WK mapping hierarchy given by 
\beq\label{examplephi}
\varphi_{\rm special}(V) \= \frac{e^{2 \, q \, V} - 1}{2 \, q} \= V \+ q\, V^2 \+ \frac{2}3 \, q^2 \, V^3 \+ \cdots\, \in \, \G
\eeq
(as in~\eqref{varphispecialdefi})
over the ground ring $R=\QQ[q]$, where $q$ is a free parameter. The inverse group element is $f_{\rm special}(v)=\log(1+2qv)/2q$.
We will establish a relationship between the WK mapping partition function associated to~\eqref{examplephi} and 
a Hodge partition function. 

Before entering into the details, we recall some general terminology for the Hodge side.  
Let $\EE_{g,n}$ be the rank~$g$ Hodge bundle on~$\overline{\M}_{g,n}$. 
By {\it Hodge integrals} we mean integrals of the form
\beq\label{hodgeintegraldef}
\int_{\overline{\M}_{g,n}} \, \psi_1^{i_1}\cdots\psi_n^{i_n} \, \lambda_{j_1} \cdots \lambda_{j_m}\,,
\eeq
where $\lambda_j:=c_j(\EE_{g,n})$, $j=0,\dots,g$, are Chern classes of~$\EE_{g,n}$, and $0\leq j_1,\dots,j_m\leq g$. 
The degree-dimension matching now reads
\beq
\sum_{a=1}^n \, i_a \+ \sum_{b=1}^m \, j_b \= 3 g -  3 \+ n \,.
\eeq
We also denote by ${\rm ch}_{j}(\mathbb{E}_{g,n})$, $j\ge0$, the components of the Chern character of~$\EE_{g,n}$. 
Mumford's relation tells that even components of the Chern character must vanish. 
The Hodge partition function~\cite{DLYZ16, FP00, Gi01} is then 
defined by 
\beq\label{hodgepar1111}
Z_{\Omega(\boldsymbol{\sigma})}(\bt;\epsilon) \= 
\exp\biggl(\,\sum_{g, \,n\ge0} \frac{\e^{2g-2}}{n!} 
\int_{\overline{\M}_{g,n}} \, \Omega_{g,n}(\boldsymbol{\sigma}) \cdot t(\psi_1)\cdots t(\psi_n)\biggr)
\eeq 
where $\bt=(t_0,t_1,t_2,\dots)$ as before, $\boldsymbol{\sigma}=(\sigma_1,\sigma_3,\sigma_5,\dots)$ is an infinite tuple of parameters, and 
\beq\label{chernchrank1omega}
\Omega_{g,n}(\boldsymbol{\sigma}) \= \exp \biggl(\sum_{j\geq1} \, \sigma_{2j-1} \, {\rm ch}_{2j-1}(\mathbb{E}_{g,n}) \biggr) \,.
\eeq
Obviously, $Z_{\Omega({\bf 0})}(\bt;\epsilon)=Z^{\scriptscriptstyle\rm WK}(\bt;\e)$.
The logarithm $\log Z_{\Omega(\boldsymbol{\sigma})}(\bt;\epsilon)=:\F_{\Omega(\boldsymbol{\sigma})}(\bt; \epsilon)$ is called the 
{\it free energy of Hodge integrals}, for short the {\it Hodge free energy}. 
By definition the free energy $\F_{\Omega(\boldsymbol{\sigma})}(\bt; \epsilon)$ admits the genus expansion:
\beq\label{defHg1109}
\F_{\Omega(\boldsymbol{\sigma})}(\bt; \epsilon) 
\,=:\, \sum_{g\ge0} \, \e^{2g-2} \, \F_{\Omega_g(\boldsymbol{\sigma})}(\bt) \,. 
\eeq
We call $\F_{\Omega_g(\boldsymbol{\sigma})}(\bt)$ $(g\ge0)$ the {\it genus~$g$ Hodge free energy}. 

Faber and Pandharipande in \cite{FP00}
obtain the following explicit formula for the Hodge partition function:
\beq
Z_{\Omega(\boldsymbol{\sigma})}(\bt;\epsilon) \= e^{\sum_{k\ge1} \frac{B_{2k}}{(2k)!} \sigma_{2k-1} D_k} \bigl(Z^{\scriptscriptstyle\rm WK}(\bt;\e)\bigr) \, ,
\eeq
where $B_m$ denote the $m$th Bernoulli number, and $D_k$, $k\geq1$, are operators given by
\beq
D_k \= \frac{\p}{\p t_{k}} - \sum_{i\ge0} \, t_i \, \frac{\p}{\p t_{i+2k-1}} \+ \frac{\e^2}2 \, \sum_{m=0}^{2k-2} \, (-1)^m \frac{\p^2}{\p t_m\p t_{2k-2-m}} \,.
\eeq
This formula is interpreted by Givental~\cite{Gi01} as a Givental group action, and is generalized 
 from the viewpoint of Virasoro-like algebra in~\cite{LYZZ22}.

It was shown that the Hodge partition function gives rise to a $\tau$-symmetric integrable hierarchy of hamiltonian evolutionary PDEs,
 called the {\it Hodge hierarchy}, so that the Hodge partition function is a $\tau$-function for the 
 Hodge hierarchy \cite{BPS12-1,BPS12-2,DLYZ16,DZ-norm}.
Roughly speaking, the Hodge hierarchy is an integrable perturbation of the 
KdV hierarchy, with $\sigma$'s being the deformation parameters.

The interest of this section will be focused on the following specialization of the parameters in the Hodge partition function:
\beq\label{sigmaspecial1111}
\sigma_{2j-1}^{\rm special} \= (4^j-1) \, (2j-2)! \, q^{2j-1} \,, \quad j\geq1\,.
\eeq 
Using relations between the Chern character and the Chern polynomial~\cite{DLYZ16}, we have 
\beq\label{Omegaspecial81}
\Omega_{g,n}(\boldsymbol{\sigma}^{\rm special}) \= \Lambda_{g,n}(2q)^2 \, \Lambda_{g,n}(-q) \,=:\, \Omega^{\rm special}_{g,n}(q) \,,
\eeq
where $\Lambda_{g,n}(z)=\sum_{j=0}^g \lambda_j \, z^j$ denotes the Chern polynomial of~$\EE_{g,n}$. 
We call Hodge integrals with $\Omega^{\rm special}_{g,n}(q)$ the {\it special-Hodge integrals}, whose 
significance 
is manifested by the Gopakumar--Mari\~no--Vafa conjecture
regarding the Chern--Simons/string duality~\cite{GV99, MV02, LLZ03}, and is discussed in~\cite{DLYZ16, DLYZ20, DY17, YZ21}
from the viewpoints of bihamiltonian structures and random matrices.  
We call $Z_{\Omega^{\rm special}(q)}(\bt;\epsilon)$ and $\F_{\Omega^{\rm special}(q)}(\bt;\epsilon)$ 
the {\it special-Hodge partition function} and the {\it special-Hodge free energy}, respectively.

The Hodge hierarchy~\cite{DLYZ16} under the specialization~\eqref{sigmaspecial1111}, 
called the {\it special-Hodge hierarchy}, has the form
\begin{align}
&\frac{\p w}{\p t_1} \= 
w \, w' \+ \e^2 \Bigl(\frac{1}{12} w''' + \frac{q}{4} w'  w'' \Bigr) \label{hodgeh1}\\
& \quad +\, \e^4 \Bigl(\frac{q}{240} w'''''  +
\frac{q^2}{80} w'\,w'''' + \frac{q^2}{16} w'' w''' + \frac{q^3}{180} w'^2 w''' + \frac{q^3}{90} w' w''^2 \Bigr) \+ \mathcal{O}(\e^6) \,, \nn\\
&\frac{\p w}{\p t_i} \= \frac1{i!} \, w^i w' \+ \mathcal{O}(\e^2) \,, \quad i\geq 2\,. \label{hodgeh2}
\end{align}
Here, $w:=\e^2 \p_{t_0}^2 \bigl(\F_{\Omega^{\rm special}(q)}(\bt;\epsilon)\bigr)$ is the normal coordinate, and, 
prime, ${}'$, denotes the derivative with respect to~$t_0$. 

\subsection{The Hodge--WK correspondence}
The goal is to connect $Z_{\Omega^{\rm special}(q)}(\bt;\epsilon)$ with $Z^{\varphi_{\rm special}}(\bt;\e)$.
Let $\bt$ and $\bT$ be related by $\bT=\bt.\varphi_{\rm special}$ as in~\eqref{defsigmaction}
 with $\varphi_{\rm special}$ given by~\eqref{varphispecialdefi} or~\eqref{examplephi}. 
We have the following lemma.
\begin{lemma}\label{alexandrovmap}
We have 
\begin{align}\label{Ttmap}
& T_i - \delta_{i,1} \= \sum_{m=0}^i \, A(i,m) \, q^{i-m} \, (t_m-\delta_{m,1}) \,, \qquad i\geq0\,,
\end{align}
where
\beq\label{Adef}
A(i,m) \;:=\; \frac1{2^m\, m!} \, \sum_{j=0}^m \, (-1)^{m-j} \, (2j+1)^i \, \binom{m}{j} \,,
\eeq
or equivalently via the following generating series
\begin{align}
T(z) \= \sum_{m\geq0} \, (t_m-\delta_{m,1}) \, \frac{z^m}{(1-qz)(1-3qz)\cdots(1-(2m+1)qz)} \,. \label{Agen}
\end{align}
The inverse map is given by
\beq\label{tTmap814}
t_m - \delta_{m,1} \= \sum_{i=0}^m \, P(m,i) \, q^{m-i} \, (T_i-\delta_{i,1}) \,,
\eeq
where 
\beq
P(m,i) \,:=\, (-1)^{m-i} \, \sum_{j=i}^m  \, (-2)^{m-j} \, \binom{j}{i} \, s(m,j) \,.
\eeq
Here, $s(m,j)$ denotes the Stirling number of the first kind.
\end{lemma}
\begin{proof}
Formula~\eqref{Ttmap} can be obtained via a direct verification. The rest of the proof is an elementary exercise. 
\end{proof}

It follows from the definition of Stirling numbers that the $P(m,i)$ admit the generating function:
\beq
\sum_{i=0}^{m} \, P(m,i) \, z^i \= (z-1)(z-3)(z-5) \cdots (z-(2m-1)) \= 2^m\,\Bigl(\frac{z}2-\frac{2m-1}2\Bigr)_{m} \,, 
\eeq
and also that $P(m,0)$ and $P(m,1)$ have the more explicit expressions
\beq\label{pm1116}
P(m,0) \= (-1)^m (2m-1)!!\,, \;  \; 
P(m,1) \= (-1)^{m-1} (2m-1)!! \, \sum_{j=0}^{m-1} \, \frac{1}{2j+1} \,,
\eeq
the first values of $P(m,1)$ being 0, 1, $-4$, $23$, $-176$, $1689$.

The loop equation~\eqref{eqn:loopdz} now reads explicitly as follows:
\begin{align}
&  \sum_{k\ge0} \, \sum_{m=1}^k \, \binom{k}{m} \, \Biggl(\p^{m-1} \biggl(\frac{\sqrt{1+2 q v}}{(1+2 \lambda  q) \sqrt{\lambda - v}}\biggr) \, 
\p^{k-m+1} \biggl(\frac{\sqrt{1+2 q v}}{(1+2 \lambda  q) \sqrt{\lambda - v}}\biggr) \Biggr)_- \, \frac{\p F_{\rm h.g.}^{\varphi_{\rm special}}}{\p V_{k}}  \label{230} \\
& + \, \sum_{k\ge0} \p^k\biggl(\frac{1+2 q v}{(1+2 \lambda  q)^2 (\lambda - v)}\biggr)_- \, \frac{\p F_{\rm h.g.}^{\varphi_{\rm special}}}{\p V_{k}} \nn\\
& - \frac{\e^2}2 \, 
\sum_{k_1,k_2\ge0} \, \Biggl(\p^{k_1+1} \biggl(\frac{\sqrt{1+2 q v}}{(1+2 \lambda  q) \sqrt{\lambda - v}}\biggr) \, 
\p^{k_2+1} \biggl(\frac{\sqrt{1+2 q v}}{(1+2 \lambda  q) \sqrt{\lambda - v}}\biggr) \Biggr)_- \, 
s \bigl(F_{\rm h.g.}^{\varphi_{\rm special}}, V_{k_1}, V_{k_2} \bigr) \nn\\
& - \frac{\e^2}{16} \, \sum_{k\ge0} \p^{k+2} \biggl(\frac{1}{(\lambda-v)^2}\biggr)
\frac{\p F_{\rm h.g.}^{\varphi_{\rm special}}}{\p V_{k}} \= \frac{1} {16 \, (\lambda-v)^2}  \,. \nn
\end{align}
(Recall that ``h.g." stands for 
``higher genera" and refers to the sum over all contributions from~$g>0$, 
while $\varphi_{\rm special}$ is the function defined in~\eqref{examplephi}.)

We are ready to give a proof of the Hodge--WK correspondence.
\begin{proof}[Proof of Theorem~\ref{thmhwk}]
According to Definition~\ref{definitionZphi}, 
$$
Z^{\varphi_{\rm special}}(\bT;\e) \= Z^{\scriptscriptstyle\rm WK}(\bt,\e) \,.
$$
To show~\eqref{mainidentity}, it is equivalent to show
\beq
Z_{\Omega^{\rm special}(q)}(\bT; \e) \= Z^{\varphi_{\rm special}}(\bT;\e) \, ,
\eeq 
or equivalently to show for all $g\ge0$, 
$\F_{\Omega_g^{\rm special}(q)}(\bT) = \F^{\varphi_{\rm special}}_g(\bT)$.
Let us first prove their genus zero parts are equal, 
and by Theorem~\ref{thmgenus0} this is equivalent to the following known fact for Hodge integrals:
\beq
\F_{\Omega_0^{\rm special}(q)}(\bT) \= \F_0(\bT) \,.  
\eeq 

Let us continue to show the higher genus parts of $\F_{\Omega^{\rm special}(q)}(\bT; \e) $ 
and $\F^{\varphi_{\rm special}}(\bT;\e)$ are equal. 
To this end, we first notice that both $\F_{\Omega_g^{\rm special}(q)}(\bT)$ and 
$\F^{\varphi_{\rm special}}_g(\bT)$ for $g\ge1$ admit the jet representations. Indeed, for
$\F^{\varphi_{\rm special}}_g(\bT)$, the jet representation is given by Proposition~\ref{jetreprep}; 
for $\F_{\Omega_g^{\rm special}(q)}(\bT)$, 
it is known from for example~\cite{DLYZ16, DLYZ20} that 
\begin{align} 
& \F_{\Omega_1^{\rm special}(q)}(\bT) \= F_{\Omega_1^{\rm special}(q)}\biggl(E(\bT), \frac{\p E(\bT)}{\p T_0}\biggr)\,, \label{h1815-1}
\end{align}
with
\begin{align}
& F_{\Omega_1^{\rm special}(q)}(V_0,V_1) \,:=\, \frac1{24} \, \log V_1 \+ \frac{q}{8} \, V_0 \,, \label{h1815-2}
\end{align}
and that for each $g\geq2$ there exists
$$F_{\Omega_g^{\rm special}(q)}=F_{\Omega_g^{\rm special}(q)}(V_1,\dots,V_{3g-2}) \;\in\; \QQ[q]\bigl[V_1^{-1},V_1,V_2,\dots,V_{3g-2}\bigr]\,,$$
satisfying
\begin{align}
&\sum_{k=1}^{3g-2} \, k \, V_k \, \frac{\p F_{\Omega_g^{\rm special}(q)}}{\p V_k} \= (2g-2) \, F_{\Omega_g^{\rm special}(q)} \,, \label{2g-2Fom1203}\\ 
& q \, \frac{\p F_{\Omega_g^{\rm special}(q)}}{\p q} \+ \sum_{k=1}^{3g-2} \, (k-1) \, V_k \, \frac{\p F_{\Omega_g^{\rm special}(q)}}{\p V_k} 
\= (3g-3) \, F_{\Omega_g^{\rm special}(q)} \,, 
\end{align}
such that 
\beq
\F_{\Omega_g^{\rm special}(q)}(\bT) 
\= F_{\Omega_g^{\rm special}(q)}\biggl(\frac{\p V(\bT)}{\p T_0}, \dots, \frac{\p^{3g-2} V(\bT)}{\p T_0^{3g-2}}\biggr)\,.
\eeq

So, in order to show $F_{\Omega_g^{\rm special}(q)}(\bT)=\F^{\varphi_{\rm special}}_g(\bT)$, $g\geq1$, it suffices to show that 
\beq\label{HFg810}
F_{\Omega_g^{\rm special}(q)} \= F_g^{\varphi_{\rm special}} \,.
\eeq

For $g=1$, \eqref{HFg810} is true due to~\eqref{h1815-2}, \eqref{genus1K1}, \eqref{examplephi}.

For $g\ge2$, let us compare the loop equations. 
We use $F_{\rm h.g.}^{\Omega^{\rm special}(q)}$ to denote $\sum_{g\ge1} \e^{2g-2} F_{\Omega_g^{\rm special}(q)}$. The following loop equation 
for $F_{\rm h.g.}^{\Omega^{\rm special}(q)}$ can be obtained from~\cite{DLYZ20}: 
\begin{align}
& \sum_{k\geq 0} 
\biggl(\p^k \Bigl(\frac1{P^2}\Bigr) + \sum_{r=1}^k \binom{k}{r} \, \p^{r-1} \Bigl(\frac1P\Bigr) \, \p^{k-r+1} \Bigl(\frac1P\Bigr)\biggr)  
\frac{\p F_{\rm h.g.}^{\Omega^{\rm special}(q)}}{\p V_k}  \label{loopf} \\
& - \, \frac{\e^2}2 \, \sum_{k_1, k_2 \ge 0} \p^{k_1+1} 
\Bigl(\frac1P\Bigr) \p^{k_2+1} \Bigl(\frac1P\Bigr)  \Biggl(\frac{\p F_{\rm h.g.}^{\Omega^{\rm special}(q)}}{\p V_{k_1}} 
\frac{\p F_{\rm h.g.}^{\Omega^{\rm special}(q)}}{\p V_{k_2}} + \frac{\p^2 F_{\rm h.g.}^{\Omega^{\rm special}(q)}}{\p V_{k_1} \p V_{k_2}}\Biggr)  \nn\\
& - \, \frac{\e^2}{16} \, \sum_{k\ge0} \, \p^{k+2} \biggl(\frac1{P^4} + \frac{4 q}{P^2} \biggr) \, 
\frac{\p F_{\rm h.g.}^{\Omega^{\rm special}(q)}}{\p V_k} 
\,-\, \frac{q} {4 P^2} \,-\, \frac1{16 P^4} \= 0 \,, \nn
\end{align}
where $P=\sqrt{-\frac1{2q}-\frac{4 \, e^{-2qV}}{\lambda}}$ and $\p=\sum_k V_{k+1} \p/\p V_k$. Recall that the loop equation~\eqref{loopf} holds identically in~$\lambda$ and so holds identically in~$P$. Note that 
\beq\label{dp815}
\p(P) \= - \Bigl(q\, P + \frac{1}{2 P}\Bigr) \, V_1 \,.
\eeq
Introduce
\beq
\widetilde P \= e^{-q V} \sqrt{\lambda -\frac{e^{2 q V}-1}{2 q}} \,,
\eeq
and observe that 
\beq
\p\bigl(\widetilde P\bigr) \= - \biggl(q \, \widetilde P + \frac{1}{2 \widetilde P}\biggr) \, V_1 \,,
\eeq
which has the same form as~\eqref{dp815}. Therefore,  
\begin{align}
& \sum_{k\geq 0} 
\biggl(\p^k \biggl(\frac1{\widetilde P^2}\biggr) 
+ \sum_{r=1}^k \binom{k}{r} \, \p^{r-1} \biggl(\frac1{\widetilde P}\biggr) \; \p^{k-r+1} \biggl(\frac1{\widetilde P}\biggr)\biggr)  
\frac{\p F_{\rm h.g.}^{\Omega^{\rm special}(q)}}{\p V_k}  \label{loopf815} \\
& - \, \frac{\e^2}2 \, \sum_{k_1, k_2 \ge0} \p^{k_1+1} 
\biggl(\frac1{\widetilde P}\biggr) \p^{k_2+1} \biggl(\frac1{\widetilde P}\biggr)  \biggl(\frac{\p F_{\rm h.g.}^{\Omega^{\rm special}(q)}}{\p V_{k_1}} 
\frac{\p F_{\rm h.g.}^{\Omega^{\rm special}(q)}}{\p V_{k_2}} + \frac{\p^2 F_{\rm h.g.}^{\Omega^{\rm special}(q)}}{\p V_{k_1} \p V_{k_2}}\biggr)  \nn\\
& - \, \frac{\e^2}{16} \, \sum_{k\ge0} \, \p^{k+2} \biggl(\frac1{\widetilde P^4} + \frac{4 \, q}{\widetilde P^2} \biggr) \, 
\frac{\p F_{\rm h.g.}^{\Omega^{\rm special}(q)}}{\p V_k} \,-\, \frac{q} {4 \widetilde P^2} \,-\, \frac1{16 \widetilde P^4} \= 0 \nn
\end{align}
holds identically in~$\widetilde P$. Dividing both sides of~\eqref{loopf815} by $(1+2\lambda q)^2$, we obtain 
\begin{align}
& \sum_{k} \sum_{m=1}^k \, \binom{k}{m} \, \p^{m-1} \Biggl(\frac{e^{q V}/(1+2\lambda q)}{ \bigl(\lambda -\frac{e^{2 q V}-1}{2 q}\bigr)^{1/2}}\Biggr) \, 
\p^{k-m+1} \Biggl(\frac{e^{q V}/(1+2\lambda q)}{\bigl(\lambda -\frac{e^{2 q V}-1}{2 q}\bigr)^{1/2}}\Biggr) \, 
\frac{\p F_{\rm h.g.}^{\Omega^{\rm special}(q)}}{\p V_{k}} \nn \\
& + \, \sum_{k\ge0} \p^k\Biggl(\frac{e^{2q V}}{(1+2\lambda q)^2 \bigl(\lambda-\frac{e^{2 q V}-1}{2 q}\bigr)}\Biggr) \, 
\frac{\p F_{\rm h.g.}^{\Omega^{\rm special}(q)}}{\p V_{k}} \nn\\
& - \, \frac{\e^2}2 \, \sum_{k_1,k_2\ge0} \, 
\p^{k_1+1} \Biggl(\frac{e^{q V}/(1+2\lambda q)}{ \bigl(\lambda -\frac{e^{2 q V}-1}{2 q}\bigr)^{1/2}}\Biggr) \, 
\p^{k_2+1} \Biggl(\frac{e^{q V}/(1+2\lambda q)}{\bigl(\lambda -\frac{e^{2 q V}-1}{2 q}\bigr)^{1/2}}\Biggr) \nn \\
& \qquad \qquad \times \Biggl(\frac{\p F_{\rm h.g.}^{\Omega^{\rm special}(q)}}{\p V_{k_1}} 
\frac{\p F_{\rm h.g.}^{\Omega^{\rm special}(q)}}{\p V_{k_2}} + \frac{\p^2 F_{\rm h.g.}^{\Omega^{\rm special}(q)}}{\p V_{k_1} \p V_{k_2}}\Biggr) \nn\\
& - \frac{\e^2}{16} \, \sum_{k\ge0} \, \p^{k+2} 
\Biggl(\frac{(1+2 \lambda  q)^2}{\bigl(\lambda-\frac{e^{2 q V}-1}{2 q}\bigr)^2} \Biggr) \, 
\frac{\p F_{\rm h.g.}^{\Omega^{\rm special}(q)}}{\p V_{k}} \= \frac{1} {16 \bigl(\lambda-\frac{e^{2 q V}-1}{2 q}\bigr)^2} \,-\, \frac{q^2}{4 (1+2\lambda q)^2}  \,. \nn
\end{align}
This implies that $F_{\rm h.g.}^{\Omega^{\rm special}(q)}$ satisfies equation~\eqref{230}. Since the solution to~\eqref{230} is 
unique with~\eqref{2g-2Fom1203}, we conclude that $F_{\rm h.g.}^{\varphi_{\rm special}}= F_{\rm h.g.}^{\Omega^{\rm special}(q)}$.
The theorem is proved. 
\end{proof}

We mention that for $\varphi_{\rm special}$ 
we have the following explicit 
$v\leftrightarrow V$ map: 
\begin{align}
& M_0(Y;q) \;=\; \frac{e^{2q Y_0}-1}{2q} \,, \label{M0Y} \\
& \frac1{2q} \,+\,  \sum_{j\geq1} \, e^{-(j+2)qY_0} \, M_j(Y;q) \, \frac{\lambda^j}{j!} \;=\; \frac{1}{2q} \, \frac{w(Y;\lambda;q)^2}{\lambda^2} 
\end{align}
with $w=w(Y;\lambda;q)=\lambda+\cdots$ being the unique element in $\lambda+\CC[Y,q][[\lambda]]\lambda$ satisfying
\beq
 w \, e^{- \, q \, R(Y;w)} \= \lambda\,, \quad R(Y;w)\,:=\, \sum_{k\ge1} \, Y_k \, \frac{w^k}{k!} \,.
\eeq

Although it is not completely obvious, Theorem~\ref{thmhwk} is equivalent to the following theorem obtained  
by Alexandrov.

\smallskip

\noindent {\bf Theorem~B} (Alexandrov~\cite{Al212}). 
{\it Define an invertible $\bt \to \hat\bT$ map by 
\beq\label{amap813}
\hat T_i \= S^{i} (t_0) \,,  \quad i\geq0\,,
\eeq
where $S$ denotes the following differential operator 
\beq
S \= \sum_{m\ge0} \, t_{m+1} \, \frac{\p }{\p t_m} \+ q \, \sum_{m\ge0} \, (2m+1) \, t_{m} \, \frac{\p }{\p t_m} \,,
\eeq
and define a sequence of elements $c_m \in \QQ[q]$ by
\beq
c_0 \= c_1 \= 0\,, \quad c_m \= P(m,1) \, q^{m-1}
\eeq
with $P(m,1)$ as in~\eqref{pm1116}.
Then we have}
\beq
Z_{\Omega^{\rm special}}\bigl(\hat \bT; \e\bigr) \= Z^{\scriptscriptstyle \rm WK}(\bt;\e)|_{t_m\mapsto t_m-c_m,\, m\ge0} \,.
\eeq

 To see the equivalence between Theorem~\ref{thmhwk} and Theorem~B, 
first of all, it is an elementary exercise to verify that the linear $\bt \leftrightarrow \hat\bT$ map defined 
by~\eqref{amap813} coincides with the one given in 
equation~\eqref{tTdef627} specialized to $\varphi_{\rm special}(V)=(e^{2qV}-1)/2q$. 
The equivalence is then proved by using the relation~\eqref{thatTrelation}.

Since the WK partition function $Z^{\scriptscriptstyle \rm WK}(\bt; \e)$ is a KdV $\tau$-function, Theorem~\ref{thmhwk} 
immediately implies the following corollary.

\smallskip 

\noindent {\bf Corollary}. The partition function 
{\it $Z_{\Omega^{\rm special}(q)}(\bt . \varphi_{\rm special};\e)$ is a KdV $\tau$-function.}

\smallskip

From~\eqref{ThatT} we know that the Corollary can be equivalently stated as the following

\smallskip

\noindent {\bf Proposition~A} (Alexandrov~\cite{Al21}). {\it The partition function
$Z_{\Omega^{\rm special}(q)}\bigl(\bt | \varphi_{\rm special};\e\bigr)$ is a KdV $\tau$-function.}

\begin{remark}
Proposition~A was obtained by Alexandrov in~\cite{Al21}. Alexandrov's proof in~\cite{Al212} of Theorem~B
was based on Proposition~A. 
In an early version of the current paper, we also deduced Theorem~B  
 from Proposition~A before \cite{Al212} appeared on arXiv with a complete proof (different from Alexandrov's),  
 and also sketched a second proof, not based on Proposition~A but using the Virasoro constraints~\eqref{virasorowk} instead, which 
reduces Theorem~B to a few elementary identities including e.g. 
\begin{align}
&\sum_{i=m}^n \, \frac{2i+1}2 \, A(i,m) \, 
P(n,i) \;-\; \sum_{j\ge1} \, \frac{4^j-1}{2j} \, B_{2j} \, \sum_{i=m}^{n-(2j-1)} A(i,m) \, 
 P(n,i+2j-1) \\
&=  \delta_{n,m} \, \frac{2m+1}{2}
+ \delta_{n\geq m+1} \frac{(-1)^{n-m-1}}{(n-m) (n-m+1)}  \, \frac{(2n+1)!!}{2 \cdot (2m-1)!!}  \,, \quad \forall\, n\geq m\geq0. \nn
\end{align} 
The proof of Theorem~\ref{thmhwk} given here provides a yet different proof of Theorem~B, again not using Proposition~A (and thus self-contained),
but using the loop equations instead. 
In this proof, 
we used a technique 
found during the preparation of the paper~\cite{YZ21} by Q.~Zhang and the first-named author
 of the present paper of identifying solutions to the loop equations (although in the end that 
 technique was not given there), whereas our earlier and less self-contained 
proof preceded~\cite{YZ21}.
\end{remark}

\subsection{Two applications} 
In this subsection we will give two applications of   
the Hodge--WK correspondence. 

\smallskip

\noindent {\bf Application I.} {\it The WK--GUE correspondence.} 
Recall that the Hodge--GUE correspondence was 
recently obtained in~\cite{DLYZ20, DY17} (see also~\cite{DLYZ16}), which establishes a relationship between the special-Hodge 
partition function with $q=-1/2$ 
and the even Gaussian Unitary Ensemble (GUE) partition function  (see~\cite{AvM, BIZ, DY17-2, HZ86, Wi91, Y22}). 
So 
\beq
\varphi_{\rm special}(V)|_{q=-1/2} \= 1-e^{-V}
\eeq
is now under consideration. 

Let 
\beq\label{globaldefZ}
Z^{\scriptscriptstyle \rm eGUE1}_n({\bf s};\e) 
 \,:=\, 
2^{-\frac{n}2} (\pi \e)^{-\frac{n^2}2} 
\int_{{\mathcal H}(n)} e^{-\frac1{2\e} {\rm tr} \, M^2 + {\frac1\e}  \sum_{j\ge1} s_j {\rm tr} \, M^{2j}} dM 
\eeq
be the normalized even GUE partition function of size~$n$ \cite{BIZ, HZ86, Wi91}.
Here ${\bf s}=(s_1,s_2,\cdots)$ is an infinite tuple of indeterminates, $\e$ is an indeterminate, and 
 ``normalized" means that $Z^{\scriptscriptstyle \rm eGUE1}_n({\bf 0};\e) = 1$ for all $n$ and~$\e$.
According to~\cite{BIZ, HZ86, Hooft1, Hooft2}, 
the logarithm 
$\log Z^{\scriptscriptstyle \rm eGUE1}_n({\bf s};\e)=:\F^{\scriptscriptstyle \rm eGUE1}_n({\bf s};\e)$ has the expansion:
\beq
\F^{\scriptscriptstyle \rm eGUE1}_n({\bf s};\e) \= 
\sum_{g\ge0} \sum_{k\ge1} \sum_{j_1,\dots,j_k\ge1} a_g(2j_1,\dots,2j_k) \, s_{j_1} \dots s_{j_k} \, n^{2-2g-k+|j|} \, \e^{|j|-k}\,,
\eeq
where $|j|:=j_1+\dots+j_k$, and 
\beq
a_g(2j_1,\dots,2j_k) \= \sum_{G\in {\rm Ribb}_g(2j_1,\dots,2j_k)} \frac{2j_1\cdots 2j_k}{|{\rm Aut}(G)|}\,.
\eeq
Here ${\rm Ribb}_g(2j_1,\dots,2j_k)$ denotes the set of connected ribbon graphs of genus $g$ with 
$k$ vertices of valencies $2j_1$, \dots, $2j_k$. Note that $a_g(2j_1,\dots,2j_k)$ vanishes unless 
$2-2g-k+|j|>0$. Therefore, $\log Z^{\scriptscriptstyle \rm eGUE1}_n({\bf s};\e)\in n \, \QQ[n,\e][[{\bf s}]]$.

Let $x=n\e$ denote the t'Hooft coupling constant \cite{Hooft1, Hooft2}, and introduce 
\beq
\gamma(z)= \frac{z^2}2 \Bigl(\log z - \frac32 \Bigr) \,-\, \frac{\log z}{12} 
 \+ \sum_{g\ge2} \frac{B_{2g}}{2g (2g-2) \, z^{2g-2}} \,,
\eeq
which satisfies the second-order difference equation
\beq
\gamma(z+1) \,-\, 2 \, \gamma(z) \+ \gamma(z-1) \= \log z 
\eeq
and for $z$ large and integral gives the asymptotic expansion of $\log (1! \thin 2! \cdots (z-1)!)$ 
up to an additive affine-linear function $\zeta'(-1)+  \log (2\pi) z/2$ \cite{WW}.
Following~\cite{AvM, DY17, DY17-2}, 
define the {\it corrected even GUE free energy} (for short the {\it even GUE free energy}), 
denoted $\F^{\scriptscriptstyle \rm eGUE}(x, {\bf s};\e)$, as follows: 
\begin{align}
\F^{\scriptscriptstyle \rm eGUE}(x, {\bf s};\e) \:= 
C(x,\e) \+  \F^{\scriptscriptstyle \rm eGUE1}_{x/\e}({\bf s};\e) \,,
\end{align}
where 
\begin{align}
C(x,\e) \:= \, & \gamma(\frac{x}\e) \+ \frac{x^2}{2 \, \e^2} \log \e 
\,-\, \frac1{12} \log \e \\
= \, & \frac{x^2}{2\,\e^2} \, \Bigl(\log x - \frac32 \Bigr) \,-\, \frac{\log x}{12} 
 \+ \sum_{g\ge2} \frac{\e^{2g-2} B_{2g}}{2g \, (2g-2) \, x^{2g-2}} \,.
\end{align}
We also define the even GUE partition function $Z^{\scriptscriptstyle \rm eGUE}(x, {\bf s};\e)$ as  
$e^{\F^{\scriptscriptstyle \rm eGUE}(x, {\bf s};\e)}$. From the definition we know that 
$\F^{\scriptscriptstyle \rm eGUE}(x, {\bf s};\e)\in \e^{-2} \QQ[[x-1, \e^2]][[{\bf s}]]$ and 
$Z^{\scriptscriptstyle \rm eGUE}(x, {\bf s};\e) \in \QQ((\e^2))[[x-1]][[{\bf s}]]$.
For more details about the  
even GUE partition function see \cite{BIZ, Du09, DLYZ20, DY17, DY17-2, HZ86, Wi91, Y22}. 
For $k\geq1$, and $j_1,\dots,j_k\ge1$, introduce the notation:
\begin{align}
\langle m_{j_1}\dots m_{j_k}\rangle (x,\e) \:= 
\frac{\p^k \F^{\scriptscriptstyle \rm eGUE1}_{x/\e}({\bf s};\e)}{\p s_{j_1} \dots \p s_{j_k}} \bigg|_{{\bf s}={\bf 0}} \,.
\end{align}
 Explicitly, 
\begin{align}
\langle m_{j_1}\dots m_{j_k}\rangle (x,\e) \= k!\, \sum_{0\leq g\leq \frac{|j|}2 -\frac{k}2 +\frac12} \, a_g(2j_1,\dots,2j_k)  \, x^{2-2g-k+|j|} \, \e^{2g-2}\,.
\end{align}

The {\it modified even GUE partition function} $Z^{\scriptscriptstyle \rm meGUE}(x, {\bf s};\e)$
 is introduced in~\cite{DLYZ16} (see also~\cite{DY17}) as the unique element in $\QQ((\e^2))[[x-1]][[{\bf s}]]$ such that 
\beq
Z^{\scriptscriptstyle \rm eGUE}(x, {\bf s};\e) \= 
Z^{\scriptscriptstyle \rm meGUE}\Bigl(x-\frac\e2, {\bf s};\e\Bigr) \, Z^{\scriptscriptstyle \rm meGUE}\Bigl(x+\frac\e2, {\bf s};\e\Bigr) \,.
\eeq
This partition function also relates to the Laguerre Unitary Ensemble (LUE) or say to 
Grothendieck's dessins d'enfant \cite{GGR, YZ22, Zhou}. The logarithm 
$\log Z^{\scriptscriptstyle \rm meGUE}(x, {\bf s};\e)=:\F^{\scriptscriptstyle \rm meGUE}(x, {\bf s};\e)$ 
is called the {\it modified even GUE free energy}, which has the form
\begin{align}
& \F^{\scriptscriptstyle \rm meGUE}(x, {\bf s};\e) \= B(x,\e)  \+ 
\sum_{k\ge1} \, \frac1{k!} \, \sum_{j_1,\dots,j_k\ge1} \, 
\langle \phi_{j_1} \dots \phi_{j_k} \rangle(x,\epsilon) \, s_{j_1} \cdots s_{j_k}\,, 
\end{align}
where 
\begin{align}
& B(x,\e) \= \gamma\Bigl(\frac{x}{2\e}+\frac14\Bigr) \+ \gamma\Bigl(\frac{x}{2\e}-\frac14\Bigr) 
\+ \frac{x^2}{4\, \e^2} \log (2\e) \,-\, \frac{5}{48} \, \log (2\e) \label{263712} \\
& \quad\quad\quad \= \Bigl(\frac14 \, \log x - \frac 38\Bigr) \, \frac{x^2}{\e^2} 
\,-\, \frac{5}{48} \, \log x \,-\, \frac{53 \, \e^2}{3840 \, x^2} \+\frac{599 \, \e^4}{64512 \, x^4}  \+ \dots, \\
& \langle\phi_{j_1}\dots \phi_{j_k}\rangle(x,\e) \= \frac1{e^{\e\p_x/2}+e^{-\e \p x/2}} \,\langle m_{j_1} \dots m_{j_k} \rangle(x,\e) \,, \quad k,j_1,\dots,j_k\ge1\,.
\end{align}
Recall that the Hodge--GUE correspondence says that the following identity holds true in $\CC((\e^2))[[x-1]][[{\bf s}]]$:
\beq\label{Hodgegue1107}
Z_{\rm \Omega^{\rm special}(-1/2)}\bigl({\bf T}^{\scriptscriptstyle \rm Hodge-GUE}(x,{\bf s});\sqrt{2} \e\bigr) 
\, e^{\frac{A(x,{\bf s})}{2\e^2}} \= Z^{\scriptscriptstyle \rm meGUE}(x, {\bf s};\e) \,,
\eeq
where $A(x,{\bf s})$ is the explicit quadratic series given by~\eqref{aseries1107}
and ${\bf T}^{\scriptscriptstyle \rm Hodge-GUE}(x,{\bf s})$ is defined by 
\beq
T^{\scriptscriptstyle \rm Hodge-GUE}_i(x,{\bf s}) \= -  1 \+ \delta_{i,1} 
\+ x \, \delta_{i,0} \+ \sum_{j\geq1} \, j^{i+1} \, \binom{2j}{j} \, s_j \,, \quad i\geq 0\,.
\eeq 

We are ready to prove Theorem~\ref{WKGUEthm81}.
\begin{proof}[Proof of Theorem~\ref{WKGUEthm81}]
By using~\eqref{Hodgegue1107} and~\eqref{mainidentity}. 
\end{proof}

Let us denote  
\beq
f_m(x) \,:=\, \frac{(2m-1)!!}{2^{m}} \, x  - \frac{(2m+1)!!}{2^{m}} \quad (m\geq2).
\eeq
We will also use Witten's notation $\langle \tau_{m_1}\dots\tau_{m_n}\rangle_{g}$:
\beq
\langle \tau_{m_1}\dots\tau_{m_n}\rangle_{g} \= 
\int_{\overline{\M}_{g,n}} \, \psi_1^{m_1}\cdots\psi_n^{m_n} \,.
\eeq
Performing the Taylor expansion with respect to $s_1,s_2,\cdots$ 
on the logarithms of both sides of~\eqref{wkgue1106} and using the dilaton equation~\eqref{gwdilaton}, we arrive at the 
following proposition. 
\begin{prop}\label{taylorwkgue}
We have
\begin{align}
& \frac{\frac14-x}{\e^2} - \frac1{24} \log \frac{3-x}2 \label{272712}\\
& +\, \sum_{g, \, p\ge0} \, \e^{2g-2} \frac{(x-1)^p}{p!}  \sum_{\lambda\in\mathcal{P}_{3g-3+p}} 
\frac{\langle\tau_0^p \tau_{\lambda+1} \rangle_{g}}{{\rm mult}(\lambda)!}
\Bigl(\frac{2}{3-x}\Bigr)^{2g-2+\ell(\lambda)+p} 
\prod_{s=1}^{\ell(\lambda)} f_{1+\lambda_s}(x) \nn\\
& = \, B\Bigl(x,\frac{\e}{\sqrt{2}}\Bigr)\,, \nn
\end{align}
where $B(x,\e)$ is defined in~\eqref{263712}, and for $k, j_1,\dots,j_k\ge1$, 
\begin{align}
& \frac{\delta_{k,1}}{\e^2} \binom{2j_1}{j_1} \Bigl(x - \frac{j_1}{j_1+1}\Bigr) \+ \frac{\delta_{k,2}}{\e^2} \, \frac{j_1j_2}{j_1+j_2} \binom{2j_1}{j_1} \binom{2j_2}{j_2} \label{274712} \\
& + \, \sum_{m_1,\dots,m_k\ge0}  U_{m_1,\dots,m_k}(x,\e) \, \prod_{i=1}^k e_{m_i, \, j_i}  \nn\\
&= \,  \langle\phi_{j_1}\dots \phi_{j_k}\rangle\Bigl(x,\frac\e{\sqrt2}\Bigr) \,, \nn
\end{align}
where 
\beq
e_{m,j} \,:=\, \frac{(2m+2j-1)!!}{2^{m-j} (j-1)!}\,,
\eeq and for $m_1,\dots,m_k\ge0$,
\begin{align}
&U_{m_1,\dots,m_k}(x,\e) \, := \, \sum_{g\ge0} \e^{2g-2} \sum_{p\ge0} \, \frac{(x-1)^p}{p!} \\
& \quad \times \sum_{\lambda\in\mathcal{P}_{3g-3+p+k-|m|}} 
\frac{\langle\tau_0^p \tau_{\lambda+1} \tau_{m_1}\dots\tau_{m_k}\rangle_{g}}{{\rm mult}(\lambda)!} 
\Bigl(\frac{2}{3-x}\Bigr)^{2g-2+\ell(\lambda)+k+p} 
\prod_{s=1}^{\ell(\lambda)} f_{1+\lambda_s}(x) \,. \nn
\end{align}
\end{prop}

For instance, using~\eqref{genuszero712} 
the coefficient of $\e^{-2}$ of the left-hand side of~\eqref{272712} begins
\begin{align}
&\frac14 - x \+ \frac{(x-1)^3}{3!} \biggl[\frac2{3-x}\biggr] 
\+ \frac{(x-1)^4}{4!} \biggl[\Bigl(\frac2{3-x}\Bigr)^3 \Bigl(\frac34 x-\frac{5!!}4\Bigr)\biggr] \\
&\+ \frac{(x-1)^5}{5!} \biggl[\Bigl(\frac2{3-x}\Bigr)^4 \Bigl(\frac{5!!}8 x-\frac{7!!}8\Bigr)+
\frac6{2!}\Bigl(\frac2{3-x}\Bigr)^5 \Bigl(\frac34 x-\frac{5!!}4\Bigr)^2\biggr] \+ \cdots \,, \nn
\end{align}
which is consistent with the expansion of $\frac{x^2} 4 \log x - \frac 38 x^2$
as $x\to 1$. (With the above given terms one can check the agreement up to and including ${\rm O}((x-1)^5)$.) 
Similarly, using $\langle\tau_0\tau_2\rangle_1=\langle\tau_0^3\tau_3\rangle_1=1/24$ 
and $\langle\tau_0^2\tau_2^2\rangle_1=1/6$,
the coefficient of~$\e^0$ of the left-hand side of~\eqref{272712} begins
\begin{align}
&- \frac1{24} \log \frac{3-x}2 \+ (x-1) \biggl[\frac1{24}\Bigl(\frac{2}{3-x}\Bigr)^2 \Bigl(\frac34 x-\frac{5!!}4\Bigr)\biggr] \\
&\+ 
\frac{(x-1)^2}{2!} \biggl[\frac1{24}\Bigl(\frac{2}{3-x}\Bigr)^3\Bigl(\frac{5!!}8 x-\frac{7!!}8\Bigr) 
+ \frac16 \frac1{2!} \Bigl(\frac{2}{3-x}\Bigr)^4 \Bigl(\frac34 x-\frac{5!!}4\Bigr)^2\biggr] \+ \cdots\,, \nn
\end{align}
which is consistent with the expansion of $\frac{5}{48} \log x$ as $x\to 1$.  
The coefficient of $\e^2$ of the left-hand side of~\eqref{272712} begins
\begin{align}
& \biggl[\langle\tau_4\rangle_2 \Bigl(\frac{2}{3-x}\Bigr)^4 \Bigl(\frac{7!!}{16} x-\frac{9!!}{16}\Bigr) \+ 
\langle\tau_3\tau_2\rangle_2 \Bigl(\frac{2}{3-x}\Bigr)^5
\Bigl(\frac{5!!}8 x-\frac{7!!}8\Bigr)\Bigl(\frac34 x-\frac{5!!}4\Bigr) \\
&\+ \frac{\langle\tau_2^3\rangle_2}{3!}\Bigl(\frac{2}{3-x}\Bigr)^6\Bigl(\frac34 x-\frac{5!!}4\Bigr)^3\biggr] \+ 
(x-1) \Bigl[\cdots\Bigr] \+ \cdots \,, \nn
\end{align}
which is consistent with the expansion of $- \frac12 \frac{53 \, \e^2}{3840 \, x^2}$ as $x\to 1$ 
by substituting $\langle\tau_4\rangle_2=1/1152$, $\langle\tau_3\tau_2\rangle_2=29/5760$ 
and $\langle\tau_2^3\rangle_2=7/240$ \cite{Wi91}.
Let us present one more verification.
Recall that $\langle m_1 \rangle (x,\e)=x^2/\e^2$, giving  
$\langle \phi_1 \rangle(x,\e) = \frac{x^2}{2\, \e^2} - \frac1 8$. 
Using~\eqref{genuszero712} the coefficient of $\e^{-2}$ of the left-hand side of~\eqref{274712} for $k=1$ and $j_1=1$ begins
\begin{align}
& e_{0,1} \biggl\{ \frac{(x-1)^2}{2!} \biggl[\frac2{3-x}\biggr] \+ 
\frac{(x-1)^3}{3!}\biggl[\Bigl(\frac2{3-x}\Bigr)^3 \Bigl(\frac34 x-\frac{5!!}4\Bigr) \biggr] \+ \cdots\biggr\} \\
&\+ e_{1,1} \biggl\{\frac{(x-1)^3}{3!}\biggl[\Bigl(\frac2{3-x}\Bigr)^2\biggr]\+\cdots\biggr\} \+ \cdots \+2\Bigl(x-\frac12\Bigr) \,, \nn
\end{align}
which is consistent with the expansion of $x^2$ as $x\to1$.

\smallskip

\noindent {\bf Application II.} {\it The WK--BGW correspondence.}  The Hodge--BGW correspondence was 
 recently found in~\cite{YZ21}, which gives a relationship between 
 the special-Hodge partition function again with $q=-1/2$ 
and the generalized Br\'ezin--Gross--Witten (BGW) model (see \cite{Al18, BG, GW, MMS96}). 
Let $Z^{\scriptscriptstyle\rm cBGW}(x,{\bf r};\e)$
denote the generalized BGW partition function in the sense of~\cite{YZ21} 
(it is denoted by $Z(x,{\bf T};\hbar)$ in~\cite{YZ21}), which is an element in $\CC((\e^2))[[x+2]][[{\bf r}]]$.
Here $x$ is the {\it Alexandrov coupling constant}, $\e$ is an indeterminate, and ${\bf r}=(r_0,r_1,r_2,\dots)$ 
is an infinite tuple of indeterminates. 
In particular, we recall that 
\beq\label{276712}
\log Z^{\scriptscriptstyle\rm cBGW}(x,{\bf 0};\e) \= \frac{x^2}{4\,\e^2} \, \Bigl(\log \bigl(-\frac{x}2\bigr) - \frac32 \Bigr) \+ \frac{\log \bigl(-\frac{x}2\bigr)}{12} 
 - \sum_{g\ge2} \frac{\e^{2g-2} (-2)^{g-1} B_{2g}}{2g \, (2g-2) \, x^{2g-2}}
\eeq
For more details about the generalized BGW partition function 
see e.g.~\cite{DYZ21, YZ21, YZ23} and the references therein. 

The Hodge--BGW correspondence says that 
\beq\label{hodgebgw1107id}
Z_{\Omega(-1/2)}\bigl({\bf T}^{\scriptscriptstyle \rm Hodge-BGW}(x,{\bf r});\sqrt{-4} \e\bigr) \, e^{\frac{A_{\scriptscriptstyle \rm cBGW} (x,{\bf r})}{\e^2}} \= 
Z^{\scriptscriptstyle\rm cBGW}(x,{\bf r};\e) \,,
\eeq
where $A_{\scriptscriptstyle \rm cBGW} (x,{\bf r})$ is the quadratic series given by~\eqref{Abgw1107}, 
and 
\begin{align}
&T^{\scriptscriptstyle \rm Hodge-BGW}_i(x,{\bf r}) \label{Trhb1107}\\
& \quad \=  - \biggl(-\frac12\biggr)^{i-1} \+ \delta_{i,1} \+ x \, \delta_{i,0} - 
2 \, \sum_{j\ge0} \, \frac1{j!} \, \biggl(-\frac{2j+1}2\biggr)^i \, r_j \,, \quad i\geq 0\,. \nn
\end{align}

Both the WK partition function $Z^{\scriptscriptstyle \rm WK}({\bf t}; \epsilon)$ 
and the generalized BGW partition function $Z^{\scriptscriptstyle\rm cBGW}(x,{\bf r};\e)$ 
are particular $\tau$-functions for the KdV hierarchy~\cite{DYZ21}. However, 
we would like to remark that the Hodge--WK correspondence~\eqref{mainidentity} is different from 
the Hodge--BGW correspondence~\eqref{hodgebgw1107id}. 
This can be seen from the simple fact that  
the change of the independent variables~\eqref{Trhb1107} in the Hodge--BGW correspondence is 
NOT invertible, while that in the Hodge--WK correspondence is part of a group action and hence certainly invertible.
As far as we know the WK partition function 
$Z^{\scriptscriptstyle \rm WK}({\bf t}; \epsilon)$ cannot be 
obtained from the generalized BGW partition function $Z^{\scriptscriptstyle\rm cBGW}(x,{\bf r};\e)$ 
 by just shifting the independent variables (vice versa). However, these two different correspondences enable us 
 to establish a relationship between $Z^{\scriptscriptstyle \rm WK}({\bf t}; \epsilon)$ 
 and $Z^{\scriptscriptstyle\rm cBGW}(x,{\bf r};\e)$ (see Theorem~\ref{thmwkbgw} below).
Of course, some important connections between $Z^{\scriptscriptstyle \rm WK}({\bf t}; \epsilon)$ 
 and $Z^{\scriptscriptstyle\rm cBGW}(x,{\bf r};\e)$ were 
 already known: the genus zero parts of the WK partition function $Z^{\scriptscriptstyle \rm WK}({\bf t}; \epsilon)$ and the  
generalized BGW partition function $Z^{\scriptscriptstyle\rm cBGW}(x,{\bf r};\e)$ are different, 
but they are equal for genus bigger than or equal to~1 in a non-obvious way. Indeed, for $g\ge1$, 
the jet representations of the genus $g$ WK free energy and the genus $g$ 
 generalized BGW free energy are the same, which is actually the content of the Okuyama--Sakai conjecture~\cite{OS20} 
 proved in~\cite{YZ21, YZ23}.

We are ready to prove Theorem~\ref{thmwkbgw}.
\begin{proof}[Proof of Theorem~\ref{thmwkbgw}]
By using~\eqref{hodgebgw1107id} and~\eqref{mainidentity} with $q=-1/2$. 
\end{proof}

Denote  
\beq
g_m(x) \,:=\, \frac{(2m-1)!!}{2^m} \, x  \quad (m\geq2),
\eeq
and denote by $\langle \omega_{j_1}\dots\omega_{j_k} \rangle (x,\e)$ the generalized BGW correlators, i.e.,
\beq
\langle \omega_{j_1}\dots\omega_{j_k} \rangle (x,\e) \= 
\frac{\p^k \log Z^{\scriptscriptstyle\rm cBGW}(x,{\bf r};\e)}{\p r_{j_1} \dots \p r_{j_k}}\bigg|_{{\bf s}={\bf 0}} \,.
\eeq
Similarly to Proposition~\ref{taylorwkgue} we have the 
following proposition. 
\begin{prop}
There holds that 
\begin{align}
& \frac{\frac{x}3+\frac12}{\e^2} - \frac1{24} \log \Bigl(-\frac{x}2\Bigr) \label{correctionbgw712}\\
&+ \, \sum_{g,p\ge0} \,  (\sqrt{-4}\e)^{2g-2} \frac{(x+2)^p}{p!}  \sum_{\lambda\in\mathcal{P}_{3g-3+p}} 
\frac{\langle\tau_0^p \tau_{\lambda+1} \rangle_{g}}{{\rm mult}(\lambda)!}
\Bigl(-\frac{2}{x}\Bigr)^{2g-2+\ell(\lambda)+p} 
\prod_{s=1}^{\ell(\lambda)} g_{1+\lambda_s}(x) \nn \\
& = \, \log Z^{\scriptscriptstyle\rm cBGW}(x,{\bf 0};\e) \,, \nn
\end{align}
where the expression of 
$\log Z^{\scriptscriptstyle\rm cBGW}(x,{\bf 0};\e)$ is given in~\eqref{276712}, and for $k\ge1$ and $j_1,\dots,j_k\ge0$, 
\begin{align}
& - \frac{\delta_{k,1}}{\e^2} \frac1{j_1!} \Bigl(\frac{x}{2j_1+1} + \frac{1}{j_1+1}\Bigr) 
\+ \frac{\delta_{k,2}}{\e^2} \, \frac{1}{j_1! j_2! (j_1+j_2+1)} \label{289713} \\
& + \, \sum_{0\leq m_1\leq j_1\,, \dots, \, 0 \leq m_k\leq j_k}  V_{m_1,\dots,m_k}(x,\e) \, \prod_{i=1}^k E_{m_i, \, j_i}  \nn\\
&= \,  \langle\omega_{j_1}\dots \omega_{j_k}\rangle(x,\e) \,, \nn
\end{align}
where 
\beq
E_{m,j} \,:=\, -2 \, \frac{(-1)^m}{(j-m)!}\,,
\eeq and for $0\leq m_1\leq j_1$, \dots, $0\leq m_k\leq j_k$,
\begin{align}
&V_{m_1,\dots,m_k}(x,\e) \, := \, \sum_{g\ge0} (-4)^{g-1} \e^{2g-2} \sum_{p\ge0} \, \frac{(x+2)^p}{p!} \\
& \quad \times \sum_{\lambda\in\mathcal{P}_{3g-3+p+k-|m|}} 
\frac{\langle\tau_0^p \tau_{\lambda+1} \tau_{m_1}\dots\tau_{m_k}\rangle_{g}}{{\rm mult}(\lambda)!}
\Bigl(-\frac{2}{x}\Bigr)^{2g-2+\ell(\lambda)+k+p} 
\prod_{s=1}^{\ell(\lambda)} g_{1+\lambda_s}(x) \,. \nn
\end{align}
\end{prop}

For instance, using~\eqref{genuszero712} the coefficient of $\e^{-2}$ of the left-hand side of~\eqref{correctionbgw712} begins
\begin{align}
& x+\frac12 - \frac14 \, \biggl\{\frac{(x+2)^3}{3!} \biggl[-\frac2x\biggr] 
\+ \frac{(x+2)^4}{4!} \biggl[\Bigl(-\frac2x\Bigr)^3 \Bigl(\frac34 x\Bigr)\biggr] \\
&\+ \frac{(x+2)^5}{5!} \biggl[\Bigl(-\frac2x\Bigr)^4 \Bigl(\frac{5!!}8 x\Bigr)+
\frac6{2!}\Bigl(-\frac2x\Bigr)^5 \Bigl(\frac34 x\Bigr)^2\biggr] \+\cdots\biggr\} \,,\nn
\end{align}
which agrees with the expansion of 
$\frac{1}{4} x^2 \log \bigl(-\frac{x}2\bigr) - \frac38 x^2$
as $x\to -2$. Similarly, the coefficient of $\e^0$ of the left-hand side of~\eqref{correctionbgw712} begins
\begin{align}
&- \frac1{24} \log \Bigl(-\frac{x}2\Bigr) \+ (x+2) \biggl[\frac1{24}\Bigl(-\frac{2}{x}\Bigr)^2 \Bigl(\frac34 x\Bigr)\biggr] \\
&\+ 
\frac{(x+2)^2}{2!} \biggl[\frac1{24}\Bigl(-\frac{2}x\Bigr)^3\Bigl(\frac{5!!}8 x\Bigr) 
+ \frac16 \frac1{2!} \Bigl(-\frac{2}{x}\Bigr)^4 \Bigl(\frac34 x\Bigr)^2\biggr] \+ \cdots\,, \nn
\end{align}
which agrees with the expansion of $\frac{1}{12} \log \bigl(-\frac{x}2\bigr)$ as $x\to -2$.
The coefficient of~$\e^2$ of the left-hand side of~\eqref{correctionbgw712} begins
\begin{align}
& -4 \, \biggl[\langle\tau_4\rangle_2 \Bigl(-\frac{2}{x}\Bigr)^4 \Bigl(\frac{7!!}{16} x\Bigr) \+ 
\langle\tau_3\tau_2\rangle_2 \Bigl(-\frac{2}{x}\Bigr)^5
\Bigl(\frac{5!!}8 x\Bigr)\Bigl(\frac34 x\Bigr) \\
&\+ \frac{\langle\tau_2^3\rangle_2}{3!}\Bigl(-\frac{2}{x}\Bigr)^6\Bigl(\frac34 x\Bigr)^3\biggr] - 4 \, 
(x+2) \Bigl[\cdots\Bigr] \+ \cdots \,, \nn
\end{align}
which agrees with the expansion of $-\frac{1}{480} \frac1{x^2}$ as $x\to -2$.
Using~\eqref{genuszero712} the coefficient of $\e^{-2}$ of the left-hand 
side of~\eqref{289713} for $k=1$ and $j_1=1$ begins
\begin{align}
& -\frac14 \, E_{0,1} \biggl\{ \frac{(x-1)^2}{2!} \biggl[-\frac2{x}\biggr] \+ 
\frac{(x-1)^3}{3!}\biggl[\Bigl(-\frac2{x}\Bigr)^3 \Bigl(\frac34 x\Bigr) \biggr] \+ \cdots\biggr\} \\
& - \frac14 \, E_{1,1} \biggl\{\frac{(x-1)^3}{3!}\biggl[\Bigl(-\frac2{x}\Bigr)^2\biggr]\+\cdots\biggr\} \+ \cdots - \Bigl(\frac{x}3+\frac12\Bigr) \,, \nn
\end{align}
which is consistent with the expansion of $x^4/96$ as $x\to-2$.

\section{The Hodge mapping partition function} \label{section8}
In the previous sections, we have introduced a $\G$-action on infinite tuples, have defined for any $\varphi\in\G$ the 
WK mapping partition function associated to~$\varphi$, and have associated to it the WK mapping hierarchy. 
In this section, just like the previous constructions, for any $\varphi\in\G$
we define the {\it Hodge mapping partition function} $Z^{\varphi}_{\Omega(\boldsymbol{\sigma})}$ associated to~$\varphi$ 
by 
\beq
Z^{\varphi}_{\Omega(\boldsymbol{\sigma})}(\bT;\e) \= Z_{\Omega(\boldsymbol{\sigma})}\bigl(\bT . \varphi^{-1};\e\bigr) \,,
\eeq
which as we will see also has several nice properties. Here 
$Z_{\Omega(\boldsymbol{\sigma})}(\bt;\e)$ is the Hodge partition function defined in~\eqref{hodgepar1111}. 
Obviously, $Z^{\varphi}_{\Omega({\bf 0})}(\bT;\e) = Z^{\varphi}(\bT;\e).$

Recall that the Hodge partition function satisfies the dilaton equation:
\beq
\sum_{i\ge0} \, t_i \, \frac{\p Z_{\Omega(\boldsymbol{\sigma})}(\bt;\e)}{\p t_i} \+ \e \, \frac{\p Z_{\Omega(\boldsymbol{\sigma})}(\bt;\e)}{\p \e} 
\+ \frac1{24} \, Z_{\Omega(\boldsymbol{\sigma})}(\bt;\e)
\= \frac{\p Z_{\Omega(\boldsymbol{\sigma})}(\bt;\e)}{\p t_1}\,. \label{dilatonhodgeoriginal1111}
\eeq
It follows the dilaton equation for the Hodge mapping partition function:
\beq
\sum_{i\ge0} \, T_i \, \frac{\p Z^{\varphi}_{\Omega(\boldsymbol{\sigma})}(\bT;\e)}{\p T_i} 
\+ \e \, \frac{\p Z^{\varphi}_{\Omega(\boldsymbol{\sigma})}(\bT;\e)}{\p \e} 
\+ \frac1{24} \, Z^{\varphi}_{\Omega(\boldsymbol{\sigma})}(\bT;\e)
\= \frac{\p Z^{\varphi}_{\Omega(\boldsymbol{\sigma})}(\bT;\e)}{\p T_1}\,. \label{dilatonhodgemapping1111}
\eeq
The Virasoro constraints for the Hodge partition function can be found in~\cite{LYZZ22}. So it is possible to translate 
them to the Hodge mapping partition function, which we will do elsewhere.  

As before, the logarithm $\log Z^{\varphi}_{\Omega(\boldsymbol{\sigma})}(\bT;\e) =: \F^{\varphi}_{\Omega(\boldsymbol{\sigma})}(\bT;\e)$,
called the {\it Hodge mapping free energy}, has a genus expansion:
\beq
\F^{\varphi}_{\Omega(\boldsymbol{\sigma})}(\bT;\e) \,=:\, \sum_{g\ge0} \, \e^{2g-2} \F^{\varphi}_{\Omega_g(\boldsymbol{\sigma})}(\bT) \, .
\eeq
We call $\F^{\varphi}_{\Omega_g(\boldsymbol{\sigma})}(\bT)$, $g\ge0$, {\it the genus $g$ Hodge mapping free energy}.
By Theorem~\ref{thmgenus0} and the well-known fact
$\F_{\Omega_0(\boldsymbol{\sigma})}(\bt)\equiv\F_0^{\scriptscriptstyle \rm WK}(\bt)$ 
we immediately obtain the following 
\begin{prop}\label{propgenus01011}
For any $\varphi\in\G$, we have $\F^\varphi_{\Omega_0(\boldsymbol{\sigma})}=\F_0^{\scriptscriptstyle\rm WK}$.
\end{prop}
For genus bigger than or equal to~1, it is known from e.g.~\cite{DLYZ16, DY20-2} that the genus $g$ ($g\ge1$) Hodge free energy 
$\F_{\Omega_g(\boldsymbol{\sigma})}(\bt)$ has the $(3g-2)$-jet representation, i.e., 
there exists $F_{\Omega_g(\boldsymbol{\sigma})}(v_0, v_1,\dots,v_{3g-2}; \boldsymbol{\sigma})$, such that 
\beq\label{Hodge3gminus21111}
\F_{\Omega_g(\boldsymbol{\sigma})}(\bt) \= 
F_{\Omega_g(\boldsymbol{\sigma})}\biggl(E(\bt), \frac{\p E(\bt)}{\p t_0}, \dots, \frac{\p^{3g-2} E(\bt)}{\p t_0^{3g-2}}; \boldsymbol{\sigma}\biggr)\,, 
\quad g\geq1\,,
\eeq 
where $E(\bt)$ is defined by~\eqref{defVandvf1001}.
We then have the following 
\begin{prop}\label{jetreprephodgemap1111}
For $g=1$ we have the identity:
\beq\label{genus1K1hodge1111}
\F^\varphi_{\Omega_1(\boldsymbol{\sigma})}(\bT) \= 
F^\varphi_{\Omega_1(\boldsymbol{\sigma})}\biggl(E(\bT), \frac{\p E(\bT)}{\p X}; \boldsymbol{\sigma}\biggr) \,, 
\eeq
with
\beq
F^\varphi_{\Omega_1(\boldsymbol{\sigma})}(V,V_1; \boldsymbol{\sigma}) \,:=\, \frac{1}{24} \log V_1 \+ \frac1{16} \log \varphi'(V) 
\+ \frac{\sigma_1}{24} \, \varphi(V)\,.
\eeq
For each $g\ge2$, there exists $F^\varphi_{\Omega_g(\boldsymbol{\sigma})}(V_0, \dots, V_{3g-2}; \boldsymbol{\sigma})$ that is a 
polynomial of $\sigma_1,\dots,\sigma_{2g-1}$, 
$V_2$, \dots, $V_{3g-2}$ and a rational function of~$V_1$, such that  
\beq\label{jetfvg89new}
\F^\varphi_{\Omega_g(\boldsymbol{\sigma})}(\bT) \= 
F^\varphi_{\Omega_g(\boldsymbol{\sigma})}\biggl(E(\bT), \dots, \frac{\p^{3g-2} E(\bT)}{\p X^{3g-2}}; \boldsymbol{\sigma}\biggr) \,.
\eeq
\end{prop}

Let 
\beq\label{defU118}
U^{\varphi}_{\Omega(\boldsymbol{\sigma})}(\bT;\e) \,:=\, 
\e^2 \, \frac{\p^2 \F^\varphi_{\Omega(\boldsymbol{\sigma})}(\bT;\e)}{\p X^2} \,, 
\eeq
where $X=T_0$. This gives a quasi-Miura transformation 
\beq\label{qMhodge69}
V \;\mapsto\; U^{\varphi}_{\Omega(\boldsymbol{\sigma})} \= V \+ \sum_{g\ge1} \, \e^{2g} \, \p^2 \bigl(F^\varphi_{\Omega_g(\boldsymbol{\sigma})}\bigr)\,,
\eeq
which transforms the abstract local RH hierarchy $D_S(v) = S(v) \, v_1$ to
\begin{align}
& D_{S}  (U) \=  S\, U_1 \+ 
\e^2 \biggl(\frac{S'}{12} \, U_3 \+ \biggl(\frac{\varphi'' S'}{8 \varphi'} +\frac{S''}{6} + \frac{\sigma_1}{12} \varphi' S'\biggr) \, 
U_1 \, 
U_2  \label{abslocH1224}\\
& \qquad \qquad \quad +\, \biggl(-\frac{ \varphi''^2 S'}{16 \varphi'^2}
+\frac{\varphi''' S'+\varphi'' S''}{16 \varphi'} +\frac{S'''}{24} 
+ \frac{\sigma_1}{24}  \bigl(\varphi' S'\bigr)' \biggr) \, U_1^3\biggr) \+ \cdots \,, \nn
\end{align}
where $U=U^{\varphi}_{\Omega(\boldsymbol{\sigma})}$, and 
we omitted the arguments $U^{\varphi}_{\Omega(\boldsymbol{\sigma})}$ from $\varphi', \varphi''$, $\cdots$ 
and from $S, S', S''$, $\cdots$.
We call~\eqref{abslocH1224} the {\it abstract local Hodge mapping hierarchy associated to~$\varphi$}.
When $\varphi=id$, we call~\eqref{abslocH1224} the {\it abstract local Hodge hierarchy}.

Define 
$\Omega_{S_1(w),S_2(w)}^{\scriptscriptstyle \rm Hodge}$, $S_1(w),S_2(w)\in \mathcal{O}_c(w)$, as the 
substitution of the inverse of the quasi-Miura type transformation
$v \mapsto w = v + \sum_{g\ge1} \, \e^{2g} \, \p^2 \bigl(F_{\Omega_g(\boldsymbol{\sigma})}\bigr)$ 
in~$\int^w S_1 S_2+ \sum_{g\ge1} \e^{2g} D_{S_1} D_{S_2} (F_{\Omega_g(\boldsymbol{\sigma})})$. Similarly, 
 define 
$\Omega_{S_1(U),S_2(U)}^{\varphi,\, {\scriptscriptstyle \rm Hodge}}$, $S_1(U),S_2(U)\in \mathcal{O}_c(U)$, as the 
substitution of the inverse of~\eqref{qMhodge69} 
in $\int^U S_1 S_2+ \sum_{g\ge1} \e^{2g} D_{S_1(U)} D_{S_2(U)} (F^\varphi_{\Omega_g(\boldsymbol{\sigma})})$.

The following theorem, which is a refinement of Theorem~\ref{thmmain2}, gives a generalization of Theorem~\ref{conj0307} 
and some results in~\cite{BPS12-1, BPS12-2, DLYZ16}.
\begin{theorem}\label{conj060923}
The abstract local Hodge mapping hierarchy~\eqref{mappinghierarchy}
have {\it polynomiality}: for any $S$
the right-hand side of~\eqref{mappinghierarchy} belongs to 
$\mathcal{A}_{U^\varphi_{\Omega(\boldsymbol{\sigma})}}[[\e^2]]_{1}$. 
Moreover, the elements $\Omega_{S_1(U), \, S_2(U)}^{\varphi,\, {\scriptscriptstyle \rm Hodge}}$, $S_1(U),S_2(U)\in \mathcal{O}_c(U)$, belong to $\mathcal{A}_{U^\varphi_{\Omega(\boldsymbol{\sigma})}}[[\e^2]]_0$.
\end{theorem}
Let us first prove Theorem~\ref{conj060923} for the case when $\varphi={\rm id}$. Indeed, 
similarly to the proof of Proposition~\ref{localkdv63}, by using the properties of 
the $\tau$-symmetric hamiltonian densities of the Hodge hierarchy \cite{DLYZ16} and the 
results in~\cite{Buryak, BPS12-1, BPS12-2}, 
we arrive at the following proposition.
\begin{prop}\label{localhodge69}
Theorem~\ref{conj060923} holds when $\varphi={\rm id}$.  Moreover, the elements 
$\Omega_{S_1(w), \, S_2(w)}^{\scriptscriptstyle \rm Hodge}$, $S_1(w),S_2(w)\in \mathcal{O}_c(w)$, 
belong to $\mathcal{A}_{w}[[\e^2]]_0$.
\end{prop}

\begin{proof}[Proof of Theorem~\ref{conj060923}]
First, 
\beq
\p \= \sum_{m\ge0} \frac{\p t_m}{\p X} D_{w^m/m!} \=  D_{\sqrt{\varphi'(\varphi^{-1}(w))}}\,.
\eeq
Here $D_{w^m/m!}$, $m\ge0$, are derivations of the abstract Hodge hierarchy. 
By Proposition~\ref{localhodge69} the element $D_{\sqrt{\varphi'(\varphi^{-1}(w))}}(w)$ has polynomiality.
Note that  
\beq
\Omega^{\varphi,\, {\scriptscriptstyle \rm Hodge}}_{U^i/i!, \, U^j/j!} \= 
\sum_{i_1, \, j_1\ge0} \frac{\p t_{i_1}}{\p T_i} \frac{\p t_{j_1}}{\p T_j} \, 
\Omega^{\scriptscriptstyle \rm Hodge}_{w^{i_1}/i_1!, \, w^{j_1}/j_1!} \,, \quad i,j\ge0\,.
\eeq
By an iteration, the $\p_x$-flow for $w$ with $\p=\p_X$ as the spatial derivative 
is an evolutionary PDE in Dubrovin--Zhang's normal form. 
Since $\Omega^{\scriptscriptstyle \rm Hodge}_{w^{i_1}/i_1!, \, w^{j_1}/j_1!}\in \mathcal{A}_w[[\e^2]]_0$
and by substituting the $\p_x$-flow, 
we find that $\Omega^{\scriptscriptstyle \rm Hodge}_{w^{i_1}/i_1!, \, w^{j_1}/j_1!}$ are power series of $\e^2$ 
with coefficients being polynomials of $\p_X(w)$, $\p_X^2(w)$, \dots, so are 
$\Omega^{\varphi,\, {\scriptscriptstyle \rm Hodge}}_{U^i/i!, \, U^j/j!}$. This implies in particular 
that $U=U^{\varphi}_{\Omega(\boldsymbol{\sigma})}  
= \Omega^{\varphi,\, {\scriptscriptstyle \rm Hodge}}_{1,1} = \varphi^{-1}(w) + \dots$ 
gives a Miura-type transformation. So 
$\Omega^{\varphi,\, {\scriptscriptstyle \rm Hodge}}_{U^i/i!, \, U^j/j!}\in \mathcal{A}_U[[\e^2]]_0$, 
and thus 
$\Omega^{\varphi,\, {\scriptscriptstyle \rm Hodge}}_{S_1(U), \, S_2(U)}\in \mathcal{A}_U[[\e^2]]_0$.
Finally, 
$
D_{S(U)} (U) = 
D_{S(U)} \bigl(\Omega^{\varphi,\, {\scriptscriptstyle \rm Hodge}}_{1, \, 1}\bigr) = 
\p \bigl(\Omega^{\varphi,\, {\scriptscriptstyle \rm Hodge}}_{1, \, S(U)}\bigr) \in \mathcal{A}_U[[\e^2]]_{1}$.
\end{proof}
We also verified Theorem~\ref{conj060923} directly up to and including terms of~$\e^8$. 

By using again the definition (i.e., using the quasi-Miura map), we find that the abstract local Hodge mapping hierarchy~\eqref{mappinghierarchy} has the more precise form:
\beq\label{mappinghierarchydivergencehodge}
D_{S} (U) \= 
\p \Biggl(\int^{U} S \+ \sum_{g\ge1} \e^{2g} \sum_{\lambda \in \mathcal{P}_{2g}} 
\sum_{j=1}^{\ell(\lambda)+g-1} Y^\varphi_{\lambda,j}(l_1(U),\dots; m_1(U), \dots) \, S^{(j)}(U)\, U_\lambda\Biggr),
\eeq
where $U=U^{\varphi}_{\Omega(\boldsymbol{\sigma})}$,  
$Y^\varphi_{\lambda,j}(\ell_1,\dots; \rho_1, \dots)$ are weighted homogeneous polynomials 
of degree $\ell(\lambda)+g-1-j$ in 
variables $\ell_i$ and $\rho_i$ of weight~$i$ ($i\ge 1$), $l_i(U)$ are defined in~\eqref{deflkv63}, and $m_i(U)=\sigma_{2i-1} \varphi'(U)^i$.

The abstract local Hodge mapping hierarchy~\eqref{abslocH1224} can also be written in the form 
\beq\label{hodgemappinglocalham69}
D_{S} (U) \= P_1^\varphi(U) 
\biggl( \frac{\delta \int h^\varphi_{1;S}}{\delta U} \biggr) \,, \quad S\in \mathcal{O}_c\,,
\eeq
where $U=U^{\varphi}_{\Omega(\boldsymbol{\sigma})}$, 
$P_1^\varphi(U)$ is the operator given by 
\beq\label{Poisson1112}
P_1^\varphi(U) 
\,:=\, \sum_{k,\ell\ge0} \, (-1)^\ell \, \frac{\p U}{\p V_k} \circ \p^k \circ \biggl(\frac12 \, 
\frac{1}{\varphi'(V)} \circ \p \+ \frac12 \, \p \circ \frac{1}{\varphi'(V)} \biggr) \circ \p^\ell \circ \frac{\p U}{\p V_\ell} \,,
\eeq
and the hamiltonian density $h^\varphi_{1;S}$ is understood as the substitution of the inverse of 
the quasi-Miura transformation~\eqref{qMhodge69} into~\eqref{h1s69}. As before, $P^\varphi(U)$ 
has the form:
\begin{align}
& P_1^\varphi(U) \= \sum_{g\ge0} \, \e^{2g} \, P^{\varphi,[g]}_{1;\Omega(\boldsymbol{\sigma})} \,, \quad 
P^{\varphi,[0]}_{1;\Omega(\boldsymbol{\sigma})} 
\= \frac{1}{2\,\varphi'(U)} \circ \p \+ \p \circ \frac{1}{2\,\varphi'(U)}\,, \label{explicitp1hodge}\\
& P^{\varphi,[g]}_{1;\Omega(\boldsymbol{\sigma})} \= \sum_{j=0}^{3g+1} \, A^\varphi_{2g,j;\Omega(\boldsymbol{\sigma})} \, \p^j\,, 
\quad A_{2g,j;\Omega(\boldsymbol{\sigma})}^{\varphi} \in \mathcal{O}_c(U)\bigl[U_1,\dots,U_{3g+1},U_1^{-1}\bigr][\boldsymbol{\sigma}]\,,\\
& \sum_{m\ge1} \, m \, U_m \, \frac{\p A^\varphi_{2g,j;\Omega(\boldsymbol{\sigma})} }{\p U_m} 
\= (2g+1-j) A^\varphi_{2g,j;\Omega(\boldsymbol{\sigma})}  \,.
\end{align}
We have the following conjecture. 
\begin{conjecture}\label{mainconjecturehodge} 
For $g\ge0$ and $0\leq j\leq 3g+1$, the elements $A^\varphi_{2g,j;\Omega(\boldsymbol{\sigma})}$ all belong to $\mathcal{A}_{U}^{[2g+1-j]}$. 
Moreover, for $i\ge0$, the variational derivatives of the hamiltonians $\int h^\varphi_{1;S}$ 
with respect to~$U$ belong to $\mathcal{A}_{U}[[\e]]$.
\end{conjecture}

Motivated by the Hodge universality conjecture 
proposed in~\cite{DLYZ16} (see Remark~\ref{remark12-717}) and the classification work mentioned in Section~\ref{section6}, 
we propose the following {\it Hodge mapping universality conjecture}.

\begin{conjecture} \label{Hmuconj226}
The abstract local Hodge mapping hierarchy is a universal object for hamiltonian perturbations of the abstract local RH hierarchy possessing  
a $\tau$-structure.
\end{conjecture}

Conjecture~\ref{Hmuconj226} generalizes Theorem~\ref{WKmuconj226} as well as 
the Hodge universality conjecture from~\cite{DLYZ16}.
Let us verify 
Conjecture~\ref{Hmuconj226} {\it directly}
up to and including terms of order~$8$ in~$\epsilon$. Indeed, the following Miura-type transformation 
\begin{align}
& w \= M(U) \+ \sum_{k=1}^4 \e^{2k} \sum_{\lambda \in \mathcal{P}_{2k}} C_\lambda(U) \, U_{\lambda}  \+ \mathcal{O}(\e^{10})
\end{align}
transforms the abstract local Hodge mapping hierarchy~\eqref{abslocH1224}
to the standard form~\eqref{standardform1113} up to~$\e^8$, with $U=U^{\varphi}_{\Omega(\boldsymbol{\sigma})}$,
\beq
M(U) \= \int_0^{U} \sqrt{\varphi'(y)} \, dy \,,
\eeq
\beq
a_0(w) \= M'(M^{-1}(w)) \,,
\eeq
and 
\begin{align}
& C_{(2)} (U) \= -\frac{\sigma_1}{24} \, \varphi'(U)^{3/2}\,, \nn\\
& C_{(1^2)} (U) \= - \frac{\sigma_1}{24}  \sqrt{\varphi'(U)} \, \varphi''(U)  
\+ \frac{\varphi''(U)^2}{24 \, \varphi'(U)^{3/2}}
 - \frac{\varphi^{(3)}(U)}{48 \sqrt{\varphi'(U)}}  \,, \nn\\
& C_{(4)} (U) \= -\frac{\sigma_1}{240}  \sqrt{\varphi'(U)} \, \varphi''(U) 
\+ \frac{\sigma_1^2}{1920}  \varphi'(U)^{5/2} 
\+\frac{\varphi''(U)^2}{384 \, \varphi'(U)^{3/2}}
-\frac{\varphi^{(3)}(U)}{480 \sqrt{\varphi'(U)}}
 \,,\nn\\
& \dots \,, \nn\\
& C_{(1^8)}(U) \= - \frac{107}{185794560} \, \frac{\varphi^{(12)}(U)}{\sqrt{\varphi'(U)}} 
\+ ~{\rm more~than~two~hundred~terms} \,. \nn
\end{align}
Here the beginning relationships between 
the classification invariants $q_1,q_2$, $\cdots$ and the Chern-Hodge-Mumford parameters $\sigma_1,\sigma_3$, $\cdots$ are given by
\begin{align}
& q_1 \= \frac{\sigma_1}{2^5 \, 3^2 \, 5^1} \,, \qquad q_2 \=  \frac{2 \, \sigma_1^3-\sigma_3}{2^{10} \, 3^5 \, 5^1} \,, 
\qquad 
 q_3 \= \frac{16 \, \sigma_1^5 - 20 \, \sigma_1^2\sigma_3 + \sigma_5}{2^{13} \, 3^6 \, 5^2 \, 7^1}\,. \label{qsigma1119} 
\end{align}
We note that the relations in~\eqref{qsigma1119} coincide 
with the ones given in~\cite{DLYZ16} (see~also~\cite{BDGR20}). Note that in~\cite{DLYZ16, BDGR20} only the 
case with \hbox{$a_0(w) \equiv 1$} (i.e., the case 
with \hbox{$\varphi(V)=V$}) was considered. But the results in this paper show that the above beginning relations
\eqref{qsigma1119} do not depend on~$\varphi$. In general, this independence of~$\varphi$ is   
expected.   
Note that equations~\eqref{qsigma1119} 
specialize to~\eqref{q2q31129} when the $\sigma$'s are specialized by~\eqref{sigmaspecial1111}.

\begin{remark}
For each CohFT, A.~Buryak~\cite{Buryak15} defined the {\it double ramification (DR) hierarchy}, which 
is a $\tau$-symmetric hamiltonian system~\cite{BDGR18, BDGR20}. For the trivial
case (the case when the CohFT is given by $\Omega({\bf 0})=1$), the DR hierarchy coincides with the KdV hierarchy. 
For the Hodge CohFT $\Omega(\boldsymbol{\sigma})$ (see~\eqref{chernchrank1omega}),
it is conjectured in~\cite{Buryak15} and refined in~\cite{BDGR18, DLYZ16} that 
the DR hierarchy associated to $\Omega(\boldsymbol{\sigma})$ is normal Miura-type 
equivalent~\cite{DLYZ16, DZ-norm} to the Hodge hierarchy. Later it is shown by 
Buryak, Dubrovin, Gu\'er\'e and Rossi~\cite{BDGR20} that 
the DR hierarchy is the standard deformation with $a_0(w)\equiv 1$, 
and moreover, by an explicit computation in the DR side they obtained the following conjectural 
relations between $q$'s and $\sigma$'s when $a_0(w)\equiv1$:
\beq\label{BDGRrelation}
q_{g-1} \= (3g-2) \, \int_{\overline{\mathcal{M}}_{g,0}} 
\lambda_g \exp \biggl(\sum_{j\geq1} \, \sigma_{2j-1} \, {\rm ch}_{2j-1}(\mathbb{E}_{g,0}) \biggr) \,, \quad g\ge2\,.
\eeq
By the discussion given right above this remark, we conjecture this holds for all $\varphi$ which makes the discussion 
more explicit. Using formula~\eqref{BDGRrelation} and the algorithm in~\cite{DLYZ16} 
for computing Hodge integrals (or the Hodge--GUE 
correspondence~\cite{DLYZ16, DLYZ20, DY17}), we can compute more 
explicit values for $q_i$ in the following table:
\begin{align}
& \arraycolsep=5.1pt\def\arraystretch{1.5}
\begin{array}{?c?c|c|c|c|c|c|c?}
\Xhline{2\arrayrulewidth}  i &  1 & 2 & 3 & 4 & 5 & 6 &7   \\
\Xhline{2\arrayrulewidth}  q_i &  \frac{q}{2^5 3^1 5^1} & \frac{q^3}{2^7 3^4 5^1}  & 0 & \frac{-13 \, q^7}{2^{10} 3^4 5^2 7^1 11^1} 
& \frac{-59 \, q^9}{2^{5} 3^7 5^2 7^2 11^1 13^1} & \frac{19 \, q^{11}}{2^{11} 3^4 5^1 7^2 11^1 13^1} 
& \frac{1493 \, q^{13}}{2^{9} 3^7 5^3 7^2 13^1 17^1}   \\
\Xhline{2\arrayrulewidth}
\end{array} \nn 
\end{align}
\end{remark}

\medskip

\section{The generalized Hodge--WK correspondence}\label{section9}
In this section, by using the $\G$-action and the Hodge--WK correspondence we obtain 
 explicit relationships between the WK mapping partition functions and the special-Hodge mapping partition 
 functions, and we investigate bihamiltonian structures for the Hodge mapping hierarchy.  

\begin{theorem}\label{thm1114}
The special-Hodge mapping partitions and the WK mapping partition functions are related by 
\beq\label{ZZss1112}
Z^\psi_{\Omega^{\rm special}(q)} \= Z^\varphi \,,
\eeq
where the two power series $\psi$ and $\varphi$ are related by~\eqref{phipsi1122}.
\end{theorem}

\begin{proof}
Recall from Section~\ref{sectionexample223} that the Hodge--WK correspondence says
$$
Z_{\Omega^{\rm special}(q)}(\bT;\e) \= Z^{\scriptscriptstyle \rm WK}\bigl(\bT . \varphi_{\rm special}^{-1};\e\bigr) \,,
$$
where we recall that $\varphi_{\rm special}$ is defined as in~\eqref{examplephi}.
Therefore, 
$$
Z_{\Omega^{\rm special}(q)}\bigl(\bT . \psi^{-1};\e\bigr) \= 
 Z^{\scriptscriptstyle \rm WK}\bigl(\bT . \psi^{-1} \circ \varphi_{\rm special}^{-1};\e\bigr)
 \= 
 Z^{\scriptscriptstyle \rm WK}\bigl(\bT . (\varphi_{\rm special} \circ \psi)^{-1};\e\bigr) \,.
$$
The theorem is proved.
\end{proof}

We call~\eqref{ZZss1112} the {\it generalized Hodge--WK correspondence}.  
From the definition, an alternative form of~\eqref{ZZss1112} is 
\beq
Z^\psi_{\Omega^{\rm special}(q)}(\bt . \varphi;\e) \= Z^{\scriptscriptstyle \rm WK} (\bt;\e) \,,
\eeq
where $\varphi$ and $\psi$ are related by~\eqref{phipsi1122}.

Let us consider the Poisson geometry behind 
this theorem.
Indeed, via a bihamiltonian test, we find that 
up to order~$\e^8$,
the Hodge mapping hierarchy associated to 
an arbitrarily given group element $\psi\in\G$ is bihamiltonian if and only if its parameters have the specific values
\beq\label{2041114}
 \sigma_1 \= 3 \, q \,, \quad \sigma_3 \= 30 \, q^3\,, \quad \sigma_5 \= 1512 \, q^5\,, \quad \sigma_7 \= 183600 \, q^7\,.
\eeq
This specialization is remarkable because it does not depend on~$\psi$.
For the case when $\psi(V) = V$, we already know that the answer is 
the special-Hodge specialization~\eqref{sigmaspecial1111}, conjectured\footnote{About this known conjecture, the sufficiency part 
is proved by the Hodge--GUE correspondence~\cite{DLYZ20} but the necessity part is still open.} in~\cite{DLYZ16}. 
So we expect that the Hodge mapping hierarchy associated to 
$\psi\in\G$ is bihamiltonian if and only if $\sigma_{2j-1}=\sigma_{2j-1}^{\rm special} \= (4^j-1) \, (2j-2)! \, q^{2j-1}$, $j\geq1$, 
of which the first four values are the ones given in equation~\eqref{2041114}.
We call the Hodge mapping hierarchy associated to~$\psi$ with this specialization the {\it special-Hodge mapping hierarchy associated to~$\psi$}. The following corollary gives the sufficiency part. 
\begin{cor}
The special-Hodge mapping hierarchy associated to~$\psi$ 
has a bihamiltonian structure with the Poisson pencil $Q_2^\psi\bigl(U_{\Omega^{\rm special}(q)}^\psi\bigr) 
+ \lambda \, Q_1^\psi\bigl(U_{\Omega^{\rm special}(q)}^\psi\bigr)$ given by   
\beq
Q_1^\psi\bigl(U_{\Omega^{\rm special}(q)}^\psi\bigr) \,:=\, \sum_{k,\ell\ge0} \, (-1)^\ell \, 
\frac{\p U_{\Omega^{\rm special}(q)}^\psi}{\p V_k} \circ \p^k \circ Q^{\psi,[0]}(V) 
\circ \p^\ell \circ \frac{\p U_{\Omega^{\rm special}(q)}^\psi}{\p V_\ell} \,,
\eeq
\beq
Q_1^{\psi,[0]}(V) \,:=\, \frac12 \, \frac{e^{-2 \, q \, \psi(V)}}{\psi'(V)} \circ \p \+ \frac12 \, \p \circ \frac{e^{-2 \, q \, \psi(V)}}{\psi'(V)} \,,
\eeq 
\beq
Q_2^\varphi\bigl(U_{\Omega^{\rm special}(q)}^\psi\bigr) \,:=\, \sum_{k,\ell\ge0} \, (-1)^\ell \, \frac{\p U_{\Omega^{\rm special}(q)}^\varphi}{\p V_k} \circ \p^k \circ Q_2^{\psi,[0]}(V) \circ \p^\ell \circ \frac{\p U_{\Omega^{\rm special}(q)}^\psi}{\p V_\ell} \,,
\eeq
\beq
Q_2^{\psi,[0]}(V) \,:=\, \frac12 \, 
\frac{1-e^{-2 \, q \, \psi(V)}}{2 \, q\, \psi'(V)} \circ \p \+ \frac12 \, \p \circ \frac{1-e^{-2 \, q \, \psi(V)}}{2 \, q\, \psi'(V)} \,.
\eeq 
\end{cor}
\begin{proof}
By Theorem~\ref{mainconjecture} and Theorem~\ref{thm1114}.
\end{proof}
Note that, by definition, the Schouten bracket of 
$Q_2^\psi\bigl(U_{\Omega^{\rm special}(q)}^\psi\bigr) \+ \lambda \, Q_1^\psi\bigl(U_{\Omega^{\rm special}(q)}^\psi\bigr)$ 
and itself vanishes identically in~$\lambda$, so 
 the non-trivial part of the above corollary is about the polynomial dependence of the 
 coefficients of both $Q_1^\psi\bigl(U_{\Omega^{\rm special}(q)}^\psi\bigr)$ and $Q_2^\psi\bigl(U_{\Omega^{\rm special}(q)}^\psi\bigr)$. We also verified the polynomiality directly up to and including the terms of order~8 in~$\e$.
It also follows from 
Theorem~\ref{mainconjecture}, Theorem~\ref{thm1114} and the computation for~\eqref{203719} that 
for any $\psi\in\G$, the central invariant of the Poisson pencil 
$
Q_2^\psi\bigl(U_{\Omega^{\rm special}(q)}^\psi\bigr) 
+ \lambda \, Q_1^\psi\bigl(U_{\Omega^{\rm special}(q)}^\psi\bigr)
$
is $1/24$ identically in~$q$.

There can be choices for 
$Q_a^\psi\bigl(U_{\Omega^{\rm special}(q)}^\psi\bigr)$, $a=1,2$, for a pencil. Our choice satisfies 
\beq
Q_2^\psi\bigl(U_{\Omega^{\rm special}(q)}^\psi\bigr) \+ \frac1{2\,q} \, Q_1^\psi\bigl(U_{\Omega^{\rm special}(q)}^\psi\bigr) 
\= P^\psi\bigl(U_{\Omega^{\rm special}(q)}^\psi\bigr)\,,
\eeq
where $P^\psi\bigl(U_{\Omega^{\rm special}(q)}^\psi\bigr)$ is defined in~\eqref{Poisson1112}. 
Note that we did not choose either the Poisson operator 
$Q_1^\psi\bigl(U_{\Omega^{\rm special}(q)}^\psi\bigr)$ or 
$Q_2^\psi\bigl(U_{\Omega^{\rm special}(q)}^\psi\bigr)$ to simply be $P^\psi\bigl(U_{\Omega^{\rm special}(q)}^\psi\bigr)$, 
but we choose them to match with the Poisson pencil for the bihamiltonian structure for the WK mapping hierarchy,  
along the generalized Hodge--WK correspondence. 
For the particular case when $\psi(V)=V$, a similar but different choice was made in~\cite{DLYZ16}, 
where $Q_2^\psi\bigl(U_{\Omega^{\rm special}(q)}^\psi\bigr)$ was chosen 
to be $-P^\psi\bigl(U_{\Omega^{\rm special}(q)}^\psi\bigr)$ and $Q_1^\psi\bigl(U_{\Omega^{\rm special}(q)}^\psi\bigr)$ was  
chosen the same as above, giving rise also to the central invariant $1/24$. 

Before ending the paper, we would like to mention a generalization of 
part of our constructions to semisimple Frobenius manifolds. This will be studied in a subsequent publication.

Let $M$ be an $n$-dimensional calibrated semisimple Frobenius manifold. Denote by $Z_M(\bt)$ 
and $Z_{M, \, \Omega(\boldsymbol{\sigma})}(\bt)$ 
the topological partition function of~$M$ and the Hodge partition function of~$M$, respectively. 
Here $\bt=(t^{\alpha,k})_{\alpha=1,\dots,n, k\ge0}$ is an infinite tuple of indeterminates. 
Recall that the integrable hierarchies corresponding to 
the partition functions $Z^M$ and $Z_{M, \, \Omega(\boldsymbol{\sigma})}$ are the 
{\it Dubrovin--Zhang hierarchy of~$M$} ({\it aka the integrable hierarchy of topological type of~$M$}) and the {\it Hodge hierarchy of~$M$}, 
respectively. The logarithm $\log Z_M=:\F_M$ is called the {\it topological free energy of~$M$},
 and $\log Z_{M, \, \Omega(\boldsymbol{\sigma})}=:\F_{M, \, \Omega(\boldsymbol{\sigma})}$ the {\it Hodge free energy of~$M$}.
 Both $\F_M$ and $\F_{M, \, \Omega(\boldsymbol{\sigma})}$ have genus expansions:
 \beq
 \F_M(\bt;\e) \= \sum_{g\ge0} \e^{2g-2} \F_{M, \, g}(\bt) \,,\quad 
 \F_{M, \, \Omega(\boldsymbol{\sigma})} (\bt;\e) \= \sum_{g\ge0} \e^{2g-2} \F_{M, \, \Omega(\boldsymbol{\sigma}), \, g}(\bt)\,.
 \eeq
In this more general context, the group $\G$ is replaced by a more general group of 
affine-linear transformations 
 such that 
$\F_{M, \, 0}(\bt)=\F_{M, \, \Omega(\boldsymbol{\sigma}), \, 0}(\bt)$ is invariant under the transformation. 
For any such transformation~$\varphi$ we define
the {\it mapping partition function of~$M$ associated to~$\varphi$} 
as before by $Z^\varphi_M(\bT;\e):=Z_M(\bT. \varphi^{-1};\e)$ 
and the {\it Hodge mapping partition function of~$M$ associated to~$\varphi$} 
by $Z^\varphi_{M, \, \Omega(\boldsymbol{\sigma})}(\bT;\e):=Z_{M, \, \Omega(\boldsymbol{\sigma})}(\bT.\varphi^{-1};\e)$.
Since $Z^\varphi_{M, \, \Omega({\bf 0})}(\bT;\e)=Z^\varphi_M(\bT;\e)$, 
it is enough to study the Hodge mapping partition function of~$M$.
Let $X=T^{1,0}$ and let 
\beq
U^{\varphi}_{\alpha, \, M, \, \Omega(\boldsymbol{\sigma})}(\bT;\e) \,:=\, 
\e^2 \, \frac{\p^2 F^\varphi_{M, \, \Omega(\boldsymbol{\sigma})}(\bT;\e)}{\p X \p T^{\alpha,0}} \,, \quad \alpha=1,\dots,n\,.
\eeq
By the arguments similar to the proof of Theorem~\ref{conj060923}, we know that 
$U^{\varphi}_{\alpha, \, M, \, \Omega(\boldsymbol{\sigma})}(\bT;\e)$, $\alpha=1,\dots,n$, satisfy an integrable hierarchy
of evolutionary PDEs, which we call the {\it Hodge mapping hierarchy of~$M$ associated to~$\varphi$}. 
We expect that this hierarchy is hamiltonian. In 
particular, when $\boldsymbol{\sigma}={\bf 0}$ we call it the {\it Dubrovin--Zhang mapping hierarchy of~$M$ associated to~$\varphi$}, which is bihamiltonian for reasons similar to the 
proof of Theorem~\ref{mainconjecture} (cf.~\cite{DZ-norm, LWangZ21, LWangZ}).
We also call $Z^\varphi_{M, \, \Omega^{\rm special}(q)}(\bT;\e)$ the {\it special-Hodge mapping partition function of~$M$ associated to~$\varphi$}, 
and the integrable hierarchy satisfied by $U^{\varphi}_{\alpha, \, M, \, \Omega^{\rm special}(q)}(\bT;\e)$ is 
called the {\it special-Hodge mapping hierarchy of~$M$ associated to~$\varphi$}.

\medskip
\medskip

\noindent Di Yang

\noindent School of Mathematical Sciences, University of Science and Technology of China,

\noindent Jinzhai Road 96, Hefei 230026, P.R. China 

\noindent diyang@ustc.edu.cn

\medskip
\medskip

\noindent Don Zagier

\noindent Max-Planck-Institut f\"ur Mathematik, Vivatsgasse 7, Bonn 53111, Germany 

and 

\noindent International Centre for Theoretical Physics, Strada Costiera 11, Trieste 34014, Italy

\noindent dbz@mpim-bonn.mpg.de

\end{document}